\RequirePackage{fix-cm}
\documentclass{article}
\usepackage{iclr2027_conference,times}
\usepackage{amsmath,amssymb,amsthm}
\usepackage{booktabs}
\usepackage{capt-of}
\usepackage{graphicx}
\usepackage{microtype}
\usepackage{tabularx}
\usepackage{placeins}
\usepackage{needspace}
\usepackage{xspace}
\usepackage{hyperref}
\usepackage{xurl}
\hypersetup{hidelinks,pdftitle={AS2D: Accelerating On-Demand Audio Understanding on Mobile Devices},pdfauthor={Yunzhe Li, Kyoungjun Park, Hongzi Zhu, Lili Qiu}}

\newlength{\cadreDefaultTextFloatSep}
\newlength{\cadreDefaultFloatSep}
\AddToHook{env/figure/begin}[cadre-main-spacing]{%
  \setlength{\abovecaptionskip}{2pt}%
  \setlength{\belowcaptionskip}{0pt}}
\AddToHook{env/table/begin}[cadre-main-spacing]{%
  \setlength{\abovecaptionskip}{0pt}%
  \setlength{\belowcaptionskip}{5pt}}

\theoremstyle{definition}
\newtheorem{proposition}{Proposition}

\newcommand{\system}{\textsc{As}\textsuperscript{2}\textsc{d}\xspace}

\providecommand{\Description}[1]{}
\title{\system: Accelerating On-Demand\\
Audio Understanding on Mobile Devices}
\author{Yunzhe Li$^{1}$ \quad Kyoungjun Park$^{1}$ \quad Hongzi Zhu$^{2}$ \quad Lili Qiu$^{1}$\\
\normalfont $^{1}$The University of Texas at Austin \quad $^{2}$Shanghai Jiao Tong University\\
\normalfont\texttt{yunzhe.li.cs@gmail.com}}
\iclrfinalcopy

\begin{document}
\raggedbottom
\maketitle
\lhead{Preprint}
\suppressfloats
\begin{abstract}
Speculative decoding accelerates autoregressive generation by using a smaller drafter to propose tokens for batched verification by a larger target. However, conventional speculative decoding couples drafting to the target's evolving verified prefix, serializing drafting and verification. We ask whether this dependency is necessary for source-conditioned generation. Our key observation is that, for audio language models, the input audio and user request can provide useful speculative candidates without following the target's evolving text prefix. Based on this observation, we propose \system (Audio Speculative Speculative Decoding), which enables target-decoupled drafting: an audio-conditioned drafter follows its own generation history while the target independently verifies and corrects ready candidates. Without usable candidates, the target advances alone. Thus, target feedback determines which candidates are committed but no longer determines when the drafter can make progress, enabling drafting and verification to proceed concurrently while retaining target-side verification and correction.
We implement \system in MNN for Android and evaluate two target models across four phones, seven datasets, and three tasks covering 12.2 hours of audio. Across four phones, \system improves pooled ASR throughput by 42--76\% over target-only decoding, while only 5.7\% of evaluation windows are slower than target-only, compared with 58.1--63.0\% for speculative baselines. For ASR, \system reaches 97.33--98.20\% of a hindsight per-window oracle's pooled throughput over the evaluated drafter/budget catalog. Native on-demand execution with a 7B target achieves up to 78\% higher throughput than target-only. These results show that source-conditioned audio generation can relax the conventional dependence of speculative drafting on the target's evolving output prefix, exposing substantial parallelism for efficient inference.
\end{abstract}

\section{Introduction}
\label{sec:intro}

Autoregressive generation is inherently sequential: each output token
depends on the preceding output prefix, making decoding a major bottleneck
for language-model inference. Speculative decoding reduces this bottleneck
by using a smaller drafter to propose multiple tokens that a larger target
verifies in parallel, substantially accelerating generation while preserving
the target distribution \citep{leviathan2023speculative}.
However, conventional speculative decoding retains another sequential
dependency: after each verification round, the drafter conditions on the
target's accepted or corrected prefix before preparing its next proposal.
Drafting and verification therefore remain coupled through the target's
evolving autoregressive trajectory.

We ask whether this dependency is always necessary. We study this question
in audio language models, where the input audio itself provides a rich
source of information about future output tokens. One might instead
transcribe the audio and apply a text language model, but transcription
can discard information carried by the original signal, including
intonation, stress, and emotional cues
\citep{wang2025speechprosody}. Modern audio language models therefore
operate directly on audio to support tasks such as transcription,
translation, and question answering \citep{xu2025qwen25omni}.
Because these outputs are grounded in a shared audio source, a drafter may
be able to anticipate useful future tokens from the audio and instruction
without waiting for the target's latest verified prefix.

We investigate this opportunity in \emph{on-demand audio understanding},
where users request transcripts, translations, or answers about previously
received audio, supporting voice assistants and wearable memory aids
\citep{zulfikar2024memoro,openbmb2025minicpmo26,hegde2026larag}.
In this setting, audio can be prefilled before the request
\citep{openbmb2025minicpmo26,sun2026omnimem}, removing much of the audio
processing from the request-time critical path. Response generation,
however, remains autoregressive: longer replies require more sequential
decoding steps (Figure~\ref{fig:causal-streaming-scenario}(a--b)).
The problem is particularly acute for local execution on phones, where
model weights, retained audio state, and runtime buffers must fit mobile
memory while generation operates within limited compute
\citep{he2019streaming,orhon2025whisperkit}.
We seek lower request-time latency and resource cost while preserving
target outputs.

\begin{figure}[!t]
  \centering
  \setlength{\parskip}{0pt}
  \setlength{\abovecaptionskip}{-4pt}
  \includegraphics[width=\linewidth,trim=0 0bp 0 4bp,clip]{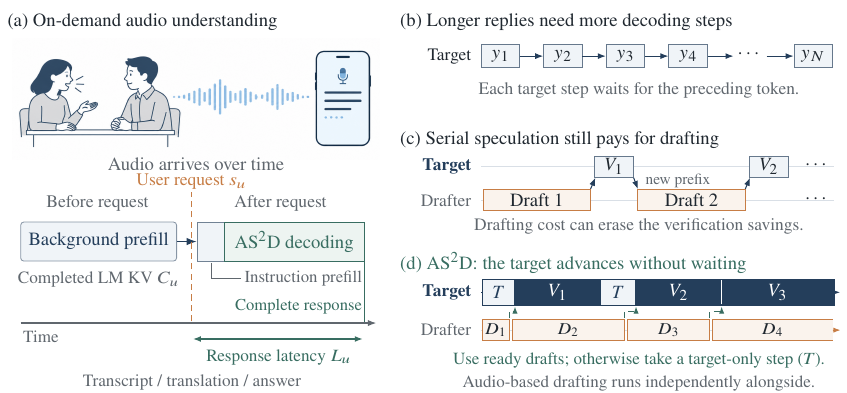}
  \Description{Panel a separates background audio prefill and completed language-model
  KV from request-time instruction prefill and decoding. A request boundary starts the
  response-latency clock, ending at the complete response. Panel b shows sequential target-token generation. Panel c shows
  serial drafting and target verification, with a new verified prefix needed
  before the next draft. Panel d makes target execution the primary continuous
  lane. The target verifies ready candidates or advances alone if no usable
  candidate is ready. Target-decoupled drafting runs independently; dashed arrows
  offer ready candidates at target boundaries without blocking the target.
  All timelines are schematic.}
  \par\nointerlineskip
  \caption{\textbf{\system overlaps drafting and verification for on-demand audio understanding.}
  (a) Background audio prefill prepares LM KV before the request. Response
  latency includes instruction prefill, decoding, and any residual audio
  preparation. (b) Longer replies require more sequential target steps.
  (c) Batched verification saves target work, but serial drafting can consume
  those savings on a phone. (d) \system keeps the target advancing: it verifies
  ready candidates (\(V_i\)) or takes a target-only step (\(T\)) without
  waiting for drafting. Audio-based drafting (\(D_i\)) runs alongside,
  independent of the latest verified prefix. Dashed arrows supply ready
  candidates at target boundaries. The target determines committed tokens.
  Decoding timelines are schematic and start after the request.}
  \label{fig:causal-streaming-scenario}
\end{figure}

Conventional speculative decoding does not fully address this challenge.
On phones, serial drafting can consume much of the savings from batched
verification (Figure~\ref{fig:causal-streaming-scenario}(c)), as illustrated
by WhisperKit's omission of speculation from its final Large v3 Turbo
configuration \citep{orhon2025whisperkit}.
Asynchronous methods overlap drafting and verification by parallelizing
target execution, as in Distributed Speculative Inference (DSI), or
preparing continuations for predicted verification outcomes, as in
Speculative Speculative Decoding (SSD)
\citep{timor2025dsi,kumar2026ssd}.
These methods are studied in multi-GPU server settings. On a single phone,
their additional computation and speculative state compete with verification
for limited resources on heterogeneous processors.

Our key observation is that, for source-conditioned audio generation,
drafting need not always follow the target's evolving text prefix.
Audio and the instruction can directly supply useful candidates while
the target continues along its own autoregressive trajectory.
A controlled ASR study supports this premise: removing preceding
reference-transcript context leaves draft acceptance unchanged in
\textbf{72.3\%} of 480 paired proposal blocks, with a median acceptance
change of zero.

To exploit this property, we propose \system (Audio Speculative Speculative
Decoding), built on \emph{target-decoupled drafting} to remove this verifier-to-drafter synchronization dependency for source-conditioned audio generation.
An autoregressive drafter follows its own generation history without
synchronizing its state to the target's accepted or corrected prefix.
Meanwhile, the target advances along the authoritative output trajectory,
verifying aligned candidates when available and decoding alone when none
is usable. Thus, target feedback determines which candidates are committed,
but no longer determines when the drafter can make progress.
This decoupling lets candidate generation and target verification proceed
concurrently across the phone's heterogeneous processors
(Figure~\ref{fig:causal-streaming-scenario}(d)).
No separate continuations are prepared for alternative verification
outcomes, and the target retains its context and determines the committed
output.

We implement \system in MNN for Android \citep{jiang2020mnn} and evaluate
\textbf{two target LLMs}, Qwen3-ASR-1.7B \citep{shi2026qwen3asr} and
Qwen2.5-Omni-7B \citep{xu2025qwen25omni}.
The study spans \textbf{four phones, seven datasets, and three task types}:
speech recognition
\citep{panayotov2015librispeech,carletta2006ami,conneau2022fleurs,tang2021kespeech,zhang2022m4singer},
translation \citep{wang2021covost2}, and spoken question answering
\citep{lee2018spokensquad}.
The corpus evaluation includes \textbf{928 task requests} over 913 distinct
audio inputs of 1.6--75.2s, totaling \textbf{12.2 hours of audio and
97,953 output tokens}.
Across four phones, \system improves pooled ASR throughput by
\textbf{42--76\%} over target-only decoding.
Only \textbf{5.7\%} of its evaluation windows are slower than target-only,
versus 58.1--63.0\% for the speculative baselines
(Table~\ref{tab:throughput-coverage}), and improvements over the strongest
speculative baseline reach \textbf{114.6\%}
(Table~\ref{tab:representative-cases}).
\system reaches \textbf{97.33--98.20\%} of a hindsight per-window
oracle's pooled TPS over the evaluated drafter/budget catalog.
Our native \textbf{7B deployment} processes a \textbf{10-minute audio stream}
in real time (Appendix~\ref{sec:native-sustained-streaming}) and delivers
up to \textbf{78\% higher TPS} than target-only for on-demand audio
understanding (Figure~\ref{fig:native-ondemand-fourway}).

Our contributions are threefold:
\textbf{(1)} We identify a target-prefix synchronization dependency in
conventional speculative decoding and show that source-conditioned audio
generation can relax this dependency: useful candidates can often be
generated without observing the target's evolving verified prefix.
\textbf{(2)} We introduce target-decoupled drafting and design \system,
which allows the drafter and target to advance concurrently while retaining
target-side verification and correction.
\textbf{(3)} We implement \system on Android and evaluate it across two
target models, four phones, seven datasets, and three audio-language tasks,
demonstrating substantial and robust decoding acceleration in both
phone-profile replay and native execution.

\section{Related work}
\label{sec:main-related-work}

\noindent\textbf{On-demand audio understanding.}\quad
MiniCPM-o 2.6 supports continuous audio input independently of user queries
and separates language-model prefill from response generation
\citep{openbmb2025minicpmo26}. OmniMem compresses cached audio-visual memory
to retain context for later queries \citep{sun2026omnimem}. These methods
retain prepared source context but leave sequential decoding on the response
path. We complement them by accelerating responses under mobile resource constraints.

\noindent\textbf{Speculative speech decoding.}\quad
WhisperKit omits speculation from its final Large v3 Turbo configuration
because of drafter overhead \citep{orhon2025whisperkit}. SpecASR exploits
acoustic alignment through adaptive draft lengths, draft recycling, and sparse
trees, while regenerating drafts from verified text \citep{wei2025specasr}.
We build on its local candidate reuse, but let the audio-conditioned producer
advance independently of target-prefix updates. Whisper-Medusa adds prediction
heads, and CTC encoder drafts provide acoustic candidates. Their acceptance
rules can change the target's greedy output
\citep{segalfeldman2025whispermedusa,saon2026selfspec}. \system retains target
verification and correction across recognition, translation, and question
answering.

\noindent\textbf{Asynchronous and mobile execution.}\quad
PEARL uses pre- and post-verification to overlap drafting and verification
\citep{liu2025pearl}. DSI orchestrates parallel target and drafter instances
\citep{timor2025dsi}, while SSD caches continuations for predicted verification
outcomes on separate server GPUs \citep{kumar2026ssd}. These methods retain
dependencies on verified or hypothesized target-text prefixes. AHASD addresses mobile
asynchronous speculation through a custom NPU--PIM architecture, using target
feedback to confirm or roll back drafts \citep{ma2026ahasd}.
On a single phone, extra computation and speculative state compete with
verification for limited resources. Assigning verification to the processor
best suited to it can leave prefix maintenance on a slower processor.
\system instead uses audio and the instruction to supply target-decoupled
candidates on available phone processors, without preparing continuations for
hypothetical target outcomes. It still verifies, corrects, and aligns
candidates. Appendix~\ref{sec:related-work} covers additional execution and
configuration methods.

\section{Problem formulation}
\label{sec:main-problem}

\noindent\textbf{Notation.}\quad
Let \(x\) denote the incoming audio stream and \(T\) the target model.
Request \(u\) arrives at time \(s_u\) with instruction \(p_u\) for an audio
view \(x_u\) selected from \(x_{\leq s_u}\) by protocol \(\mathcal S\).
Its input \(z_u=(x_u,p_u)\) stays fixed during generation. \(C_u\) is the completed target language-model KV state
prepared from \(x_u\). \(\mathcal S\) fixes audio preparation,
prompt construction, state handling, and termination.

A decoding scheme \(\mathcal A\) specifies candidate generation,
verification, and scheduling under fixed deployment settings.
Let \(Y_{\mathcal A}(z_u;\mathcal S)\) and \(Y_T(z_u;\mathcal S)\) denote
its output and the greedy-target reference. For completion time
\(f_u^{\mathcal A}\), response latency is
\(L_u^{\mathcal A}=f_u^{\mathcal A}-s_u\), including all work after request arrival.
Let \(\mathfrak A\) be the set of schemes respecting input availability,
candidate readiness, task/backend support, output limits, and
accepted-prefix state restoration.

Let \(m_{\mathcal A}(v)\) denote resident model, cache, and buffer memory at
wall-clock time \(v\), with memory budget \(M_{\max}\).
\(Q_\alpha\) denotes the \(\alpha\)-quantile of request latency, with
optional limit \(L_{\max}\).

\noindent\textbf{On-demand audio language model inference problem.}\quad
For a fixed phone and request distribution, we choose
\(\mathcal A\in\mathfrak A\) to minimize mean response latency
\(\mathbb E[L_u^{\mathcal A}]\) while preserving the target's output.
The scheme must satisfy \(m_{\mathcal A}(v)\leq M_{\max}\) at all times
and, when required, the tail-latency constraint
\(Q_\alpha(L_u^{\mathcal A})\leq L_{\max}\).
Post-prefill throughput isolates decoding.
Appendix~\ref{sec:problem-formulation} details the protocol.

\section{Design of \system}
\label{sec:main-method}

\label{sec:method-overview}
\begin{figure}[!t]
  \centering
  \setlength{\parskip}{0pt}
  \setlength{\abovecaptionskip}{-7pt}
  \includegraphics[width=\textwidth,trim=0 1bp 0 2bp,clip]{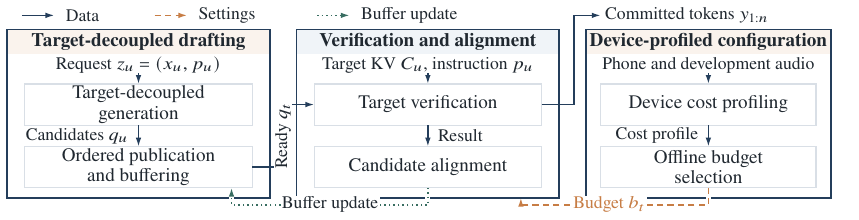}
  \Description{Target-decoupled drafting generates and buffers candidates
  directly from request input z_u, comprising a fixed audio view and instruction.
  Verification and alignment uses target language-model KV C_u already
  retained in memory to process the request instruction p_u, then verifies ready candidates or
  decodes alone. It emits committed tokens and updates the candidate cursor.
  The producer has separate model state and receives no target-text feedback.
  The right module profiles costs and selects the budget offline for a fixed drafter.
  The left and middle modules run concurrently. Solid arrows carry data,
  dashed arrows supply settings, and dotted arrows update the buffer.}
  \par\nointerlineskip
  \caption{\textbf{\system decouples candidate generation from target verification.}
  Request \(z_u=(x_u,p_u)\) supplies fixed audio and an instruction to
  target-decoupled drafting. Verification and alignment processes \(p_u\)
  using retained target KV \(C_u\). The target verifies ready candidates or
  decodes alone while the producer continues independently. Alignment updates
  the buffer. Device-profiled configuration combines measured phone costs
  with development traces to select a fixed budget for the given drafter.}
  \label{fig:overview}
\end{figure}

\begin{figure}[!b]
  \centering
  \begin{minipage}[t]{0.48\linewidth}
    \centering
    \includegraphics[width=\linewidth]{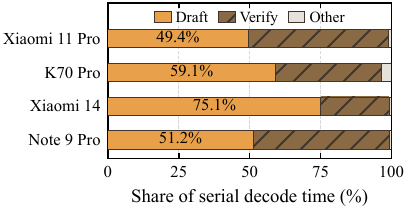}
    \par\small (a) Native serial decoding cost
  \end{minipage}\hfill
  \begin{minipage}[t]{0.48\linewidth}
    \centering
    \includegraphics[width=\linewidth]{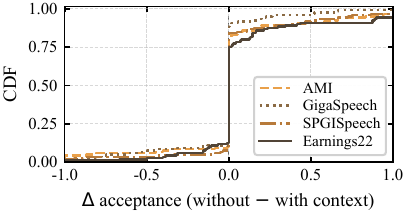}
    \par\small (b) Sensitivity to text context
  \end{minipage}
  \caption{\textbf{Drafting is costly, yet need not wait for updated text.}
  (a) Native serial SD pairs a CPU 0.6B drafter with an OpenCL 1.7B target.
  Time shares pool five runs per phone on one 59.5-second audio,
  with seven candidates per round. Drafting includes prefix catch-up.
  Loading and initial prefill are excluded.
  (b) We compare acceptance with and without preceding transcript context.
  The comparison uses 120 four-second blocks per corpus, a 32-token proposal
  cap, and identical verification context. Both panels use Qwen3-ASR.
  Appendix~\ref{sec:drafting-diagnostics} gives both protocols.}
  \label{fig:prefix-context-cdf}
  \Description{Left: drafting accounts for 49.4 to 75.1 percent of native
  serial decoding time on four phones in one common audio-window study.
  Right: four empirical CDFs show acceptance-ratio differences when
  preceding transcript context is removed from the drafter. All 480 pairs
  are retained; 72.3 percent have zero change. Negative and positive tails
  show that missing context can also reduce or increase acceptance.}
\end{figure}

Our core idea is to generate audio-conditioned candidates independently of
target-prefix updates, allowing drafting and target verification to run
concurrently (Figure~\ref{fig:overview}).
At request time, \emph{Target-decoupled drafting} publishes
candidates to a buffer, and \emph{Verification and alignment} consumes
ready blocks to produce committed tokens. Alignment feedback updates only
the buffer cursor, allowing the producer to continue independently.
If no candidate is usable, the target decodes alone.
\emph{Device-profiled configuration} supplies the budget chosen before execution.

In \emph{Target-decoupled drafting}, \emph{Target-decoupled generation}
maps the request input \(z_u=(x_u,p_u)\) to candidate tokens \(q_u\)
using the fixed task-specific drafter.
\emph{Ordered publication and buffering} supplies a ready block \(q_t\)
to the target at decoding boundary \(t\).
\emph{Verification and alignment} begins with \emph{Target verification},
which uses the existing audio KV \(C_u\), instruction \(p_u\), and committed
context to check ready candidates and emit tokens
\(y_{1:n}\). Its result feeds \emph{Candidate alignment}, which advances
or repairs the buffer cursor to expose the next usable block.
For \emph{Device-profiled configuration}, \emph{Device cost profiling}
measures phone execution costs.
\emph{Offline budget selection} combines these costs with development
audio to select the budget cap \(b_t\) supplied to target verification.

\subsection{Target-decoupled drafting}
\label{sec:method-drafting}
\noindent\textbf{Target-decoupled generation.}\quad
Serial drafting adds substantial work before each verification.
On one common audio window across four phones, it accounts for
49.4--75.1\% of native serial SD decoding time
(Figure~\ref{fig:prefix-context-cdf}(a)). To avoid serializing drafting
with verification, we seek proposals that do not depend on its outcome.
Our key observation is that the input audio already carries
information about upcoming output tokens, allowing a drafter to make
progress without the latest target text.
In a controlled ASR probe, omitting preceding reference-transcript context leaves
draft acceptance unchanged in 72.3\% of 480 paired blocks, with a median
change of zero (Figure~\ref{fig:prefix-context-cdf}(b)). Motivated by this
observation, we generate candidates from the request audio and instruction
independently of target-prefix updates.
After the instruction is known, the fixed drafter is dispatched at
\(\tau_u\geq s_u\) and generates
\(q_u=\operatorname{Tok}_T\!\left(G_{\phi}(x_u,p_u)\right)\),
where \(\operatorname{Tok}_T\) converts proposals into the target vocabulary.
An autoregressive drafter conditions on its own generated text, without
reading target-text updates or target KV. It can therefore continue while
the target verifies earlier candidates.

\phantomsection
\begin{figure}[!t]
  \centering
  \includegraphics[width=\linewidth,trim=0 1bp 0 2bp,clip]{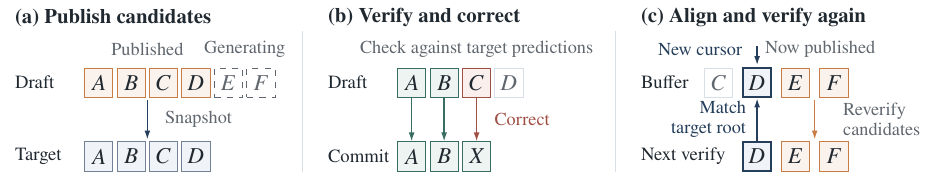}
  \caption{\textbf{\system publishes, verifies, and reuses partial drafts.}
  (a) Only published candidates enter verification.
  (b) The target accepts \(A\,B\), corrects \(C\) to \(X\), and commits \(A\,B\,X\).
  (c) The target's next token \(D\) matches the buffer and supplies the root
  for verifying ready \(E\,F\). The producer continues unchanged.
  The target verifies each reused candidate.
  Tokens and batch lengths are illustrative.}
  \label{fig:token-flow}
  \Description{Three schematic panels follow one draft. The producer publishes
  A B C D while generating E F. Verification accepts A B and supplies correction
  X, leaving the draft D uncommitted. The target then supplies its own next
  token D, which matches a position in the immutable candidate buffer. The
  read cursor moves there. The next verification uses target D as the known
  root and the now-published E F as unaccepted candidates. The producer never
  incorporates the target correction and continues its original sequence.}
\end{figure}

\noindent\textbf{Ordered publication and buffering.}\quad
\label{sec:method-queue}
The target needs partial drafts promptly, but its corrections must not
interrupt ongoing generation. We connect the producer and target through
an ordered token buffer. The producer publishes token batches as it
generates them, leaving published tokens unchanged. In
Figure~\ref{fig:token-flow}(a),
the target can check a published sequence \(A\,B\,C\,D\) while the producer
generates \(E\,F\,\ldots\). Newly published tokens are available to later
verification calls, without changing the candidates already being checked.
If the target accepts \(A\,B\) but replaces \(C\) with \(X\), the producer
continues its own sequence. A separate \emph{read cursor} tracks where to
look for the next usable candidates. Candidate alignment updates this
cursor without changing the producer's text or model state. If no usable
continuation is ready, the target decodes alone.
Appendix~\ref{sec:appendix-drafting} gives the publication details.

\subsection{Verification and alignment}
\label{sec:method-verification}
\phantomsection
\noindent\textbf{Target verification.}\quad
\label{sec:method-target-verification}
After processing instruction \(p_u\) using the retained audio KV \(C_u\),
the target checks up to \(b_t\) ready candidates under the committed prefix
using standard greedy speculative verification
\citep{leviathan2023speculative}. In Figure~\ref{fig:token-flow}(b), it checks
\(A\,B\,C\,D\), agrees on \(A\,B\), but predicts \(X\) instead of \(C\).
It therefore commits \(A\,B\,X\). The alignment step then determines
whether the uncommitted candidates can be reused. If no usable
candidate is ready, the target decodes alone without waiting for the
producer. Appendix~\ref{sec:tree-contract} gives the verification and KV
repair details.

\phantomsection
\noindent\textbf{Candidate alignment.}\quad
\label{sec:method-alignment}
A rejected token need not invalidate the remaining draft. We use local
token matching to recover reusable continuations \citep{wei2025specasr}.
Figure~\ref{fig:token-flow}(c) continues the example: the target next predicts \(D\) after
committing \(A\,B\,X\). We locate \(D\) near the current buffer position
and move the read cursor to that match. The following \(E\,F\,\ldots\)
can then be verified under the updated target context as they become
available. Full acceptance simply advances the cursor. If no match is
found, the target continues decoding alone. Throughout this process, the
producer keeps generating its own sequence.
Appendix~\ref{sec:candidate-recycling} gives the bounded search rule.

\subsection{Device-profiled configuration}
\label{sec:main-joint-control}
\noindent\textbf{Device cost profiling.}\quad
Batching has a non-smooth per-token cost. Larger batches can be less efficient,
with different penalties across phones. We profile target-reference token
blocks on each phone and backend to isolate hardware efficiency before
budget selection (Appendix~\ref{sec:verification-width-profile}). We combine
measured costs and concurrent slowdown with development traces to estimate
decoding time, including correction, alignment, and producer drain.

\noindent\textbf{Offline budget selection.}\quad
Each task/model pairing uses a fixed drafter. Device profiles and development
traces guide the budget choice within memory and backend constraints.
This budget remains fixed throughout
deployment; candidate readiness can still reduce the actual verification width.
Appendix~\ref{sec:appendix-eval-deployment} details development splits, retained
configurations, calibration, and time accounting.

\subsection{When Does Target-Decoupled Drafting Help?}
\label{sec:when-help}

Target decoupling can hide drafting behind target execution, reducing the
idealized cost from $D+V$ to $\max(D,V)+H$, where $D$, $V$, and $H$ denote
drafting, verification/decoding, and additional overhead, respectively.
Drafting accounts for 49.4--75.1\% of
native serial SD time.

Speedup depends on three factors: (1) drafting cost, or how much work can be hidden; (2) candidate readiness, or whether candidates arrive before they are needed; and (3) useful candidate coverage, or how much target output they cover. Audio provides a favorable setting: removing preceding transcript context leaves acceptance unchanged in 72.3\% of our ASR probes. High acceptance alone, however, is insufficient: spoken QA achieves 84.7\% acceptance but only 28.3\% candidate coverage, yielding an 8.0\% throughput improvement over the strongest speculative baseline. Thus, AS$^2$D benefits most when drafting is costly and independent candidates are timely and sufficiently cover target output.

\section{Evaluation}
\label{sec:main-evaluation}

\textbf{RQ1:} How much does \system
accelerate decoding across phones, models, and audio tasks?
\textbf{RQ2:} Which components drive the gains?
\textbf{RQ3:} Does \system accelerate native on-demand requests?

\subsection{Experimental setup}
\label{sec:main-evaluation-setup}

\noindent\textbf{Implementation.}\quad
We implement \system in MNN for Android \citep{jiang2020mnn,mnn2026},
with resident drafter/target models and separate KV caches.
Redmi K70 Pro, Xiaomi 11 Pro, Xiaomi 14, and Redmi Note 9 Pro span
Snapdragon 750G, 888, and 8 Gen 3 platforms
\citep{qualcomm2026sd750g,qualcomm2026sd888,xiaomi2026mi14specs},
with heterogeneous CPUs, Adreno GPUs, and 7.1--14.8\,GiB of OS-visible RAM.
Our targets are Qwen3-ASR-1.7B and Qwen2.5-Omni-7B
\citep{shi2026qwen3asr,xu2025qwen25omni}.
ASR pairs an OpenCL target with Qwen3-ASR-Audio-0.6B on four CPU threads.
All main translation and spoken-QA comparisons pair Omni-7B with
Qwen2.5-Omni-3B as the drafter, sharing the target tokenizer.
See Appendix~\ref{sec:appendix-eval-setup}.

\noindent\textbf{Datasets and tasks.}\quad
Our evaluation covers seven datasets and three task families, totaling 12.2 hours
of distinct audio and 97,953 target output tokens, excluding EOS.
\emph{Speech recognition} covers five settings: read-speech transcription
with LibriSpeech (Libri) \citep{panayotov2015librispeech}, meeting transcription
with AMI SDM \citep{carletta2006ami}, multilingual recognition with six
languages from FLEURS \citep{conneau2022fleurs}, dialect recognition with
KeSpeech \citep{tang2021kespeech}, and singing transcription with
M4Singer \citep{zhang2022m4singer}.
\emph{Speech translation} uses CoVoST2 \citep{wang2021covost2} for
English-to-Chinese and Chinese-to-English translation.
\emph{Spoken question answering} uses Spoken SQuAD \citep{lee2018spokensquad}.
Audio shared by multiple QA questions is counted once in the total duration.
See Appendix~\ref{sec:appendix-eval-setup}.

\noindent\textbf{Baselines.}\quad
\emph{Target-only} generates one token at a time using the target model alone.
\emph{Standard speculative decoding (SD)} uses a smaller, prefix-conditioned drafter to propose candidates,
then verifies them in a batch with the target \citep{leviathan2023speculative}.
\emph{SpecASR (ASP)} adapts draft length using confidence and recycles
rejected candidates \citep{wei2025specasr}.
See Appendix~\ref{sec:baseline-implementations}.

\noindent\textbf{Metric.}\quad
We report tokens per second (TPS), the number of committed output tokens
divided by elapsed time. Audio prefill finishes before request arrival.
TPS measures full-response generation, whereas time to first token (TTFT)
captures only the wait for the first token.

\Needspace{8\baselineskip}
\subsection{RQ1: Performance across phones, models, and tasks}
\label{sec:main-evaluation-results}

We compare the four methods on the same 688 ASR windows from five datasets
using phone-cost replay. Replay and native TPS agree closely on average
in matched timing checks (Appendix~\ref{sec:replay-agreement}).
\system deploys a fixed Qwen drafter with budget 7 on K70, Xiaomi 11 Pro,
and Xiaomi 14, and 3 on Note 9 Pro, selected without test outcomes.
We then evaluate Omni-7B translation and spoken QA.

\begin{figure}[!t]
  \centering
  \includegraphics[width=\linewidth]{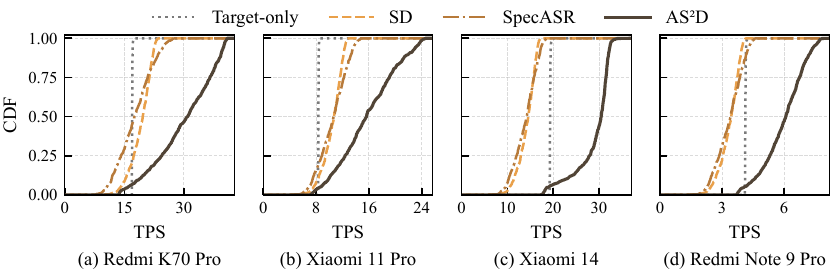}
  \caption{\textbf{\system improves decoding throughput across four phones.}
  Each panel contains the same 688 complete-audio test windows from five
  datasets, with separately measured phone-cost replay. Curves are unsmoothed.
  Farther right indicates higher TPS. SpecASR uses the serial ASP + recycling
  port. SD averages K=5,6,7,8 TPS within each window.
  Horizontal scales differ by phone.}
  \label{fig:fourphone-tps-cdf}
\end{figure}

\noindent\textbf{Across-phone performance.}\quad
Figure~\ref{fig:fourphone-tps-cdf} plots the per-window TPS distributions
for Target-only, SD, SpecASR, and \system on each phone.
\system improves pooled throughput by 42--76\% over target-only.
On K70, SD averages 18.97 TPS versus target-only's 16.96, but its
21.52 TPS on LibriSpeech falls to 15.46 on AMI, below target-only's 16.92.
Serial speculation therefore does not uniformly help. With fixed Qwen/budget 7,
\system reaches 27.92 TPS, improving throughput by 64.6\% over target-only
and 47.2\% over SD.
See Appendix~\ref{sec:appendix-evaluation}.

\phantomsection
\label{sec:main-case-studies}
\begin{table}[!t]
\centering
\begin{minipage}[t]{0.50\linewidth}
\vspace{0pt}
\centering
\fontsize{9}{10.5}\selectfont
\caption{\textbf{\system improves TPS across representative cases.}
Gains use the faster of SD and SpecASR (ASP) in Xiaomi 14 phone-cost replay.}
\label{tab:representative-cases}
\begingroup
\fontsize{9}{10.5}\selectfont
\setlength{\tabcolsep}{1.5pt}
\renewcommand{\arraystretch}{1}
\begin{tabular*}{\linewidth}{@{\extracolsep{\fill}}lrrrrr@{}}
\toprule
Case & Target & SD & ASP & \system & Gain \\
\midrule
Read & 19.26 & 16.65 & 16.88 & \textbf{32.24} & +91.0\% \\
Multi. & 19.24 & 15.58 & 16.02 & \textbf{32.48} & +102.7\% \\
Dialect & 19.37 & 14.19 & 14.61 & \textbf{31.34} & +114.6\% \\
Singing & 19.31 & 13.13 & 14.44 & \textbf{28.94} & +100.4\% \\
Meeting & 19.25 & 11.36 & 8.93 & \textbf{22.86} & +101.3\% \\
\bottomrule
\end{tabular*}
\endgroup

\end{minipage}\hfill
  \begin{minipage}[t]{0.48\linewidth}
    \vspace{0pt}\centering
    \setlength{\abovecaptionskip}{0pt}
    \includegraphics[width=\linewidth]{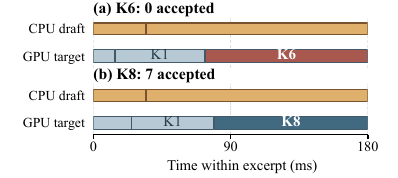}
    \captionof{figure}{\textbf{Qwen CPU drafting overlaps GPU verification on Xiaomi 14 ($b=7$).}}
    \label{fig:native-scheduling-timelines}
  \end{minipage}

\end{table}

\noindent\textbf{Workload and overlap analysis.}\quad
We examine representative workloads and native scheduling traces.
Table~\ref{tab:representative-cases}
selects a completed Xiaomi 14 window nearest each dataset's median gain
over the faster serial baseline.
SpecASR spends 73.1--76.1\% of decode time on drafting; \system gains
91.0--114.6\% over the faster serial baseline. The meeting case gains
18.7\% over target-only: 56 of 144 proposals are accepted, leaving
164 verifier rounds.

Figure~\ref{fig:native-scheduling-timelines} complements these cases with
two excerpts from the median-duration run of five native Qwen repeats on a
Xiaomi 14 LibriSpeech calibration window: the first multi-candidate rejection
and the first fully accepted K8 call. CPU drafting continues while the GPU target decodes
or verifies ready candidates. With proposal budget 7, the target uses K6
when five candidates are ready and K8 when at least seven are ready.
The accepted K8 call commits seven draft tokens in one target pass,
while the rejected K6 call illustrates target correction of mismatched
candidates. Appendix~\ref{sec:appendix-eval-cases} gives case selection and full-run accounting.

\suppressfloats[t]
\begin{table}[t]
\centering
\setlength{\abovecaptionskip}{0pt}
\setlength{\belowcaptionskip}{5pt}
\begin{minipage}[t]{0.355\linewidth}
\vspace{0pt}
\centering
\normalsize
\setlength{\abovecaptionskip}{0pt}
\setlength{\belowcaptionskip}{5pt}
\parbox[t][49pt][t]{\linewidth}{%
\raggedright
\captionof{table}{\textbf{\system approaches the per-window oracle.}
TPS shortfall (\%) from the hindsight drafter/budget oracle;
688 windows/phone.}
\label{tab:joint-oracle-gap}}\par\nointerlineskip
\begingroup
\fontsize{9}{10.5}\selectfont
\setlength{\tabcolsep}{1pt}
\renewcommand{\arraystretch}{1}
\begin{tabular}{lrrrr}
\toprule
Dataset & K70 & 11 Pro & Mi 14 & Note 9 \\
\midrule
LibriSpeech & 2.48 & 2.09 & 2.46 & 2.67 \\
AMI & 3.56 & 5.83 & 5.93 & 2.94 \\
FLEURS & 1.33 & 1.46 & 0.65 & 1.31 \\
KeSpeech & 1.56 & 1.81 & 0.93 & 0.71 \\
M4Singer & 2.03 & 2.76 & 1.76 & 1.25 \\
\bottomrule
\end{tabular}
\endgroup

\end{minipage}\hfill
\begin{minipage}[t]{0.375\linewidth}
\vspace{0pt}
\centering
\normalsize
\setlength{\abovecaptionskip}{0pt}
\setlength{\belowcaptionskip}{5pt}
\parbox[t][49pt][t]{\linewidth}{%
\raggedright
\captionof{table}{\textbf{\system accelerates translation and QA.}
TPS uses K70 Pro phone-cost replay on 80 samples per task.
Gains use the fastest baseline.}
\label{tab:omni-task-results}}\par\nointerlineskip
\begingroup
\fontsize{9}{10.5}\selectfont
\setlength{\tabcolsep}{1.5pt}
\renewcommand{\arraystretch}{1}
\begin{tabular}{lrrr}
\toprule
Method & EN$\to$ZH & ZH$\to$EN & QA$_{\geq16}$ \\
\midrule
Target-only & \underline{8.11} & \underline{8.37} & 7.77 \\
SD & 5.21 & 6.54 & \underline{8.90} \\
SpecASR & 5.99 & 7.15 & 8.57 \\
\system & \textbf{9.52} & \textbf{10.98} & \textbf{9.62} \\
Gain (\%) & +17.4 & +31.2 & +8.0 \\
\bottomrule
\end{tabular}
\endgroup

\end{minipage}\hfill
\begin{minipage}[t]{0.25\linewidth}
\vspace{0pt}
\centering\normalsize
\parbox[t][49pt][t]{\linewidth}{%
\caption{\raggedright\textbf{\system reduces slowdowns.}
Windows below target-only TPS (\%, 688/phone).}
\label{tab:throughput-coverage}}\par\nointerlineskip
\begingroup
\fontsize{9}{10.5}\selectfont
\setlength{\tabcolsep}{0.5pt}
\renewcommand{\arraystretch}{1}
\begin{tabular*}{\linewidth}{@{\extracolsep{\fill}}lrrr@{}}
\toprule
Phone & SD & ASP & \system \\
\midrule
K70 & 21.4 & 45.5 & 6.8 \\
11 Pro & 11.2 & 19.6 & 4.5 \\
Mi 14 & 100.0 & 100.0 & 6.0 \\
Note 9 & 99.7 & 86.9 & 5.4 \\
All & 58.1 & 63.0 & 5.7 \\
\bottomrule
\end{tabular*}
\endgroup

\end{minipage}
\end{table}

\Needspace{8\baselineskip}
\noindent\textbf{Distance to the per-window oracle.}\quad
We compare the deployment with a hindsight oracle on the same
688 windows per phone.
The oracle chooses one feasible drafter and fixed budget per window
from the evaluated catalog, with zero selection cost.
Table~\ref{tab:joint-oracle-gap} reports pooled TPS shortfall by
phone and dataset. \system reaches 97.33--98.20\% of the oracle's overall
pooled TPS; AMI leaves the most headroom.
Appendix~\ref{sec:joint-oracle-gap-protocol} gives the protocol and
cost sensitivity.

\phantomsection
\label{sec:main-omni-tasks}

\noindent\textbf{Impact across source-conditioned tasks.}
We examine \system across source-conditioned tasks with different candidate-generation characteristics. While ASR closely constrains the output from audio, translation and spoken QA permit more varied responses. We compare the four methods on CoVoST2
\citep{wang2021covost2} translation in both directions and Spoken SQuAD
\citep{lee2018spokensquad} QA, with 80 examples per task and shared frozen
Omni-7B target outputs. English-to-Chinese IDs were selected by audio hash;
Chinese-to-English IDs are disjoint from the budget pilot. QA is a post-hoc
subgroup with at least 16 useful target tokens. All three tasks use
the same Omni-7B target and Omni-3B drafter across speculative methods.
We estimate post-prefill TPS using Redmi K70 Pro phone-cost replay.
Source candidates share the target tokenizer and need no text conversion.

Table~\ref{tab:omni-task-results} shows TPS improvements of 17.4\%, 31.2\%,
and 8.0\% over the fastest measured baseline, respectively. Conditional
acceptance is 54.6\%, 63.8\%, and 84.7\%, while accepted-token coverage is
26.6\%, 36.7\%, and 28.3\%. Notably, QA has the highest conditional
acceptance but the smallest improvement, with only 17.5\% of verification
rounds using multiple tokens. These results reveal a key property of
target-decoupled drafting: \emph{high acceptance alone does not translate
into high speedup}. Effective acceleration requires independent candidates
to arrive in time and cover a substantial fraction of the target output,
thereby replacing sequential target steps with multi-token verification.
This observation supports the conditions in Section~\ref{sec:when-help}
and explains when \system can provide substantial gains.
Appendix~\ref{sec:appendix-evaluation} provides selection rules, budgets,
replay cost coverage, and paired audit details.

\noindent\textbf{Frequency of slowdowns.}\quad
To check whether pooled gains hide frequent regressions,
Table~\ref{tab:throughput-coverage} counts windows with lower TPS than
target-only, assigning equal weight to all 688 windows per phone.
\system is slower on 6.8\%, 4.5\%, 6.0\%, and 5.4\%
of windows on K70, Xiaomi 11 Pro, Xiaomi 14, and Redmi Note 9 Pro,
respectively. Across all phone/window pairs, this is 5.7\%, compared with
58.1\% for SD and 63.0\% for SpecASR. Thus \system improves both pooled
throughput and the frequency of gains.
See Appendix~\ref{sec:appendix-eval-cases}.

\Needspace{10\baselineskip}
\subsection{RQ2: Sources of the speedup}
\label{sec:main-ablation}
To separate the contributions of target-decoupled drafting (TD) and
budget configuration (BC), Table~\ref{tab:main-component-ablation} removes
each component while holding the target, 688-window cohort, and
post-prefill clock fixed.
Removing TD uses serial prefix-conditioned
Qwen drafting with the same frozen budget choices. Removing BC restores
ASP's confidence threshold 0.4 and proposal cap (24 on K70, 14 elsewhere)
while retaining TD. Removing both recovers SpecASR.
Appendix Table~\ref{tab:evaluation-performance} gives dataset-level results.

\noindent\textbf{Target-decoupled drafting supplies the main gain.}\quad
With the same budget-7 configuration, \system reaches 27.92 TPS on K70
versus 19.77 without TD, a 41.2\% increase. The corresponding gains are
33.4\%, 95.2\%, and 52.5\% on Xiaomi 11 Pro, Xiaomi 14, and Note 9.
Across all four phones, TD raises pooled TPS from 8.49 to 12.90 (51.9\%).
This ablation measures the joint benefit of target-prefix independence
and asynchronous overlap. Both arms use complete audio and the same initial
prefill. Removing enough exposed drafting work can improve throughput
without increasing acceptance.

\noindent\textbf{Budget configuration adds a smaller benefit.}\quad
On K70, development workloads screen proposal limits 1--24, and validation
retains budget 7. Replacing ASP's
budget rule raises serial TPS from 16.81 to 19.77 (17.6\%). With TD already
present, it raises TPS from 25.60 to 27.92 (9.0\%), with gains of
5.2--17.3\% across the five datasets. The corresponding pooled benefits on
the other phones are 6.4\%, 2.0\%, and 15.4\%.
All four phones use a fixed budget throughout test.
Per-device ablation protocols and sensitivity checks are in
Appendix~\ref{sec:appendix-eval-rq2}.

\begin{table}[!t]
\centering
\setlength{\abovecaptionskip}{0pt}
\setlength{\belowcaptionskip}{5pt}
\begin{minipage}[t]{0.73\linewidth}
\vspace{0pt}
\centering
\normalsize
\setlength{\abovecaptionskip}{2pt}
\IfFileExists{figures/evaluation/fig_native_ondemand_two_phones.pdf}{%
\includegraphics[width=\linewidth]{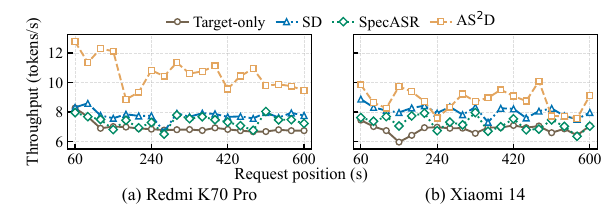}%
}{%
\PackageError{as2d-native-figure}{Measured two-phone Figure 8 is missing}%
{Run plot_fig8_two_phones.py with the completed Xiaomi 14 CSV before compiling.}%
}
\captionof{figure}{\textbf{\system accelerates native ASR requests.}
60s audio per request.}
\label{fig:native-ondemand-fourway}

\end{minipage}\hfill
\begin{minipage}[t]{0.25\linewidth}
\vspace{0pt}
\centering
\normalsize
\parbox[t][44pt][t]{\linewidth}{%
\caption{\raggedright\textbf{Removing TD lowers TPS most.}
Values show TPS after removing BC or TD.}
\label{tab:main-component-ablation}}\par\nointerlineskip
\begingroup
\fontsize{9}{10.5}\selectfont
\setlength{\tabcolsep}{0.5pt}
\renewcommand{\arraystretch}{1}
\begin{tabular*}{\linewidth}{@{\extracolsep{\fill}}lrrr@{}}
\toprule
Phone & \system & $-\mathrm{BC}$ & $-\mathrm{TD}$ \\
\midrule
K70 & \textbf{27.92} & 25.60 & 19.77 \\
11 Pro & \textbf{14.68} & 13.79 & 11.00 \\
Mi 14 & \textbf{28.40} & 27.84 & 14.55 \\
Note 9 & \textbf{5.85} & 5.07 & 3.84 \\
All & \textbf{12.90} & 11.61 & 8.49 \\
\bottomrule
\end{tabular*}
\endgroup

\end{minipage}
\end{table}

\FloatBarrier
\subsection{RQ3: Native on-demand acceleration}
\label{sec:main-native-requests}

To test whether \system accelerates complete responses after audio prefill,
we compare \system, Target-only, SD, and SpecASR in native MNN
on Redmi K70 Pro and Xiaomi 14, using different 10-minute audio
traces. Both phones keep an Omni-7B target resident on
six CPU threads and a Qwen3-ASR-Audio-0.6B drafter on OpenCL. We sample
19 request positions from 60s to 600s of each trace, spaced 30s apart.
At each position, we prepare both models' audio KV for the latest 60s,
then request transcription. All speculative methods use
budget $b=3$. Figure~\ref{fig:native-ondemand-fourway} reports request-to-EOS
TPS, including instruction-suffix prefill but excluding earlier audio
preparation and model loading.

\noindent\textbf{Native request throughput.}\quad
On Redmi K70 Pro, \system achieves the highest measured TPS at all 19 positions, improving
on Target-only by 26.7--78.0\%. At 120s, throughput rises from 6.91 to
12.29 tokens/s, and request completion time falls from 25.05s to 14.07s.
Median TPS across positions is 10.61 for \system, compared with 7.77 for
SD and 7.48 for SpecASR. At 180s, fewer usable candidates limit the gain
to 26.7\%: only 50 draft tokens are accepted and 92 of 120 target calls
remain single-token steps, versus 108 accepted draft tokens and 21 of 65
single-token calls at 120s.
On Xiaomi 14, \system improves pooled TPS by 28.2\% over Target-only
and 8.5\% over SD.

\noindent\textbf{Sustained operation.}\quad
We additionally feed the same native model pair a continuous 10-minute
audiobook recording on Redmi K70 Pro, paced as 150 non-overlapping 4s
blocks. Both models remain resident, with audio and text context reset
every 12s. Processing includes audio encoding, prefill, and decoding as
input arrives. With active cooling and 75\% target audio-token retention,
mean block processing is 3.63s, below the 4s arrival interval. The p95
completion lag is 4.68s, measured from arrival of a block's final audio
sample to output completion, and maximum queue wait is 2.07s. This
configuration keeps pace with the input over 600s.
See Appendix~\ref{sec:native-sustained-streaming}.
On the first 60s Xiaomi 14 window, \system also reduces estimated phone energy
by 30\% compared with Target-only (239\,J versus 342\,J).
See Appendix~\ref{sec:native-request-protocol}.

\section{Conclusion}
\label{sec:conclusion}
We show that audio and the user request can supply speculative candidates
without following the target's evolving verified prefix. \system decouples
drafting from target-prefix updates, allowing the drafter and target to
advance concurrently while retaining target-side verification and correction.
Our Android evaluation spans two target models, four phones, seven datasets,
and three tasks. Phone-profile replay shows 42--76\% higher pooled ASR
throughput than target-only, with slowdowns on 5.7\% of windows versus
58.1--63.0\% for speculative baselines. Native on-demand execution with a 7B
target achieves up to 78\% higher throughput than target-only.
Gains depend on candidate supply and alignment. Future work includes
adaptive processor placement and scheduling, and support for full-duplex
interaction (Appendix~\ref{sec:discussion}).

\label{page:main-text-end}

\clearpage
\setlength{\textfloatsep}{\cadreDefaultTextFloatSep}
\setlength{\floatsep}{\cadreDefaultFloatSep}
\RemoveFromHook{env/figure/begin}[cadre-main-spacing]
\RemoveFromHook{env/table/begin}[cadre-main-spacing]
\subsection*{AI use statement}
Generative AI tools assisted with method formulation, experiment design,
implementation and execution, data processing and result interpretation,
mathematical arguments and proof drafting, figure preparation, and manuscript
drafting, revision, and formatting.
The authors defined the research questions and experimental scope, directed
this assistance through iterative instructions, and reviewed proposed changes,
results, and interpretations at each stage. Research and editorial decisions
remained with the authors, who take responsibility for the paper, implementation, and results.

\subsection*{Reproducibility statement}
Appendix~\ref{sec:problem-formulation} defines the input, output-comparison,
and timing protocols. Appendix~\ref{sec:design} specifies candidate
publication, verification, alignment, and the assumptions for output
equivalence. Model and device configurations, dataset selection, baselines,
development splits, and calibration procedures are documented in
Appendices~\ref{sec:appendix-eval-setup} and~\ref{sec:appendix-eval-deployment}.
Appendix~\ref{sec:replay-agreement} compares replay with native measurements;
Appendix~\ref{sec:appendix-eval-rq3} provides the native execution protocols
and output checks. Observed numerical discrepancies and their implications
are retained in Appendix~\ref{sec:numerical-sensitivity}.

\subsection*{Ethics statement}
This work studies execution scheduling for pretrained audio language models.
The evaluation uses existing speech datasets and prerecorded audio.
The method retains target-side verification
but does not address errors, bias, or harmful outputs of the underlying
models. Applications that record or retain speech should obtain appropriate
consent, protect audio and cached states, and respect applicable model and
dataset terms.

\label{page:references-start}
\bibliographystyle{iclr2027_conference}
\bibliography{references}

@article{liang2026dynamic,
  title = {Reliable and Efficient {LLM} Inference on Resource-Constrained Mobile Devices via Dynamic Scheduling},
  author = {Liang, Aoxing and Hong, Chengfan and Li, Yunzhe and Zhu, Hongzi},
  journal = {IEEE Internet of Things Journal},
  volume = {13},
  number = {9},
  pages = {19751--19754},
  year = {2026},
  doi = {10.1109/JIOT.2026.3659676}
}

@inproceedings{zhang2025sgdrc,
  title = {{SGDRC}: Software-Defined Dynamic Resource Control for Concurrent {DNN} Inference on {NVIDIA GPUs}},
  author = {Zhang, Yongkang and Yu, Haoxuan and Han, Chenxia and Wang, Cheng and Lu, Baotong and Li, Yunzhe and Jiang, Zhifeng and Li, Yang and Chu, Xiaowen and Li, Huaicheng},
  booktitle = {Proceedings of the 30th ACM SIGPLAN Annual Symposium on Principles and Practice of Parallel Programming},
  pages = {267--281},
  year = {2025},
  publisher = {Association for Computing Machinery},
  doi = {10.1145/3710848.3710863}
}

@inproceedings{li2024anole,
  title = {{Anole}: Adapting Diverse Compressed Models for Cross-Scene Prediction on Mobile Devices},
  author = {Li, Yunzhe and Zhu, Hongzi and Deng, Zhuohong and Cheng, Yunlong and Zhang, Liang and Chang, Shan and Guo, Minyi},
  booktitle = {2024 IEEE 44th International Conference on Distributed Computing Systems (ICDCS)},
  pages = {613--623},
  year = {2024},
  doi = {10.1109/ICDCS60910.2024.00063}
}

@article{li2025sceneaware,
  title = {A Scene-Aware Model Adaptation Scheme for Cross-Scene Online Inference on Mobile Devices},
  author = {Li, Yunzhe and Zhu, Hongzi and Deng, Zhuohong and Cheng, Yunlong and Zheng, Zimu and Zhang, Liang and Chang, Shan and Guo, Minyi},
  journal = {IEEE Transactions on Mobile Computing},
  volume = {24},
  number = {10},
  pages = {11061--11075},
  year = {2025},
  doi = {10.1109/TMC.2025.3574766}
}

@article{yuan2025numerical,
  title = {Understanding and Mitigating Numerical Sources of Nondeterminism in {LLM} Inference},
  author = {Yuan, Jiayi and Li, Hao and Ding, Xinheng and Xie, Wenya and Li, Yu-Jhe and Zhao, Wentian and Wan, Kun and Shi, Jing and Hu, Xia and Liu, Zirui},
  journal = {arXiv preprint arXiv:2506.09501},
  year = {2025},
  doi = {10.48550/arXiv.2506.09501},
  url = {https://arxiv.org/abs/2506.09501}
}

@article{he2025nondeterminism,
  title = {Defeating Nondeterminism in {LLM} Inference},
  author = {He, Horace and {Thinking Machines Lab}},
  journal = {Thinking Machines Lab: Connectionism},
  year = {2025},
  doi = {10.64434/tml.20250910},
  url = {https://thinkingmachines.ai/blog/defeating-nondeterminism-in-llm-inference/}
}

@inproceedings{leviathan2023speculative,
  title = {Fast Inference from Transformers via Speculative Decoding},
  author = {Leviathan, Yaniv and Kalman, Matan and Matias, Yossi},
  booktitle = {Proceedings of the 40th International Conference on Machine Learning},
  series = {Proceedings of Machine Learning Research},
  volume = {202},
  pages = {19274--19286},
  year = {2023},
  publisher = {PMLR},
  url = {https://proceedings.mlr.press/v202/leviathan23a.html}
}

@inproceedings{cai2024medusa,
  title = {Medusa: Simple {LLM} Inference Acceleration Framework with Multiple Decoding Heads},
  author = {Cai, Tianle and Li, Yuhong and Geng, Zhengyang and Peng, Hongwu and Lee, Jason D. and Chen, Deming and Dao, Tri},
  booktitle = {Proceedings of the 41st International Conference on Machine Learning},
  series = {Proceedings of Machine Learning Research},
  volume = {235},
  pages = {5209--5235},
  year = {2024},
  publisher = {PMLR},
  url = {https://proceedings.mlr.press/v235/cai24b.html}
}

@inproceedings{li2024eagle,
  title = {{EAGLE}: Speculative Sampling Requires Rethinking Feature Uncertainty},
  author = {Li, Yuhui and Wei, Fangyun and Zhang, Chao and Zhang, Hongyang},
  booktitle = {Proceedings of the 41st International Conference on Machine Learning},
  series = {Proceedings of Machine Learning Research},
  volume = {235},
  pages = {28935--28948},
  year = {2024},
  publisher = {PMLR},
  url = {https://proceedings.mlr.press/v235/li24bt.html}
}

@inproceedings{li2024eagle2,
  title = {{EAGLE-2}: Faster Inference of Language Models with Dynamic Draft Trees},
  author = {Li, Yuhui and Wei, Fangyun and Zhang, Chao and Zhang, Hongyang},
  booktitle = {Proceedings of the 2024 Conference on Empirical Methods in Natural Language Processing},
  pages = {7421--7432},
  year = {2024},
  publisher = {Association for Computational Linguistics},
  doi = {10.18653/v1/2024.emnlp-main.422}
}

@inproceedings{wei2025specasr,
  title = {{SpecASR}: Accelerating {LLM}-based Automatic Speech Recognition via Speculative Decoding},
  author = {Wei, Linye and Zhong, Shuzhang and Xu, Songqiang and Wang, Runsheng and Huang, Ru and Li, Meng},
  booktitle = {Proceedings of the 62nd ACM/IEEE Design Automation Conference},
  pages = {1--7},
  year = {2025},
  doi = {10.1109/DAC63849.2025.11132579}
}

@inproceedings{yusuf2024ssr,
  title = {Speculative Speech Recognition by Audio-Prefixed Low-Rank Adaptation of Language Models},
  author = {Yusuf, Bolaji and Baskar, Murali Karthick and Rosenberg, Andrew and Ramabhadran, Bhuvana},
  booktitle = {Interspeech 2024},
  pages = {792--796},
  year = {2024},
  doi = {10.21437/Interspeech.2024-298}
}

@inproceedings{segalfeldman2025whispermedusa,
  title = {Whisper in Medusa's Ear: Multi-head Efficient Decoding for Transformer-based {ASR}},
  author = {Segal-Feldman, Yael and Shamsian, Aviv and Navon, Aviv and Hetz, Gill and Keshet, Joseph},
  booktitle = {2025 IEEE International Conference on Acoustics, Speech, and Signal Processing},
  pages = {1--5},
  year = {2025},
  doi = {10.1109/ICASSP49660.2025.10888140},
  url = {https://doi.org/10.1109/ICASSP49660.2025.10888140}
}

@article{saon2026selfspec,
  title = {Self-Speculative Decoding for {LLM}-based {ASR} with {CTC} Encoder Drafts},
  author = {Saon, George and Thomas, Samuel and Fukuda, Takashi and Nagano, Tohru and Dekel, Avihu and Lastras, Luis},
  journal = {arXiv preprint arXiv:2603.11243},
  year = {2026},
  url = {https://arxiv.org/abs/2603.11243}
}

@inproceedings{okabe2025smud,
  title = {Simultaneous Masked and Unmasked Decoding with Speculative Decoding Masking for Fast {ASR} without Accuracy Loss},
  author = {Okabe, Koji and Yamamoto, Hitoshi},
  booktitle = {Interspeech 2025},
  pages = {634--638},
  year = {2025},
  doi = {10.21437/Interspeech.2025-382},
  url = {https://www.isca-archive.org/interspeech_2025/okabe25_interspeech.html}
}

@inproceedings{miao2024specinfer,
  title = {{SpecInfer}: Accelerating Large Language Model Serving with Tree-based Speculative Inference and Verification},
  author = {Miao, Xupeng and Oliaro, Gabriele and Zhang, Zhihao and Cheng, Xinhao and Wang, Zeyu and Zhang, Zhengxin and Wong, Rae Ying Yee and Zhu, Alan and Yang, Lijie and Shi, Xiaoxiang and Shi, Chunan and Chen, Zhuoming and Arfeen, Daiyaan and Abhyankar, Reyna and Jia, Zhihao},
  booktitle = {Proceedings of the 29th ACM International Conference on Architectural Support for Programming Languages and Operating Systems, Volume 3},
  pages = {932--949},
  year = {2024},
  publisher = {ACM},
  doi = {10.1145/3620666.3651335},
  url = {https://doi.org/10.1145/3620666.3651335}
}

@article{chen2024sequoia,
  title = {{Sequoia}: Scalable, Robust, and Hardware-aware Speculative Decoding},
  author = {Chen, Zhuoming and May, Avner and Svirschevski, Ruslan and Huang, Yuhsun and Ryabinin, Max and Jia, Zhihao and Chen, Beidi},
  journal = {arXiv preprint arXiv:2402.12374},
  year = {2024},
  doi = {10.48550/arXiv.2402.12374},
  url = {https://arxiv.org/abs/2402.12374}
}

@inproceedings{liu2024online,
  title = {Online Speculative Decoding},
  author = {Liu, Xiaoxuan and Hu, Lanxiang and Bailis, Peter and Cheung, Alvin and Deng, Zhijie and Stoica, Ion and Zhang, Hao},
  booktitle = {Proceedings of the 41st International Conference on Machine Learning},
  series = {Proceedings of Machine Learning Research},
  volume = {235},
  pages = {31131--31146},
  year = {2024},
  publisher = {PMLR},
  url = {https://proceedings.mlr.press/v235/liu24y.html}
}

@inproceedings{zhang2024draftverify,
  title = {Draft {\&} Verify: Lossless Large Language Model Acceleration via Self-Speculative Decoding},
  author = {Zhang, Jun and Wang, Jue and Li, Huan and Shou, Lidan and Chen, Ke and Chen, Gang and Mehrotra, Sharad},
  booktitle = {Proceedings of the 62nd Annual Meeting of the Association for Computational Linguistics (Volume 1: Long Papers)},
  pages = {11263--11282},
  year = {2024},
  publisher = {Association for Computational Linguistics},
  doi = {10.18653/v1/2024.acl-long.607},
  url = {https://aclanthology.org/2024.acl-long.607/}
}

@article{xu2025edgellm,
  title = {{EdgeLLM}: Fast On-device {LLM} Inference with Speculative Decoding},
  author = {Xu, Daliang and Yin, Wangsong and Zhang, Hao and Jin, Xin and Zhang, Ying and Wei, Shiyun and Xu, Mengwei and Liu, Xuanzhe},
  journal = {IEEE Transactions on Mobile Computing},
  volume = {24},
  number = {4},
  pages = {3256--3273},
  year = {2025},
  doi = {10.1109/TMC.2024.3513457}
}

@article{xu2023llmcad,
  title = {{LLMCad}: Fast and Scalable On-device Large Language Model Inference},
  author = {Xu, Daliang and Yin, Wangsong and Jin, Xin and Zhang, Ying and Wei, Shiyun and Xu, Mengwei and Liu, Xuanzhe},
  journal = {arXiv preprint arXiv:2309.04255},
  year = {2023},
  doi = {10.48550/arXiv.2309.04255},
  url = {https://arxiv.org/abs/2309.04255}
}

@article{wang2026lever,
  title = {Lever: Speculative {LLM} Inference on Smartphones},
  author = {Wang, Tuowei and Li, Fengzu and Sun, Yanfan and Gao, Wei and Ren, Ju},
  journal = {arXiv preprint arXiv:2605.16786},
  year = {2026},
  doi = {10.48550/arXiv.2605.16786},
  url = {https://arxiv.org/abs/2605.16786}
}

@inproceedings{yang2026pelm,
  title = {{PELM}: Power Efficient On-Device {LLM} Inference with Speculative Decoding and Dynamic Voltage Frequency Scaling},
  author = {Yang, Weisi and Xia, Stephen},
  booktitle = {Proceedings of the 2026 ACM/IEEE International Conference on Embedded Artificial Intelligence and Sensing Systems},
  pages = {438--451},
  year = {2026},
  publisher = {ACM},
  doi = {10.1145/3774906.3802783},
  url = {https://doi.org/10.1145/3774906.3802783}
}

@article{shi2026qwen3asr,
  title = {{Qwen3-ASR} Technical Report},
  author = {Shi, Xian and Wang, Xiong and Guo, Zhifang and Wang, Yongqi and Zhang, Pei and Zhang, Xinyu and Guo, Zishan and Hao, Hongkun and Yu, Xi and Yang, Baosong and Xu, Jin and Zhou, Jingren and Lin, Junyang},
  journal = {arXiv preprint arXiv:2601.21337},
  year = {2026},
  url = {https://arxiv.org/abs/2601.21337}
}

@inproceedings{panayotov2015librispeech,
  title = {{LibriSpeech}: An {ASR} Corpus Based on Public Domain Audio Books},
  author = {Panayotov, Vassil and Chen, Guoguo and Povey, Daniel and Khudanpur, Sanjeev},
  booktitle = {2015 IEEE International Conference on Acoustics, Speech and Signal Processing},
  pages = {5206--5210},
  year = {2015},
  doi = {10.1109/ICASSP.2015.7178964}
}

@inproceedings{carletta2006ami,
  title = {The {AMI} Meeting Corpus: A Pre-Announcement},
  author = {Carletta, Jean and Ashby, Simone and Bourban, Sebastien and Flynn, Mike and Guillemot, Mael and Hain, Thomas and Kadlec, Jaroslav and Karaiskos, Vasilis and Kraaij, Wessel and Kronenthal, Melissa and Lathoud, Guillaume and Lincoln, Mike and Lisowska, Agnes and McCowan, Iain and Post, Wilfried and Reidsma, Dennis and Wellner, Pierre},
  booktitle = {Machine Learning for Multimodal Interaction, Second International Workshop},
  pages = {28--39},
  year = {2006},
  publisher = {Springer},
  doi = {10.1007/11677482_3}
}

@inproceedings{conneau2022fleurs,
  title = {{FLEURS}: Few-Shot Learning Evaluation of Universal Representations of Speech},
  author = {Conneau, Alexis and Ma, Min and Khanuja, Simran and Zhang, Yu and Axelrod, Vera and Dalmia, Siddharth and Riesa, Jason and Rivera, Clara and Bapna, Ankur},
  booktitle = {2022 IEEE Spoken Language Technology Workshop},
  pages = {798--805},
  year = {2022},
  doi = {10.1109/SLT54892.2023.10023141}
}

@inproceedings{tang2021kespeech,
  title = {{KeSpeech}: An Open Source Speech Dataset of Mandarin and Its Eight Subdialects},
  author = {Tang, Zhiyuan and Wang, Dong and Xu, Yanguang and Sun, Jianwei and Lei, Xiaoning and Zhao, Shuaijiang and Wen, Cheng and Tan, Xingjun and Xie, Chuandong and Zhou, Shuran and Yan, Rui and Lv, Chenjia and Han, Yang and Zou, Wei and Li, Xiangang},
  booktitle = {Proceedings of the Neural Information Processing Systems Track on Datasets and Benchmarks},
  volume = {1},
  year = {2021},
  url = {https://datasets-benchmarks-proceedings.neurips.cc/paper/2021/hash/0336dcbab05b9d5ad24f4333c7658a0e-Abstract-round2.html}
}

@inproceedings{zhang2022m4singer,
  title = {{M4Singer}: A Multi-Style, Multi-Singer and Musical Score Provided Mandarin Singing Corpus},
  author = {Zhang, Lichao and Li, Ruiqi and Wang, Shoutong and Deng, Liqun and Liu, Jinglin and Ren, Yi and He, Jinzheng and Huang, Rongjie and Zhu, Jieming and Chen, Xiao and Zhao, Zhou},
  booktitle = {Advances in Neural Information Processing Systems Datasets and Benchmarks Track},
  year = {2022},
  url = {https://proceedings.neurips.cc/paper_files/paper/2022/hash/2de60892dd329683ec21877a4e7c3091-Abstract-Datasets_and_Benchmarks.html}
}

@inproceedings{wang2021covost2,
  title = {{CoVoST 2} and Massively Multilingual Speech Translation},
  author = {Wang, Changhan and Wu, Anne and Gu, Jiatao and Pino, Juan},
  booktitle = {Interspeech 2021},
  year = {2021},
  pages = {2247--2251},
  doi = {10.21437/Interspeech.2021-2027},
  url = {https://www.isca-archive.org/interspeech_2021/wang21s_interspeech.html}
}

@inproceedings{lee2018spokensquad,
  title = {{Spoken SQuAD}: A Study of Mitigating the Impact of Speech Recognition Errors on Listening Comprehension},
  author = {Lee, Chia-Hsuan and Wu, Szu-Lin and Liu, Chi-Liang and Lee, Hung-yi},
  booktitle = {Interspeech 2018},
  year = {2018},
  pages = {3459--3463},
  doi = {10.21437/Interspeech.2018-1714},
  url = {https://www.isca-archive.org/interspeech_2018/lee18d_interspeech.html}
}

@inproceedings{jiang2020mnn,
  title = {{MNN}: A Universal and Efficient Inference Engine},
  author = {Jiang, Xiaotang and Wang, Huan and Chen, Yiliu and Wu, Ziqi and Wang, Lichuan and Zou, Bin and Yang, Yafeng and Cui, Zongyang and Cai, Yu and Yu, Tianhang and Lyu, Chengfei and Wu, Zhihua},
  booktitle = {Proceedings of Machine Learning and Systems},
  volume = {2},
  year = {2020},
  url = {https://proceedings.mlsys.org/paper_files/paper/2020/hash/bc19061f88f16e9ed4a18f0bbd47048a-Abstract.html}
}

@misc{mnn2026,
  title = {{MNN}: A Lightweight Deep Neural Network Inference Engine},
  author = {{Alibaba Group}},
  year = {2026},
  howpublished = {Software repository},
  url = {https://github.com/alibaba/MNN}
}

@inproceedings{he2019streaming,
  title = {Streaming End-to-End Speech Recognition for Mobile Devices},
  author = {He, Yanzhang and Sainath, Tara N. and Prabhavalkar, Rohit and McGraw, Ian and Alvarez, Raziel and Zhao, Ding and Rybach, David and Kannan, Anjuli and Wu, Yonghui and Pang, Ruoming and Liang, Qiao and Bhatia, Deepti and Shangguan, Yuan and Li, Bo and Pundak, Golan and Sim, Khe Chai and Bagby, Tom and Chang, Shuo-yiin and Rao, Kanishka and Gruenstein, Alexander},
  booktitle = {IEEE International Conference on Acoustics, Speech and Signal Processing (ICASSP)},
  year = {2019},
  url = {https://arxiv.org/abs/1811.06621}
}

@article{orhon2025whisperkit,
  title = {{WhisperKit}: On-device Real-time {ASR} with Billion-Scale Transformers},
  author = {Orhon, Atila and Okan, Arda and Durmus, Berkin and Nagengast, Zach and Pacheco, Eduardo},
  journal = {arXiv preprint arXiv:2507.10860},
  year = {2025},
  doi = {10.48550/arXiv.2507.10860},
  url = {https://arxiv.org/abs/2507.10860}
}

@inproceedings{zhang2026ltd,
  title = {Learning to Draft: Adaptive Speculative Decoding with Reinforcement Learning},
  author = {Zhang, Jiebin and Yu, Zhenghan and Wang, Liang and Yang, Nan and Yu, Eugene J. and Li, Zheng and Song, Yifan and Zhu, Dawei and Zhang, Xingxing and Wei, Furu and Li, Sujian},
  booktitle = {International Conference on Learning Representations},
  year = {2026},
  url = {https://arxiv.org/abs/2603.01639}
}

@inproceedings{kim2026metasd,
  title = {Multi-Drafter Speculative Decoding with Alignment Feedback},
  author = {Kim, Taehyeon and Jung, Hojung and Yun, Se-Young},
  booktitle = {Findings of the Association for Computational Linguistics: ACL 2026},
  pages = {32532--32573},
  year = {2026},
  doi = {10.18653/v1/2026.findings-acl.1629},
  url = {https://aclanthology.org/2026.findings-acl.1629/}
}

@inproceedings{do2026unispec,
  title = {{UniSpec}: Training-Free Speculative Decoding for Robust {LLM} Acceleration Across Languages and Hardware},
  author = {Do, Truong Dinh and Le, Nguyen-Khang and Nguyen, Le-Minh},
  booktitle = {Proceedings of the 64th Annual Meeting of the Association for Computational Linguistics},
  year = {2026},
  url = {https://aclanthology.org/2026.acl-long.285/}
}

@article{ma2026ahasd,
  title = {{AHASD}: Asynchronous Heterogeneous Architecture for {LLM} Adaptive Drafting Speculative Decoding on Mobile Devices},
  author = {Ma, Zirui and Fan, Zhihua and Li, Wenxing and Wu, Haibin and Zhang, Fulin and Ye, Xiaochun and Li, Wenming},
  journal = {arXiv preprint arXiv:2604.25326},
  year = {2026},
  url = {https://arxiv.org/abs/2604.25326}
}

@inproceedings{kumar2026ssd,
  title = {Speculative Speculative Decoding},
  author = {Kumar, Tanishq and Dao, Tri and May, Avner},
  booktitle = {International Conference on Learning Representations},
  year = {2026},
  url = {https://proceedings.iclr.cc/paper_files/paper/2026/hash/1b96f01343ff10150e6719eb163e1536-Abstract-Conference.html}
}

@inproceedings{liu2025pearl,
  title = {{PEARL}: Parallel Speculative Decoding with Adaptive Draft Length},
  author = {Liu, Tianyu and Li, Yun and Lv, Qitan and Liu, Kai and Zhu, Jianchen and Hu, Winston and Sun, Xiao},
  booktitle = {International Conference on Learning Representations},
  year = {2025},
  url = {https://proceedings.iclr.cc/paper_files/paper/2025/hash/03b1043052700b1a471996b0baf309d4-Abstract-Conference.html}
}

@inproceedings{timor2025dsi,
  title = {Distributed Speculative Inference ({DSI}): Speculation Parallelism for Provably Faster Lossless Language Model Inference},
  author = {Timor, Nadav and Mamou, Jonathan and Korat, Daniel and Berchansky, Moshe and Pereg, Oren and Wasserblat, Moshe and Galanti, Tomer and Gordon-Kiwkowitz, Michal and Harel, David},
  booktitle = {International Conference on Learning Representations},
  year = {2025},
  url = {https://proceedings.iclr.cc/paper_files/paper/2025/hash/b36554b97da741b1c48c9de05c73993e-Abstract-Conference.html}
}

@article{xu2025qwen25omni,
  title = {{Qwen2.5-Omni} Technical Report},
  author = {Xu, Jin and Guo, Zhifang and He, Jinzheng and Hu, Hangrui and He, Ting and Bai, Shuai and Chen, Keqin and Wang, Jialin and Fan, Yang and Dang, Kai and Zhang, Bin and Wang, Xiong and Chu, Yunfei and Lin, Junyang},
  journal = {arXiv preprint arXiv:2503.20215},
  year = {2025},
  doi = {10.48550/arXiv.2503.20215},
  url = {https://arxiv.org/abs/2503.20215}
}

@inproceedings{wang2025speechprosody,
  title = {Can {AI} Understand Mandarin Speech Prosody? A Framework and Benchmark Showcase},
  author = {Wang, Zilong and Zhang, Xiaoxue and Jiang, Xinyang and Song, Kaitao and Yu, Jue},
  booktitle = {Interspeech 2025},
  pages = {5378--5382},
  year = {2025},
  doi = {10.21437/Interspeech.2025-1873},
  url = {https://www.isca-archive.org/interspeech_2025/wang25v_interspeech.html}
}

@article{gibier2025segmentwise,
  author = {Gibier, Marcel and Duroselle, Rapha{\"e}l and Serrano, Pierre and Boeffard, Olivier and Bonastre, Jean-Fran{\c{c}}ois},
  title = {Segmentwise Pruning in Audio-Language Models},
  journal = {arXiv preprint arXiv:2511.14293},
  year = {2025},
  url = {https://arxiv.org/abs/2511.14293}
}

@misc{openbmb2025minicpmo26,
  author = {{OpenBMB}},
  title = {{MiniCPM-o 2.6}: A {GPT-4o} Level {MLLM} for Vision, Speech and Multimodal Live Streaming on Your Phone},
  year = {2025},
  howpublished = {Official model card},
  url = {https://huggingface.co/openbmb/MiniCPM-o-2_6}
}

@article{sun2026omnimem,
  author = {Sun, Guangzhi and Li, Yixuan and Yang, Yudong and Zhang, Chao},
  title = {{OmniMem}: Perturbation-aware Memory Compression for Streaming Audio-Visual {LLMs}},
  journal = {arXiv preprint arXiv:2606.07577},
  year = {2026},
  doi = {10.48550/arXiv.2606.07577},
  url = {https://arxiv.org/abs/2606.07577}
}

@inproceedings{zulfikar2024memoro,
  author = {Zulfikar, Wazeer and Chan, Samantha and Maes, Pattie},
  title = {{Memoro}: Using Large Language Models to Realize a Concise Interface for Real-Time Memory Augmentation},
  booktitle = {Proceedings of the CHI Conference on Human Factors in Computing Systems},
  year = {2024},
  doi = {10.1145/3613904.3642450},
  url = {https://arxiv.org/abs/2403.02135}
}

@article{hegde2026larag,
  author = {Hegde, Kartik and Sridhar, Arvind Krishna and Vakada, Naveen and Guo, Yinyi and Visser, Erik},
  title = {Event-Grounded Question Answering over Long Audio via Structured Retrieval},
  journal = {arXiv preprint arXiv:2602.14612},
  year = {2026},
  doi = {10.48550/arXiv.2602.14612},
  url = {https://arxiv.org/abs/2602.14612}
}

@misc{qualcomm2026sd750g,
  author = {{Qualcomm}},
  title = {{Snapdragon 750G 5G Mobile Platform}: Product Brief},
  year = {2026},
  howpublished = {Official product specifications},
  note = {Accessed September 25, 2026},
  url = {https://www.qualcomm.com/media/documents/files/snapdragon-750g-5g-mobile-platform-product-brief.pdf}
}

@misc{qualcomm2026sd888,
  author = {{Qualcomm}},
  title = {{Snapdragon 888 5G Mobile Platform}: Product Brief},
  year = {2026},
  howpublished = {Official product specifications},
  note = {Accessed September 25, 2026},
  url = {https://www.qualcomm.com/content/dam/qcomm-martech/dm-assets/documents/prod_brief_qcom_sd888_5g_0.pdf}
}

@misc{xiaomi2026mi14specs,
  author = {{Xiaomi}},
  title = {{Xiaomi 14 Specifications}},
  year = {2026},
  howpublished = {Official product specifications},
  note = {Accessed September 25, 2026},
  url = {https://www.mi.com/global/product/xiaomi-14/specs/}
}

@misc{qualcomm2023sd8gen3,
  author = {{Qualcomm}},
  title = {{Snapdragon 8 Gen 3 Mobile Platform}: Product Brief},
  year = {2023},
  note = {87-71408-1, Revision B},
  url = {https://www.qualcomm.com/content/dam/qcomm-martech/dm-assets/images/company/news-media/media-center/press-kits/snapdragon-summit-2023/documents/Snapdragon8Gen3_%20ProductBrief.pdf}
}

@misc{ggml2026opencl,
  author = {{llama.cpp contributors}},
  title = {{llama.cpp} for {OpenCL}: Hardware Support},
  year = {2026},
  note = {Official backend documentation, accessed September 25, 2026},
  url = {https://github.com/ggml-org/llama.cpp/blob/master/docs/backend/OPENCL.md}
}

@misc{xiaomi2026k70specs,
  author = {{Xiaomi}},
  title = {{Redmi K70 Pro Specifications}},
  year = {2026},
  howpublished = {Official product specifications},
  note = {Accessed September 25, 2026},
  url = {https://www.mi.com/redmi-k70-pro/specs}
}

@misc{xiaomi2026mi11prospecs,
  author = {{Xiaomi}},
  title = {{Mi 11 Pro Specifications}},
  year = {2026},
  howpublished = {Official product specifications},
  note = {Accessed September 25, 2026},
  url = {https://www.mi.com/mi11Pro/specs}
}

@article{luo2026locality,
  author = {Luo, Jiale and Liang, Xiaoyu and Hu, Haoji},
  title = {Locality Matters for Training-Free Audio Token Compression in Audio-Language Models},
  journal = {arXiv preprint arXiv:2605.25179},
  year = {2026},
  doi = {10.48550/arXiv.2605.25179},
  url = {https://arxiv.org/abs/2605.25179}
}

\clearpage
\appendix
\label{page:appendix-start}
\FloatBarrier
\section{Input protocol and measurement contract}
\label{sec:problem-formulation}

\subsection{Requests, audio availability, and reference execution}
A fixed protocol \(\mathcal S\) specifies audio preparation, prompt
assembly, the audio view at each target position, initial state, context
retention, and termination. Request \(u\) arrives at \(s_u\) with instruction
\(p_u\). Protocol \(\mathcal S\) selects its audio view \(x_u\) from audio
received by \(s_u\). Reference input \(z_u=(x_u,p_u)\) is fixed throughout
decoding, even if more audio arrives before the response completes.
Neither the producer nor the target may extend this input with later audio.
Comparisons use the same protocol and target, rather than changing the
reference's context to accommodate a faster implementation.

The on-demand setting prepares audio before the request. A reusable
state \(C_u\) contains the resulting target language-model KV,
positions, and other state needed to continue the same reference execution;
encoder embeddings alone do not remove language-model prefill. Reuse also
requires the cached tokens to be a valid prefix of the reference prompt,
with matching audio features, positions, and attention conventions. An
instruction that changes earlier prompt tokens can require reconstruction;
a cache cannot simply be attached to an arbitrary prompt. Processing the
instruction using compatible audio KV produces the initial decoding state
\(\kappa_{u,0}\). Each response starts from this state, without
silently reusing a previous response's generated text. Background work,
cache growth, and any remaining work at request arrival must be recorded.
This preparation protocol is shared across compared methods.

The target and drafter maintain separate states. Retaining target audio KV
does not prefill another model: source input preparation and any source
prefill remain producer work. The request-scoped candidate buffer starts
empty in the primary comparison. Producer dispatch satisfies
\(\tau_u\geq s_u\); no task-specific response is assumed available before
its instruction. Once dispatched, the producer does not wait for the target
to accept earlier candidates. Any variant using candidates prepared before
a request requires a separate preparation and timing protocol.

The current full-audio study makes a complete window available before
initial audio/prompt prefill and retains target context throughout that
window. The Qwen source reads the same available waveform and publishes
candidate tokens incrementally. External Parakeet and SenseVoice sources
in the deployment-selection studies instead process 2.56-second waveform
chunks; their chunk boundaries do not reset the target. The full-window
study isolates decoding and does not measure a live microphone stream.

For a fixed request \(z_u\), let \(Y_T(z_u;\mathcal S)\) denote the greedy target
reference under deterministic tie-breaking. Successful speculative execution
must satisfy
\begin{equation}
 Y_{\system}(z_u;\mathcal S)=Y_T(z_u;\mathcal S)
 \label{eq:exact-output}
\end{equation}
and restore the corresponding target state at each committed prefix.
Recognition error against an annotation is a separate quality metric;
matching a canonical token trace in replay does not validate native KV
restoration. Appendix~\ref{sec:exactness-proof} states the required backend
conditions, and Appendix~\ref{sec:evaluation} retains numerical discrepancies
observed in native checks.

\subsection{Configuration and execution boundaries}
The decoding scheme \(\mathcal A\) uses a fixed task-specific drafter and
a budget selected before deployment. Ready candidates determine the block
processed at each target boundary.
In the four-phone ASR corpus comparison, three phones retain Qwen/budget 7
and Note 9 Pro uses Qwen/budget 3. Appendix~\ref{sec:deployment-selection}
documents their development splits and selection scope. At a verification
boundary, readiness, confidence stopping, the configured cap, and the
remaining output budget determine the actual chain length. A proposal cap
\(b\) permits up to \(b\) candidates after one known target root, so physical
verification width satisfies \(K\leq b+1\).

Candidate publication preserves request/source identity and token order.
At each boundary, the target uses already-published candidates and earlier
verified outcomes. When no usable candidate is ready, it advances alone.
The target verifies proposals before committing tokens.

\subsection{Response latency, throughput, and memory}
The problem in Section~\ref{sec:main-problem} minimizes mean request latency
\(L_u^{\mathcal A}=f_u^{\mathcal A}-s_u\), ending when the complete requested output becomes
available. Its clock begins at the prescribed request time, even if audio
prefill has not finished. Request processing, residual prefill, producer
startup, queue work, verification, correction, and alignment after
that time are included. Overlapping stages share one wall clock. Producer
cleanup after the final output is recorded separately and charged to any
later requests it delays. The optional tail constraint
\(Q_\alpha(L_u^{\mathcal A})\leq L_{\max}\) is an application requirement, not a
guarantee implied by mean decoding throughput.

The current corpus studies measure committed tokens per post-prefill time.
That boundary includes exposed source work, verification,
correction/alignment, queue work, and charged producer completion; it
excludes loading and initial audio/prompt prefill. These measurements
isolate decoding cost but do not reproduce the request clock or establish
that a growing audio KV cache remains ready on a phone. Native case studies
retain their stated completion and cleanup boundaries and are not pooled
across different runners.

The continuous-output deployment in
Appendix~\ref{sec:native-sustained-streaming} instead measures
\(\ell_i=f_i-r_i\), from a block's final audio sample to its completed
output. Its per-block prefill, automatic output, and 12s context resets
differ from persistent prefill followed by an on-demand request. Its RTF
and lag cannot be relabeled as request latency. Memory accounting in the
on-demand setting includes resident models, retained audio KV and position
state, source and candidate buffers, and temporary allocations; finite
capacity bounds the audio history that can be served without eviction.

\FloatBarrier
\section{Empirical design diagnostics}
\label{sec:motivation}
\label{sec:main-empirical}
This appendix gives the measurement protocols for the drafting and
verification diagnostics in Section~\ref{sec:main-method}, together with
an additional check of how candidate acceptance varies across corpora.

\subsection{Acceptance is content dependent}

We measure strict prefix acceptance over 28,595 audio blocks from four
corpora using a Parakeet 110M CTC proposer and a Qwen3-ASR 1.7B target at
\(K=32\). The median accepted ratio is 0.55 on AMI and 0.52 on SPGISpeech,
falls to 0.31 on GigaSpeech, and is zero on Earnings22. The distributions
overlap, showing variation within each corpus as well as across corpora.

\subsection{Device-specific verification profiles}
\label{sec:verification-width-profile}

Figure~\ref{fig:verifier-width} uses native full-logits width sweeps on
Redmi K70 Pro and Xiaomi 11 Pro. Both use a Qwen3-ASR 1.7B OpenCL target,
the same 59.5125-second LibriSpeech calibration audio, and three target-prefix
offsets (0, 64, and 128 tokens). After audio/prompt prefill, each width
is measured five times at each offset, excluding warmup and with no
concurrently executing producer. Both models remain resident.
The plot uses the common range \(K=2,\ldots,15\), with 15 measurements
per phone and width. K70 Pro also completes measurements at \(K=16\), which are not plotted. Xiaomi 11 Pro's wider trial
contains only warmup records and no successful completion.
K70 Pro sweeps widths in one job, whereas Xiaomi 11 Pro uses grouped jobs
with additional drafter component probes between verifier sweeps.
The verifier operation and contexts are matched, while job ordering and
recorded operating conditions are retained.

\begin{figure}[!htb]
  \centering
  \begin{minipage}[c]{0.51\linewidth}
    \includegraphics[width=\linewidth]{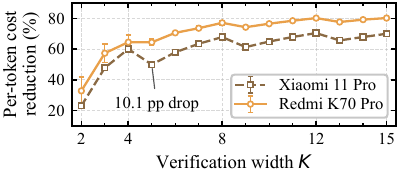}
  \end{minipage}\hfill
  \begin{minipage}[c]{0.46\linewidth}
    \setlength{\abovecaptionskip}{0pt}
    \caption{\textbf{Batching savings depend on phone and width.}
    We report per-token verification cost reductions for Qwen3-ASR 1.7B
    on OpenCL relative to the \(K=1\) verifier. Costs average 15 calls.
    Bars scale one call-time standard deviation by the same baseline.
    These are verifier savings, not full-request gains.}
    \label{fig:verifier-width}
  \end{minipage}
  \Description{One panel compares per-token full-logits verification cost
  savings on Xiaomi 11 Pro and Redmi K70 Pro for physical widths two through
  fifteen. Each cost is compared with the same phone's single-token
  full-logits verifier. A call that checks K positions is compared with K
  single-position calls. The Xiaomi curve drops by 10.1 percentage points
  from four to five. The reference is not the scalar Target-only decoder.}
\end{figure}

The timer covers the full-logits verification call, excluding loading,
initial prefill, pre-call context setup, and top-1 extraction. Each input
is a known target-reference continuation. All retained rows report
successful verification, actual width equal to requested \(K\), logits
shape \([1,K,151936]\), and no top-1 mismatches against the reference.
Let \(\bar V(K)\) be the mean call time for width \(K\). The plotted
reduction in cost per verified position is
\begin{equation}
 R(K)=100\left(1-\frac{\bar V(K)}{K\bar V(1)}\right)\%.
 \label{eq:verification-cost-reduction}
\end{equation}
The baseline consists of 15 additional full-logits \(K=1\) measurements
per phone from the same jobs. It uses the verification entry point,
whereas native Target-only uses a separate scalar decoding entry point.
Thus \(R(K)\) measures hardware batching efficiency at the profiled
contexts, not observed Target-only request-latency reduction.
Error bars are \(100s_K/(K\bar V(1))\) percentage points, where \(s_K\)
is the sample standard deviation across the 15 width-\(K\) calls.
The measured baseline mean is treated as fixed for this visualization.
These bars describe call-time dispersion, not confidence intervals for
the estimated ratio.

Increasing \(K\) from 4 to 5 reduces per-position savings from 60.1\%
to 50.0\% on Xiaomi 11 Pro, while they remain near 64.5\% on Redmi K70
Pro. The corresponding decreases are 10.1 and 0.1 percentage points,
respectively. These costs motivate profiling each phone/backend configuration.
Configuration selection then compares pooled validation throughput and
retains one pair for deployment, as described in
Appendix~\ref{sec:deployment-selection}.
The study comprises 435 width-sweep measurements and 30 additional
single-token baseline measurements.

\FloatBarrier
\subsection{Drafting overhead and prefix independence}
\label{sec:drafting-diagnostics}

\subsubsection{Native serial decoding breakdown}
\label{sec:serial-drafting-breakdown}

Figure~\ref{fig:prefix-context-cdf}(a) decomposes native serial SD on the
same 59.5125-second LibriSpeech calibration window across the four phones.
Each run uses a CPU Qwen3-ASR 0.6B drafter, an OpenCL 1.7B target, and seven
candidates after one known token (physical width \(K=8\)). All five formal
runs per phone are retained, excluding warmup. We sum measured draft
generation and prefix catch-up as drafting, sum target verification
separately, and assign the remaining decode time to other loop work.
Each plotted share divides the accumulated component time by accumulated
post-prefill decode time. Model loading and initial audio/prompt prefill
are outside this boundary.

Drafting occupies 49.4\%, 59.1\%, 75.1\%, and 51.2\% of serial decode time
on Xiaomi 11 Pro, Redmi K70 Pro, Xiaomi 14, and Redmi Note 9 Pro,
respectively. Every run reaches EOS and matches its local target reference.
Note 9 Pro emits 238 tokens per run, while the other phones emit 239.
The phones retain their recorded deployment and operating conditions,
including different CPU scheduling and frequencies. This is a within-run
cost breakdown on one audio window, rather than a controlled ranking of
processor speeds or a corpus-wide estimate. These measurements motivate
reducing serial drafting overhead through target-decoupled generation
and overlap.

\subsubsection{Audio proposals can run ahead}
\label{sec:prefix-context-probe}

A context-sensitivity probe compares Qwen3-ASR 0.6B proposals with and
without preceding transcript context against a Qwen3-ASR 1.7B target.
The 480 paired blocks comprise 120 from each of AMI, GigaSpeech,
SPGISpeech, and Earnings22. Each block spans seconds 4--8 of its source.
The supplied context approximates the preceding 0--4 seconds by slicing
reference-transcript words in proportion to audio duration. The
with-context condition (fresh) supplies this preceding text. The
without-context condition (stale) supplies no preceding text for any tested block. Both
conditions use the same supplied context for target verification and a
32-token proposal cap. This is a supplied-context ablation, not a replay
of actual target-generated text or KV synchronization.

Acceptance is the consecutive accepted-token count divided by the proposed
count within each condition. All 960 rows have nonempty proposals and
successful scores. Figure~\ref{fig:prefix-context-cdf}(b) plots all 480 paired
differences, without-context minus with-context. The median difference is
zero, and 347 pairs (72.3\%) have unchanged acceptance, including 73 pairs
where both conditions accept no token. Removing context reduces acceptance
by more than 0.05 in 9.6\% of pairs and increases it by more than 0.05 in
15.6\%. These results motivate drafting without waiting for updated text
context; they do not imply unchanged output tokens or uniform acceptance
across tasks. The target-decoupled ablations and native cases evaluate the
asynchronous execution interface directly.

The current Qwen producer reads a complete available window and publishes
batches of eight candidate tokens without target-text input. It continues
while the target verifies published candidates and does not restart after
a disagreement. The Parakeet distribution above is a motivating model
probe; native Parakeet comparisons are reported separately in
Appendix~\ref{sec:native-scheduling-cases}.

\subsection{Native drafting-cost check}

On one 60-second AMI development window on Redmi K70 Pro, five interleaved repeats give
14.58 TPS for target-only, 16.25 for Parakeet/$K=4$, and 16.03 for
SenseVoice/$K=4$ (11.5\% and 9.9\% gains). Native timing includes candidate
generation, verification, and recovery; all outputs match. Increasing
Parakeet to $K=8$ accepts a median of 43 rather
than 39 extra tokens yet lowers throughput to 15.66 TPS. Thus drafter cost and
verification budget must be judged by net throughput.
Appendix~\ref{sec:k70-drafter-case} gives the protocol and controls for this
single-window check.

\FloatBarrier
\section{Candidate publication, verification, and correctness}
\label{sec:design}

This appendix expands the target-decoupled interface in
Section~\ref{sec:main-method}. The evaluated runtime consumes chains;
its configuration is fixed before testing. Deployment-selection details
are in Appendix~\ref{sec:deployment-selection}.

\subsection{Source records and incremental publication}
\label{sec:appendix-drafting}
A logical source record identifies its request, fixed audio view, producer,
and token positions. In the existing single-window traces, the window ID
identifies this scope. The record contains the fixed input \(z_u=(x_u,p_u)\);
later audio or another request cannot change it.
A target-decoupled generator reads this audio view and instruction and,
when autoregressive, its own generated prefix. It receives neither the
latest target text nor target KV. Drafter output is converted into the
target vocabulary before publication; conversion cost belongs to source
execution.

The Qwen producer reads a complete available window and appends immutable
batches of eight candidate tokens. Publication does not wait for the whole
source sequence to finish. Each batch retains its source ID, producer ID,
token offset, publication time, and available confidence. Parakeet and
SenseVoice publish each completed 2.56-second source chunk atomically in
the external-source experiments. Their confidence values are unavailable
and are not synthesized by replay.

The buffer exposes only published tokens and preserves their order within
each request and source.
Verification reads an immutable snapshot while the producer appends later
batches. The consumption cursor advances independently of source
generation. Target context persists across publication boundaries.
Rolling native controls may combine consecutive ready batches into one
chain while retaining per-token source provenance.

\subsection{Verification and state repair}
\label{sec:tree-contract}
Before the first decoding step, the target processes the instruction
using compatible audio KV to obtain \(\kappa_{u,0}\),
using the same prompt and position conventions as target-only execution.
The cache-compatibility requirements and persistent-KV measurement scope
are specified in Appendix~\ref{sec:problem-formulation}. Candidate
generation uses separate state and receives no generated target prefix.

A verification call processes one known target root followed by up to
\(b\) candidate tokens under the actual committed prefix. Its physical
width \(K\leq b+1\) is also limited by candidate readiness and remaining
output length. Physical width,
logit indexing, attention, and cache positions must agree with the backend
contract. Only the matching greedy prefix is committed. At disagreement,
the target emits its correction, removes rejected cache entries, and
processes the correction under the same state convention as reference
execution. Every physical width requires its own feasibility and cost check.

The known target token can be carried into the next verification call;
full acceptance does not require an unconditional extra K1 forward.
If no compatible candidate is ready, a K1 target step makes progress.
Native paths that defer cache updates track logical accepted length
separately from the physical cache and remove rejected entries before the
next target forward. Appendix~\ref{sec:exactness-proof} states the
reference-equivalence and state-restoration conditions required for exact
greedy execution.

\subsection{Candidate alignment}
\label{sec:candidate-recycling}
We adapt local continuation matching \citep{wei2025specasr} to published
target-decoupled candidates. The source generator continues independently;
only the buffer's candidate cursor is realigned.

Full acceptance advances the cursor. After rejection and correction
processing, the next available target token is the first alignment anchor.
If it equals the candidate at the mismatch position, that position remains
eligible, allowing target insertions without discarding a useful suffix.
Otherwise, a bounded local search looks for the next target token, then
for the correction token and its following continuation. Unmatched
positions are retired and the target can continue alone. Exhausted records
are discarded, with source order and provenance retained.

Every reused token is verified under the current target prefix, and no old
target KV is attached to a recycled suffix. Repeated-token anchors can
select an unhelpful suffix, increasing verification work.
All alignment and recovery work remains inside the stated timing boundary.

\subsection{Correctness conditions and proof}
\label{sec:theory-proofs}
\label{sec:exactness-proof}
We specialize the standard output-preserving verification principle of
speculative decoding \citep{leviathan2023speculative} to the greedy decoding
protocol in Section~\ref{sec:main-problem}.

For on-demand execution, fix one request and its audio view. Below we
suppress the request index and include its instruction in \(\mathcal S\).
The initial-state condition includes a compatible prefilled audio KV cache.

For the fixed request, protocol \(\mathcal S\) prescribes the audio view,
initial request-prefilled state, and termination rule.
Two decoding schemes at the same committed prefix use that same view, even
if they reach the prefix at different wall times. Merely giving both
access to all audio that has arrived by their respective execution
times would not suffice. Source-side partitioning and token publication
do not reset or extend the target's request context.

Let \(g(\kappa)\) return the target's next token with deterministic tie-breaking,
and let \(F(\kappa,z)\) advance the target state after token \(z\). The state
includes the prescribed audio view, decoder KV, positions, and any pending-token
convention. The reference obeys
\begin{equation}
 y_{n+1}^T=g(\kappa_n^T),\qquad
 \kappa_{n+1}^T=F(\kappa_n^T,y_{n+1}^T).
 \label{eq:reference-transition}
\end{equation}
The verification and alignment interface has the following requirements.

\begin{enumerate}
\item \emph{Pathwise equivalence.} For every position in a candidate chain,
batched target execution produces the same greedy decision and retainable
state as sequential execution on the committed prefix followed by the
preceding candidates. Correct causal masks, positions, audio context,
padding, and logit indexing are required. Masks alone do not establish numerical
state equivalence across different execution kernels.
\item \emph{Commit and repair.} Verification commits only the consecutive
candidates matching the target's greedy tokens. At a mismatch it emits the
target's own correction. Retained state contains only the committed prefix;
rejected-candidate entries are removed. The correction is processed under the same state
convention. If compaction cannot establish equivalence, canonical replay
must restore the state before a successful commit is reported. A failed
repair invalidates the run.
\item \emph{Proposal-only recycling.} Alignment changes only the candidate
cursor, never published tokens, committed target tokens, or target state.
Every reused token is verified again under the actual current prefix.
The fixed drafter supplies proposals without changing the reference input.
\end{enumerate}
These backend and protocol conditions must be validated for every
admitted configuration.

\begin{proposition}[Proposal-independent exactness]
\label{prop:policy-exactness}
Fix the audio views, state conventions, and greedy tie-breaking of
\(\mathcal S\). If verification is reference-equivalent and restores the
accepted-prefix state, arbitrary source proposals, budgets, and candidate
recycling preserve the target tokens and state at every committed prefix.
Every completed execution then satisfies
\(Y_{\mathcal A}(x;\mathcal S)=Y_T(x;\mathcal S)\).
\end{proposition}

\begin{proof}[Proof of Proposition~\ref{prop:policy-exactness}]
Fix any candidate sequence and budget decisions. We induct on the number of committed tokens. At initialization,
the speculative execution and reference have the same input view and canonical target state.
Assume their tokens and states agree after prefix length \(n\).

For target-only execution, both choose \(g(\kappa_n^T)\) and apply the same
transition \(F\), extending the invariant by one token. For a nonempty
candidate chain, the first candidate is checked from this same state. If
it matches, it equals \(y_{n+1}^T\), and pathwise equivalence gives the same
next state. Repeating this argument covers every consecutively accepted
candidate. At a mismatch, the target correction equals the next reference
token and its canonical processing extends the invariant. If the chain
is exhausted, all committed candidates have passed this check; the next
reference token can be carried under the prescribed known-root convention.
State repair leaves exactly the state of the committed prefix.

Neither cursor realignment nor the choice of fixed drafter changes the
reference transition. A reused suffix, including a wrongly aligned suffix,
can be committed only after a fresh target check; its previous location or
confidence grants no authority. Source publication cannot change the
prescribed audio view or initial state. Induction therefore covers
all successful commits, independently of proposal and budget decisions. If
the response completes under the same termination rule,
its full token sequence equals the reference. Since the argument holds for
every proposal realization, it also holds for randomized proposal generation.
\end{proof}

Applying this invariant to a backend requires validation of pathwise
equivalence and accepted-prefix state restoration
(Appendix~\ref{sec:numerical-sensitivity}).

\FloatBarrier
\section{Implementation and evaluation details}
\label{sec:appendix-evaluation}
\label{sec:evaluation}

This appendix follows Section~\ref{sec:main-evaluation}: experimental setup,
calibration and fixed deployment, replay validation, and the three research
questions. The four-phone ASR study contains 688 full-audio test windows and
94,511 target tokens per phone. All test windows, including capped
outputs and throughput regressions, remain in the reported cohort. Native
failures are retained with their execution boundaries.

\subsection{Experimental setup}
\label{sec:appendix-eval-setup}
\subsubsection{Devices and runtime}
\label{sec:implementation}

\noindent\textbf{Test devices.}\quad
Table~\ref{tab:implementation-devices} summarizes hardware and the
four-phone ASR configurations. Hardware specifications follow the vendor
documentation \citep{qualcomm2026sd750g,qualcomm2026sd888,qualcomm2023sd8gen3,
xiaomi2026mi14specs,xiaomi2026k70specs,xiaomi2026mi11prospecs}; Adreno identifiers are also documented by
\citet{ggml2026opencl}. RAM is the kernel-reported \texttt{MemTotal},
converted to GiB. Peak CPU clocks describe hardware specifications, not
frequencies maintained throughout a run.

\begin{table}[t]
\centering
\small
\setlength{\tabcolsep}{4pt}
\caption{The four phones span different hardware and software configurations.
RAM denotes OS-visible total memory, and CPU clocks are specified maxima.
The lower block lists the fixed models, placement, and proposal budgets
used for the ASR comparison.}
\label{tab:implementation-devices}
\begin{tabularx}{\linewidth}{@{}l*{4}{>{\centering\arraybackslash}X}@{}}
\toprule
 & Redmi K70 Pro & Xiaomi 11 Pro & Xiaomi 14 & Redmi Note 9 Pro \\
\midrule
Model ID & 23117RK66C & M2102K1AC & 23127PN0CC & M2007J17C \\
Snapdragon SoC & 8 Gen 3 & 888 & 8 Gen 3 & 750G \\
Process (nm) & 4 & 5 & 4 & 8 \\
CPU family & Kryo (X4 prime) & Kryo 680 & Kryo (X4 prime) & Kryo 570 \\
CPU cores & 8 & 8 & 8 & 8 \\
Peak CPU (GHz) & 3.3 & 2.84 & 3.3 & 2.2 \\
Adreno GPU & 750 & 660 & 750 & 619 \\
RAM (GiB) & 10.91 & 7.05 & 14.81 & 7.21 \\
DRAM standard$^{*}$ & LPDDR5X & LPDDR5 & LPDDR5X & LPDDR4X \\
\midrule
Android / API & 16 / 36 & 14 / 34 & 16 / 36 & 12 / 31 \\
System build & \shortstack{OS3.0.304.0.\\WNMCNXM} & \shortstack{OS2.0.7.0.\\UKACNXM} & \shortstack{OS3.0.306.0.\\WNCCNXM} & \shortstack{V14.0.2.0.\\SJSCNXM} \\
Security patch & 2026-06-01 & 2025-04-01 & 2026-08-01 & 2023-03-01 \\
OpenCL version & 3.0 & 2.0 & 3.0 & 2.0 \\
Driver build/commit & 0762.36.1 & 329cf4c2a7 & 0762.36.1 & 16c8186230 \\
OpenCL compiler & E031.45.02.25 & E031.38.01.08 & E031.45.02.25 & E031.37.12.04 \\
\midrule
ASR target & \multicolumn{4}{c}{Qwen3-ASR-1.7B, OpenCL GPU} \\
ASR drafter & \multicolumn{4}{c}{Qwen3-ASR-Audio-0.6B, CPU, four threads} \\
Audio frontend & \multicolumn{4}{c}{CPU, four threads} \\
Precision mode$^{\dagger}$ & \multicolumn{4}{c}{Target: \texttt{low}; drafter: \texttt{normal}} \\
Publication batch & \multicolumn{4}{c}{Eight tokens} \\
Proposal budget $b$ & 7 & 7 & 7 & 3 \\
Maximum width $K$ & 8 & 8 & 8 & 4 \\
\bottomrule
\end{tabularx}
\par\smallskip
\begin{minipage}{\linewidth}
\footnotesize
$^{*}$DRAM specifications follow the phone vendor, except Note 9 Pro's
LPDDR4X, which is the chipset's supported interface.
$^{\dagger}$MNN runtime precision settings do not identify weight bit widths.
Note 9 Pro is the Chinese 5G variant. Thermal conditions and repetition
protocols are specified with the calibration and native experiments below.
Software metadata was queried directly from the phones after calibration,
with no intervening system upgrades. Driver identifiers are queried through
\texttt{CL\_DRIVER\_VERSION}.
\end{minipage}
\end{table}

\paragraph{Resident producer and target.}
For ASR, the MNN Android implementation \citep{jiang2020mnn,mnn2026}
runs a CPU source producer alongside a resident OpenCL target.
The four-phone comparison uses
Qwen3-ASR-Audio-0.6B drafting with four CPU workers and a 1.7B target.
Source and target maintain separate model state; candidate tokens pass
through the ordered buffer in Appendix~\ref{sec:design}.

For ASR, the Qwen source consumes the full available audio window once,
then publishes eight-token batches as generation proceeds. External-source
calibration uses chunk-local Parakeet and SenseVoice generation.
Source generation is independent of newly committed target text. CPU
worker placement, resident source pools, and supported widths differ by
phone and are specified with each calibration in Appendix~\ref{sec:appendix-eval-deployment}.

\paragraph{Runtime state and timing.}
\label{sec:native-runtime-optimizations}
The target's logical committed prefix governs attention, positions, and
cache repair. Known-root verification and deferred KV submission avoid an
unconditional extra target call, while target-only progress remains available
when candidates are absent or incompatible. Candidate alignment advances
the source cursor independently of target cache repair.
The implementation contract is detailed in
Appendix~\ref{sec:tree-contract}.

Monotonic stage events record source generation/publication, verification,
correction, queue work, and producer completion. Corpus replay applies
separately calibrated phone costs to saved operation trajectories. Native
case studies use the device clock and retain their own repetition counts,
completion boundaries, and output checks. Every evaluated deployment uses its
fixed drafter and preselected budget, with no request-time configuration inference
(Appendix~\ref{sec:deployment-selection}).

\subsubsection{Datasets, task cohorts, and token accounting}

\noindent\textbf{Combined evaluation scope.}\quad
Across the ASR, translation, and QA cohorts described here,
the corpus evaluation contains 928 requests: 688 ASR windows and 80
requests each for English-to-Chinese translation, Chinese-to-English
translation, and spoken QA. Deduplication by audio hash gives 913 audio
inputs spanning 1.632--75.168s and 43,844.9815s (12.18h) in total.
Shared QA passages make the request-weighted duration 12.40h.
The 97,953 target output tokens exclude EOS and include the first token
selected during prefill; this corpus-size count does not change the
task-specific TPS numerators. Calibration, selection pilots, native cases,
methods, and phone repetitions are excluded from these totals.

\paragraph{ASR evaluation cohort.}
Table~\ref{tab:evaluation-performance} uses 688 audio windows selected by
shuffling within each dataset with seed 20260915, without using model outputs.
The selection contains 165 LibriSpeech, 53 AMI, 91 FLEURS, 188 KeSpeech,
and 191 M4Singer windows, totaling 11.0 hours of audio and 93,824 non-EOS
target tokens. Selection is at the window level, without speaker or
recording grouping. All methods use the same windows and target references,
including one output that reaches the shared generation limit.
Cost calibration and fixed-configuration selection use separate development
audio, as specified below and in Appendix~\ref{sec:deployment-selection}.

\paragraph{Translation and spoken QA protocol.}
\phantomsection
\label{sec:omni-tasks}

Table~\ref{tab:omni-task-results} compares three 80-sample cohorts using
Qwen2.5-Omni-7B as the target and Qwen2.5-Omni-3B as the drafter for all
speculative methods. English-to-Chinese IDs are the first 80 of the frozen
2,733 evaluation IDs after sorting by
SHA256(\nolinkurl{20260923:audio_sha256}); this rule was fixed before the
run without inspecting TPS or acceptance. The Chinese-to-English IDs are
held out from a disjoint 24-sample budget-selection pilot. Spoken SQuAD
uses the length-conditioned subgroup below. All TPS values pool
useful target tokens over post-prefill, resource-complete replay time;
useful tokens exclude EOS and the first token selected during prefill.
Appendix~\ref{sec:baseline-implementations} specifies the serial baselines.

\system uses budget 3 (physical $K\leq4$), fixed from the disjoint
Chinese-to-English pilot and retained for both other tasks. The producer
receives only the requested audio, instruction, and its own generation
history. It never receives target-prefix updates. Source token IDs share
the target tokenizer and become available after their real $K=1$ forwards,
without text conversion. The QA buffer receives cumulative observed
prefixes with unchanged token-arrival times; these prefixes only grow.
Verification and already-started source work after target EOS count in
resource-complete time. Conversion and per-round queue fees are zero;
these profiles do not separately calibrate queue/cache-control overhead
or contention. The reported times are phone-cost surrogate estimates.

\paragraph{QA cohort selection.}
\phantomsection
\label{sec:omni-qa-subgroup}

The Spoken SQuAD column in Table~\ref{tab:omni-task-results} uses a post-hoc
subgroup selected after inspecting the target output lengths of 970 frozen
questions. The rule requires at least 16 useful target tokens, excluding
the first token selected during prefill and EOS. It selects 80 questions
and 1,895 useful tokens. These IDs were fixed before the independent
Omni-3B run, and no additional samples were filtered. The length-conditioned
result does not establish an unselected QA-wide gain.

\subsubsection{Baseline implementations}
\label{sec:baseline-implementations}

\paragraph{ASR SD replay.}
Target-path top-2 ranks are collected with the prefix-conditioned 0.6B
drafter on the server and matched to the frozen audio hashes and canonical
target tokens. Standard SD uses a known target root plus $K-1$ draft tokens
in a physical width-$K$ verification. It advances through the first mismatch,
carries the correction or bonus token as the next known root, and charges
any missing drafter KV forward before the next proposal. Target correction
has no unconditional extra K1 call. Only the common output budget can
reduce the physical width; remaining reference length is never used to
choose a cheaper tail action. EOS contributes to the committed-token count
without an extra forward after it becomes known. Model loading and both
audio/prompt prefills are excluded. Corpus replay applies separately
calibrated phone costs to these cached operation trajectories. Server/phone
numerical-parity checks do not cover the full test cohort.

\paragraph{ASR fixed-width controls and averaging.}
The Standard SD entries in Table~\ref{tab:evaluation-performance} use the
equal-weight arithmetic mean of four fixed configurations, K=5,6,7,8.
For each dataset and K, we first pool committed
tokens and total decode time; we then average the four TPS values. The All
column pools the complete five-dataset cohort within each K before taking
the same mean. This is a configuration summary, not the throughput of a
random switching policy or a validation-selected best fixed setting.
For the CDF and slowdown counts, we instead average the four TPS values
within each window before forming the distribution or comparing with
Target-only (Appendix~\ref{sec:appendix-eval-rq1}).

\paragraph{ASR SpecASR ASP port.}
We implement adaptive single-sequence prediction (ASP) with draft recycling
\citep{wei2025specasr}, using the same pretrained 0.6B drafter and 1.7B target.
No baseline-specific training is performed. The confidence threshold is 0.4
and the maximum is 24 draft tokens, fixed before testing. We interpret
normalized logits as softmax probabilities and include the low-confidence
token before stopping. The known target root makes the physical verifier
width at most 25. Actual off-path drafter tokens and confidences are collected
on the server; target-path ranks alone do not recover these trajectories.

After rejection, a masked two-query drafter forward extends the old branch
while regenerating from the corrected prefix. A match at the same or an
adjacent position permits suffix reuse. Our port then refreshes the reused
suffix's KV state in one batch under the corrected prefix and charges this
batched KV reconciliation. Drafting and target verification remain serial,
with parallel old/new branches internal to the drafter. We evaluate this ASP
port. SpecASR's alternative two-pass sparse-tree (TSP) variant is not evaluated.

\paragraph{Translation and spoken QA baselines.}
Target-only runs the Omni-7B target without drafting. SD and adapted
SpecASR use real server conditional proposals from the same
prefix-conditioned Omni-3B drafter with budget 7 (physical $K\leq8$).
SpecASR uses adaptive single-sequence prediction and rejected-draft recycling with
serial single-token drafter forwards. Its original masked-pair/TSP runtime
is outside this comparison. Their KV-ready replays charge K70 phone costs. Confidence
softmax and drafter KV branch switching are not separately calibrated.
Appendix~\ref{sec:omni-tasks} specifies task cohorts and output-token accounting.

Configuration selection is specified in
Appendix~\ref{sec:deployment-selection}. Each reported deployment uses
its own phone calibration, without borrowing costs from other phones.

\FloatBarrier

\subsection{Calibration, deployment configuration, and time analysis}
\label{sec:appendix-eval-deployment}
This section records the phone costs and fixed configurations used by the
experiments. The final analysis relates realized overlap and verifier work
to elapsed time; acceptance is an analysis variable, not a configuration
selection criterion.

\subsubsection{Feasible widths and coexistence costs}
\label{sec:action-pruning}
Only measured, executable widths enter a deployment's cost catalog.
Missing or failed shapes are excluded instead of assigned zero cost.
A proposal budget and a physical backend width are different quantities:
including the known root, budget 7 permits physical K8. Large-width
concurrent costs that use a fitted ratio rather than direct measurements
are identified in the corresponding protocol and sensitivity analysis.

Calibration loads the declared target and source models before timing.
Coexistence can change both producer and target service, so isolated
forward latency alone cannot determine placement. Verification and
correction must also preserve ownership of the same target KV. These
constraints apply to both the ASR CPU-producer/OpenCL-target placement
and the native Omni placement in Appendix~\ref{sec:appendix-eval-rq3}.

\subsubsection{Redmi K70 Pro calibration}
\label{sec:k70-baseline-replay}

\paragraph{K70 ASR cost calibration.}
Both the 0.6B CPU drafter (four threads, normal precision) and the 1.7B OpenCL
target (low precision) were resident in a paired K70 calibration on one
59.5-second calibration window. One warmup and five measured pairs produced
identical 239-token outputs through EOS. Native target-only and serial K8
throughput were 16.91 and 20.86 tokens/s on that calibration window. These
measurements calibrate the replay costs. The baseline cost model uses
59.391\,ms per target-only step,
23.347\,ms per draft forward, 112.656\,ms per K8 verification, and
10.738\,ms of other decode-loop work per serial verification round. The
residual loop cost is measured wall time minus drafting and verification
time. Smaller K costs come from a separate K70 K1--16 calibration-context
sweep at prefix offsets 0, 64, and 128.
No language, long-context, or thermal slope is fitted.

\paragraph{ASR ASP cost extension and check.}
The single-step drafter, K1--K8 verifier, and loop-residual constants are
unchanged from the other Table~\ref{tab:evaluation-performance} baselines.
Additional K9--K25 verification and batch-refresh costs use K70 profile
medians: one warmup and five measurements at each of three text offsets.
True paired-branch drafting costs 35.075\,ms per forward, before a measured
0.951\,ms of probability computation per output row; batch KV refresh is
charged by its physical width. All 90 checks against independent branch
execution pass. This accounting yields 11,136 verification rounds over the
complete 688-window test split.

A separate long development-window ASP diagnostic reaches EOS in one warmup and
five measured repetitions. It emits 240 tokens, including one additional
comma relative to scalar AR's 239 tokens. On the recorded native operation
trace, replay underestimates latency by 2.6--14.7\% across the five repetitions,
with 5.6\% error at the median. The fixed costs do not capture the observed
run-to-run slowdown.

\paragraph{Native overlap calibration and checks.}
One calibration audio uses one warmup and five paired measured repetitions,
with alternating execution order. Serial budget 7 achieves 15.41 pooled
TPS and target-decoupled overlap 29.81 TPS, a 93.4\% increase. Every run reaches
EOS and matches the native scalar-target output; an independent queue
reconstruction checks all 503 source-overlap verification rounds, including warmup.
These native repetitions supply the service-cost fit.

Five additional, predeclared development windows are excluded from cost
fitting. Three each yield one complete, output-matched measured pair after
one warmup:
serial/TD throughput is 14.78/24.20 TPS on FLEURS, 9.62/15.80 on KeSpeech,
and 11.85/12.52 on M4Singer. Replaying their native-produced candidates
with the frozen calibration costs gives signed relative TD-time errors of
$+0.6\%$, $+2.0\%$, and $-8.5\%$, respectively. The other two windows do
not yield measured pairs: the AMI reference reaches its 1,024-token cap,
and the LibriSpeech warmup emits two additional commas. These diagnostic
outcomes do not remove any corpus-test windows.
Appendix~\ref{sec:replay-agreement} reports the subsequent repeated timing
validation.

\subsubsection{Xiaomi 11 Pro calibration}
\label{sec:xiaomi11-transfer}

The Xiaomi 11 Pro (SM8350) evaluation uses the same 1.7B OpenCL target,
Qwen3-ASR 0.6B CPU drafter, and proposal budget 7 as the K70 deployment.
Standard SD averages fixed physical K=5,6,7,8 TPS
(Table~\ref{tab:sd-fixed-mean}). SpecASR uses ASP with threshold 0.4 and
at most 14 proposals, or physical $K\leq15$. All four methods retain the
same 688 windows and 94,511 reference tokens, including the capped FLEURS
output. The two capped Qwen source sequences are also retained.

\paragraph{Common clock and hardware condition.}
The measured interval starts after model loading and both initial
audio/prompt prefills. It includes draft computation, exposed waits,
verification, queue/loop work, and producer completion after target EOS.
Target-only calibration loads only the target. Serial SD and ASP load the
target and Qwen; source-selection calibration keeps the target, Qwen and
Parakeet resident to compare source choices without model replacement.
The four CPU drafting workers use cores 4--7. Before each measured episode,
we require all eight cores online and CPU4's frequency ceiling at least
1.8816\,GHz at two checks five seconds apart, after at least 15 seconds of
rest. A native monitor samples core availability and frequency ceilings
once per second. Within-run frequency changes are retained; core offlining
stops the batch and leaves the failure record intact. This protocol measures
decode episodes from a common starting condition, not continuous-load throughput.

\paragraph{Cost calibration.}
Two English calibration windows, from LibriSpeech and AMI, provide five
measured repetitions each of target-only, Qwen overlap and Parakeet overlap
with the common resident pool, plus five standalone-target repetitions per
window. All 40 measured runs reach EOS with identical target tokens.
Per-width K1--8 profiles and these episode traces calibrate the shared
verifier costs and overlap-specific concurrent costs; observed concurrency
ratios are retained without clipping them to one.

Additional calibration on the 59.5-second LibriSpeech calibration window uses
one warmup and
five alternating measured repetitions of serial K8 SD and ASP with cap 14.
Scalar drafting and the residual decode-loop cost are shared by SD and ASP;
ASP additionally charges confidence computation, paired-branch execution
and batch KV refresh. Missing verifier widths K9--15 receive five measured
runs each; CPU refresh widths K1--15 receive ten each, and paired-branch
cost uses 20 measured forwards. Component outputs are checked for physical
width $K$ and logits shape $[1,K,151936]$. All paired-branch checks match
independent branch execution. Arithmetic means retain every formal timing;
warmups are excluded. A previous K16--20 calibration stalled, so widths
above 15 remain pruned without assigning an unverified failure cause.
All timing parameters are fitted on separate calibration audio. They do not model every
language, context length, or thermal state.

The serial calibration check reaches 11.60 TPS for K8 SD and 11.56 TPS for
ASP; all ten measured outputs match the scalar reference exactly.
Repricing these recorded native operation traces gives latency residuals
from -1.7\% to +3.0\%. The corresponding Qwen/Parakeet overlap
residuals on the two calibration windows range from $-0.89\%$ to $+1.62\%$.
These residuals characterize the in-sample calibration fit. The corpus
comparison applies the calibrated phone costs to saved operation trajectories;
Appendix~\ref{sec:replay-agreement} reports separate native timing checks on K70.

\subsubsection{Xiaomi 14 and Redmi Note 9 Pro calibration}
\label{sec:additional-phone-calibration}

\paragraph{Device-specific calibration.}
We calibrate Xiaomi 14 and Redmi Note 9 Pro separately using the same fixed LibriSpeech
and AMI calibration windows. Each context supplies one warmup and five measured
repetitions of resident-target decoding, Qwen overlap, and standalone-target
decoding. Target-only loads only the target; speculative runs retain target
and Qwen. Model loading and initial audio/prompt prefill precede the measured
interval. Source waits, verification, correction/alignment, loop work, and
post-EOS producer completion remain charged. K70 and Xiaomi 11 Pro retain
their previously documented residency protocols.

Xiaomi 14 pins the four CPU drafting workers to cores 4--7. On Redmi Note 9
Pro, explicit worker-affinity setup failed after loading and prefill, before
formal measurements. We retain that failure and use four MNN workers under
OS-default scheduling for every measured method on that phone. The baseline
frontend only skips explicit placement when no CPU IDs are configured;
inference algorithms are unchanged. Each episode requires all eight cores
online, at least 15 seconds of rest, and two condition checks five seconds
apart. Within-run frequency changes are retained. We do not modify thermal
control or CPU governors.

Dense profiles measure physical K1--8 on both calibration contexts; observed
overlap rounds provide concurrent costs. Widths with fewer than five
concurrent observations use the median eligible calibration concurrent/idle
ratio, without clipping. Serial K8 SD and ASP with cap 14 receive one warmup
and five measured repetitions on the LibriSpeech window. Scalar drafting
and loop costs are shared; ASP additionally charges confidence computation,
paired-branch execution, and KV refresh. Verifier K9--15 receives five
measurements per width, refresh K1--15 receives ten, and paired execution
receives 20. Component checks enforce physical width and logits shape
$[1,K,151936]$. The two new phones use the same conservative ASP cap 14 as
Xiaomi 11 Pro; larger shapes are outside this transfer experiment.

\paragraph{Calibration diagnostics.}
On Xiaomi 14, 30/30 formal target/overlap outputs and 10/10
serial outputs match their native scalar reference through EOS.
Repricing the observed native target/source paths gives in-sample overlap
latency residuals of -3.33\% to -0.94\%; serial residuals are
-9.83\% to +16.61\%. All formal repetitions are retained.

On Redmi Note 9 Pro, 30/30 formal target/overlap outputs and 10/10
serial outputs match their native scalar reference through EOS.
Repricing the observed native target/source paths gives in-sample overlap
latency residuals of -2.34\% to -0.31\%; serial residuals are
-6.33\% to +3.38\%. All formal repetitions are retained.
These residuals characterize the in-sample calibration fit on two English
windows. The corpus comparison applies the resulting phone costs to preserved
causal operation traces, with fixed costs across languages, context lengths,
and thermal states. Appendix~\ref{sec:replay-agreement} reports separate native
timing checks on K70.

\subsubsection{Resident-source cost calibration}
\label{sec:resident-source-calibration}
For Xiaomi 14 and Redmi Note 9 Pro, a separate calibration retains the
target, Qwen3-ASR 0.6B, Parakeet CTC, and SenseVoice Small in memory.
Two fixed calibration windows, from LibriSpeech and AMI, are run with
target-only, Qwen, Parakeet, and SenseVoice arms in rotating/reversed order.
Each arm has one warmup and five measured repetitions per audio,
giving 40 measured episodes per phone. This all-resident calibration
is distinct from the target/Qwen comparisons in
Appendix~\ref{sec:additional-phone-calibration}.

External-source chunk service, publication, and startup are measured on
the phone. Concurrent/idle target ratios from the same session are applied
to the device-specific idle profile, retaining the original Qwen and
baseline cost anchors. Idle physical K9--15 costs are directly measured;
their concurrent costs extend the observed calibration ratios. Width-specific
profiles use the median available ratio at unmeasured concurrent widths,
while constant-factor profiles retain their factor. These estimates do not
represent direct measurements of every source/width pair. The English
calibration uses a fixed padded chunk shape; other supported languages
remain an extrapolation.

All measured repetitions are retained. Native outputs match the device's
scalar reference in 40/40 Xiaomi 14 episodes and 35/40 Note 9 Pro episodes.
The remaining five rows contribute cost measurements but do not support
native output equivalence.

\subsubsection{Fixed configurations and development splits}
\label{sec:deployment-selection}
The four-phone ASR corpus comparison uses a single fixed Qwen drafter and
proposal budget per device: budget 7 on K70, Xiaomi 11 Pro, and Xiaomi 14; budget 3 on
Redmi Note 9 Pro. The selected pair remains unchanged across test requests;
configuration selection incurs no request-time feature extraction or inference.
The studies below use pooled throughput to compare configurations before
their fixed evaluation, rather than selecting a budget from runtime
acceptance statistics.

On K70, the initial budget sweep compares integer proposal limits 1--24
on 50 balanced development windows and selects budget 7. The setting is retained
on a separate validation set of 100 windows. A balanced 50-window development
pilot screens Qwen3-ASR 0.6B, Parakeet CTC, and SenseVoice Small as the
source pool; adding bilingual Zipformer gives no further hindsight gain
on that pilot. The retained K70, Xiaomi 11 Pro, and Xiaomi 14 deployments
use Qwen/budget 7. These preserved development studies do not constitute
a full validation-grid search on every phone. The larger catalogs in Appendix~\ref{sec:joint-oracle-gap-protocol}
serve a separate purpose: measuring the headroom of the deployed settings.

The Note 9 result is a retrospective reanalysis of the preserved split and
phone profiles. We rank all 42 source/budget settings on the 100 validation
windows by pooled throughput, using Qwen at the same budget when a source
is language-ineligible. Validation selects fixed Qwen/budget 3 at 5.6233 TPS,
versus 5.3898 for Qwen/budget 7. Applying this setting to all 688 test
windows yields 5.8512 TPS. Test outcomes do not enter the ranking;
the reanalysis does not imply that a new prospective experiment preceded
inspection of the historical test results.

\subsubsection{Overlap, added work, and performance limits}
\label{sec:performance-derivation}
We derive a finite-execution condition under which target-decoupled overlap
compensates for additional drafting and verification work. The decomposition
uses realized costs and token progress, whereas SSD analyzes expected speed
from verification-outcome prediction \citep{kumar2026ssd}.

\paragraph{Execution boundary and accounting.}
Compare a serial baseline \(s\) and asynchronous execution \(a\) that
complete the same input under protocol \(\mathcal S\), with the same
\(N>0\) committed target tokens. Use matching start and completion events,
including producer drain if charged by the protocol. For post-prefill decode
comparisons, the clock starts after initial audio/prompt prefill. A
request-latency comparison instead starts at request arrival and includes
any remaining audio or request prefill. Both runs must satisfy the
target-equivalence conditions in Appendix~\ref{sec:exactness-proof}.

For run \(r\in\{s,a\}\), let \(\mathcal P_r\) and \(\mathcal V_r\) be the
unions of producer and verifier service intervals within \([0,T_r]\).
Producer service includes source preparation, generation, token conversion,
and publication when inside the measured boundary. Verifier service
includes target-only progress, speculative checks, correction, state
restoration, and candidate alignment. Use completed stage intervals,
not worker lifetimes: a worker blocked waiting for audio or the other
worker is not executing a stage. Internal intervals of the same stage
are merged before their durations are summed. These are elapsed service
times, not processor-utilization measures.

With \(|\cdot|\) denoting interval duration, define
\begin{equation}
 \begin{aligned}
 D_r&=|\mathcal P_r|, & V_r&=|\mathcal V_r|,\\
 O_r&=|\mathcal P_r\cap\mathcal V_r|, &
 H_r&=|[0,T_r]\setminus(\mathcal P_r\cup\mathcal V_r)|.
 \end{aligned}
 \label{eq:performance-intervals}
\end{equation}
The residual \(H_r\) includes remaining queue/synchronization work and
idle time outside both stage sets. Work already inside a service interval
is not charged again to \(H_r\). Costs come from the actual execution:
resource contention can increase \(D_a\) and \(V_a\), so isolated-device
profiles cannot simply be substituted while holding these costs fixed.
Startup and drain remain on the clock; only their actual intersections
with the other stage contribute to \(O_r\).

\paragraph{Derivation of the speedup condition.}
Inclusion--exclusion gives the exact accounting identity
\begin{equation}
 T_r=D_r+V_r-O_r+H_r.
 \label{eq:performance-accounting}
\end{equation}
The serial comparator has \(O_s=0\). Therefore, with
\(\Delta X=X_a-X_s\),
\begin{equation}
 T_s-T_a=O_a-\Delta D-\Delta V-\Delta H.
 \label{eq:performance-time-saving}
\end{equation}
Both runs commit \(N\) tokens, so their throughput ratio is
\((N/T_a)/(N/T_s)=T_s/T_a\). A strict throughput gain therefore occurs
exactly when
\begin{equation}
 T_a<T_s
 \quad\Longleftrightarrow\quad
 O_a>\Delta D+\Delta V+\Delta H.
 \label{eq:overlap-break-even}
\end{equation}
The increments may be negative: for example, a cheaper target-decoupled producer can reduce
\(D_a\), so a measured gain need not arise entirely from overlap.
Against an already asynchronous comparator, replace \(O_a\) in the
condition by \(O_a-O_s\).

\paragraph{Acceptance and verifier work.}
Let \(R_r\) be the number of sequential verifier rounds, including
target-only rounds. If round \(k\) commits \(m_{r,k}\) tokens and incurs
verifier service time \(v_{r,k}\), then
\(\sum_k m_{r,k}=N\) and \(\sum_k v_{r,k}=V_r\). With the arithmetic
round means \(\bar m_r=N/R_r\) and \(\bar v_r=V_r/R_r\),
\begin{equation}
 \frac{V_r}{N}=\frac{\bar v_r}{\bar m_r},\qquad
 \frac{\Delta V}{N}
 =\frac{\bar v_a}{\bar m_a}-\frac{\bar v_s}{\bar m_s}.
 \label{eq:performance-progress}
\end{equation}
Committed progress includes correction tokens as well as accepted
candidates. Lower progress per round increases verifier work per output
token when round cost is unchanged, but raw acceptance rate alone does
not determine either quantity. Budget, alignment, target-only fallback,
and backend shape change round cost and progress together. This is why
a faster producer with lower acceptance can still win, and why increasing
the proposal budget need not improve throughput.

\paragraph{Bottleneck bound and interpretation.}
Since \(0\leq O_a\leq\min(D_a,V_a)\),
\begin{equation}
 T_a\geq\max(D_a,V_a)+H_a,\qquad
 \frac{T_s}{T_a}\leq
 \frac{T_s}{\max(D_a,V_a)+H_a}.
 \label{eq:performance-bottleneck}
\end{equation}
For fixed realized costs, no overlap can yield a strict speedup if
\(\Delta D+\Delta V+\Delta H\geq\min(D_a,V_a)\).
The bound need not be attainable: request-time source preparation, startup, drain,
and candidate readiness can all reduce overlap. When the verifier is
already continuously supplied, faster drafting alone does not lower
its service requirement; further gains must reduce verifier work,
exposed startup/drain, or residual time. Conversely, slow supply can
leave overlap unrealized even when proposal agreement is high.

The decomposition provides a framework for interpreting the native cases
in Section~\ref{sec:main-case-studies}. Estimating its terms requires
synchronized stage traces under the matched execution boundary and
concurrent load. Predicting performance additionally requires estimates of
stage costs, committed progress, and candidate readiness under that load.

\FloatBarrier

\subsection{Replay agreement with native timing}
\label{sec:replay-agreement}

We compare replay with native target-decoupled overlap on Redmi K70 Pro,
using a Qwen3-ASR-1.7B OpenCL target, a Qwen3-ASR-0.6B drafter on four CPU
threads, and proposal budget 7. To isolate timing error, replay uses the
same phone-generated target and candidate token sequences as native
execution. The measured interval starts after prefill and includes producer
completion. Signed TPS error is $100(\mathrm{TPS}_{\mathrm{replay}}/
\mathrm{TPS}_{\mathrm{native}}-1)$. Mean absolute error averages its magnitude
over windows.

\noindent\textbf{Validation protocol.}\quad
For each LibriSpeech window, separate training audio is cropped or
zero-padded to the same duration for calibration, matching the initial KV
length. Calibration immediately precedes validation, and the resulting
costs are frozen before observing validation timings. Each calibration
mode has one warm-up and two measured repetitions. Target-only and
\system then each run one warm-up and five measured repetitions, with
their order alternating across repetitions. We use the median native
\system latency. The initial check completes nine windows; a tenth stops
during calibration before native validation. After inspecting the initial
results, we select the first five windows for fresh calibration and five
new measured repetitions per mode. All 90 initial and 50 repeat requests
match the target reference and reach EOS. Candidate sequences are complete
and identical across repetitions within each window, and all completed
windows enter the timing comparison.

\noindent\textbf{Timing agreement.}\quad
Table~\ref{tab:replay-agreement} summarizes the full initial check and the
five-window repeat. Mean absolute TPS errors are $5.81\%$ and $6.54\%$,
respectively.

\begin{minipage}{\linewidth}
\centering
\small
\setlength{\tabcolsep}{8pt}
\begin{tabular}{@{}lrrr@{}}
\toprule
Check & Windows & Mean absolute error (\%) & Pooled error (\%) \\
\midrule
Initial validation & 9 & 5.81 & $-4.58$ \\
Fresh repeat & 5 & 6.54 & $-4.95$ \\
\bottomrule
\end{tabular}
\captionof{table}{The table summarizes replay TPS errors in the initial
check and fresh repeat.}
\label{tab:replay-agreement}
\end{minipage}

\noindent\textbf{Earlier block-calibration check.}\quad
A separate 16-window experiment calibrates costs before each four-window
block, with one warm-up and two measured repetitions of Target-only,
serial speculation, and overlap per window (96 measured requests).
All 32 measured overlap outputs match the target reference. Twelve
windows also have complete, identical source sequences across the two
repetitions and qualify for a same-trajectory timing comparison. Four
have incomplete source traces, including one source-sequence mismatch.
Using the second measured native repetition, the 12 eligible pairs have
$8.43\%$ mean absolute TPS error and
$-9.39\%$ pooled error, retaining all eligible pairs.

\noindent\textbf{Measurement scope.}\quad
These checks condition on phone token trajectories. The headline corpus
replay uses server-generated traces. The context-matched profiles
are used only for timing validation and do not replace the headline
corpus profiles.

\FloatBarrier

\subsection{RQ1: Performance across devices and tasks}
\label{sec:appendix-eval-rq1}
\subsubsection{Four-phone throughput distributions and complete results}
\label{sec:fourphone-transfer}

Figure~\ref{fig:fourphone-tps-cdf} evaluates the complete-audio cohort
on Xiaomi 14 (23127PN0CC) and Redmi Note 9 Pro (M2007J17C).
Each phone uses the same 688 held-out
windows and 94,511 target tokens; the one capped target output and two
capped Qwen source sequences remain included. Fixed deployment uses
Qwen/budget 7 on Xiaomi 14 and Qwen/budget 3 on Note 9 Pro, with eight-token
source publication and canonical target verification/alignment.
The Note 9 reanalysis ranks configurations only on the preserved validation
split, then evaluates the chosen fixed setting on the test set
(Appendix~\ref{sec:deployment-selection}). No per-window selector runs. The K70 and Xiaomi 11 Pro
values reproduce Table~\ref{tab:evaluation-performance} exactly.

Table~\ref{tab:evaluation-performance} includes both additional phones and matched component removals.
Libri, Flrs, KeS, and M4S abbreviate LibriSpeech, FLEURS, KeSpeech,
and M4Singer, respectively.

\begin{table}[!t]
\centering
\caption{\textbf{\system improves pooled throughput over the baselines on all four phones.}
Values report phone-profile replay TPS on the same 688 windows. Standard SD
pools tokens and time for each fixed K=5,6,7,8, then averages the four TPS
values. All includes all five datasets. Bold identifies \system, and
underlines mark the best baseline per phone and dataset. Gain is the
relative TPS increase over that baseline, excluding ablations. Xiaomi 14
retains Qwen/budget 7, while Note 9 Pro uses validation-selected fixed
Qwen/budget 3. Matched without-TD traces cover all 688 windows.
Section~\ref{sec:main-ablation} specifies the settings, and
Appendix~\ref{sec:appendix-eval-deployment} describes calibration.}
\label{tab:evaluation-performance}
\begingroup
\fontsize{9}{10.5}\selectfont
\setlength{\tabcolsep}{2pt}
\renewcommand{\arraystretch}{1.12}
\begin{tabular*}{\linewidth}{@{\extracolsep{\fill}}lrrrrrr@{\hspace{8pt}}rrrrrr@{}}
\toprule
& \multicolumn{6}{c}{Redmi K70 Pro} & \multicolumn{6}{c}{Xiaomi 11 Pro} \\
\cmidrule(lr){2-7}\cmidrule(l){8-13}
Method & Libri & AMI & Flrs & KeS & M4S & All & Libri & AMI & Flrs & KeS & M4S & All \\
\midrule
Target-only & 16.92 & \underline{16.92} & 16.94 & 17.00 & 17.04 & 16.96 & 8.33 & 8.33 & 8.34 & 8.37 & 8.39 & 8.35 \\
Standard SD & \underline{21.52} & 15.46 & \underline{19.91} & \underline{17.63} & \underline{18.26} & \underline{18.97} & 11.77 & \underline{8.39} & 10.87 & \underline{9.59} & \underline{9.94} & \underline{10.34} \\
SpecASR & 21.44 & 11.19 & 18.35 & 15.44 & 15.61 & 16.81 & \underline{12.46} & 7.32 & \underline{10.97} & 9.37 & 9.64 & 10.21 \\
\midrule
\system & \textbf{35.37} & \textbf{18.60} & \textbf{28.22} & \textbf{27.57} & \textbf{25.59} & \textbf{27.92} & \textbf{18.94} & \textbf{9.82} & \textbf{14.98} & \textbf{14.24} & \textbf{13.22} & \textbf{14.68} \\
Gain (\%) & +64.4 & +9.9 & +41.7 & +56.4 & +40.2 & +47.2 & +52.0 & +17.0 & +36.6 & +48.4 & +33.0 & +41.9 \\
\midrule
\quad w/o BC & 33.60 & 15.86 & 26.84 & 25.53 & 22.91 & 25.60 & 18.21 & 8.89 & 14.42 & 13.37 & 12.22 & 13.79 \\
\quad w/o TD & 22.68 & 15.61 & 20.88 & 18.42 & 18.93 & 19.77 & 12.59 & 8.74 & 11.62 & 10.25 & 10.52 & 11.00 \\
\midrule
& \multicolumn{6}{c}{Xiaomi 14} & \multicolumn{6}{c}{Redmi Note 9 Pro} \\
\cmidrule(lr){2-7}\cmidrule(l){8-13}
Method & Libri & AMI & Flrs & KeS & M4S & All & Libri & AMI & Flrs & KeS & M4S & All \\
\midrule
Target-only & \underline{19.26} & \underline{19.26} & \underline{19.28} & \underline{19.35} & \underline{19.39} & \underline{19.30} & \underline{4.11} & \underline{4.11} & \underline{4.12} & \underline{4.13} & \underline{4.14} & \underline{4.12} \\
Standard SD & 15.86 & 11.56 & 14.73 & 13.11 & 13.57 & 14.06 & 3.78 & 2.70 & 3.49 & 3.08 & 3.19 & 3.32 \\
SpecASR & 16.08 & 9.91 & 14.38 & 12.64 & 12.97 & 13.55 & 3.92 & 2.39 & 3.49 & 3.01 & 3.09 & 3.26 \\
\midrule
\system & \textbf{30.99} & \textbf{22.77} & \textbf{28.79} & \textbf{28.63} & \textbf{27.54} & \textbf{28.40} & \textbf{6.72} & \textbf{4.57} & \textbf{5.94} & \textbf{5.74} & \textbf{5.49} & \textbf{5.85} \\
Gain (\%) & +60.9 & +18.2 & +49.3 & +48.0 & +42.1 & +47.1 & +63.4 & +11.1 & +44.2 & +39.0 & +32.5 & +42.0 \\
\midrule
\quad w/o BC & 30.84 & 21.25 & 28.55 & 28.23 & 26.91 & 27.84 & 6.21 & 3.54 & 5.42 & 4.94 & 4.55 & 5.07 \\
\quad w/o TD & 16.35 & 11.77 & 15.24 & 13.72 & 14.06 & 14.55 & 4.10 & 3.42 & 3.94 & 3.69 & 3.76 & 3.84 \\
\bottomrule
\end{tabular*}
\endgroup
\end{table}

\paragraph{CDF and pooled TPS.}
Every complete window receives equal CDF weight, without smoothing,
selection of a test-specific best K, or removal of distribution tails.
Target-only and \system use each window's committed target tokens
divided by its decoding time. The Standard SD curve averages that window's
four TPS values at physical K=5,6,7,8. The tables instead pool all tokens and
time separately for each K before averaging the four throughputs; the
aggregation orders need not agree. Target-only curves are nearly vertical
because the calibrated model charges one mean cost per target forward.
They are not measured distributions of native per-window wall time.

\paragraph{Xiaomi 11 Pro result and fixed configuration.}
The development study retains Qwen at budget 7. This pair remains fixed
across all test windows (Appendix~\ref{sec:deployment-selection}).

The frozen deployment reaches 14.68 pooled TPS, improving throughput by
75.8\% over target-only and 41.9\% over the strongest listed baseline (Standard SD).
The per-dataset improvement over the strongest baseline ranges from
+17.0\% to +52.0\%. These results use the calibrated chain configuration
with target-decoupled overlap.

\subsubsection{Fixed-width SD controls}
\label{sec:appendix-eval-sd-controls}

Table~\ref{tab:sd-fixed-mean} retains each configuration and its range for
both phones. On Xiaomi 11 Pro, all methods use the current calibration protocol in
Appendix~\ref{sec:xiaomi11-transfer}. Five native serial calibration runs
calibrate scalar drafting and loop work; the K1--8 verifier profile is
shared with the source-overlap replay. The fixed-SD mean is 10.34 TPS,
with pooled rates from 9.71 to 10.70 TPS across the four widths.

The draft-step and per-round loop constants are shared across widths;
separate corpus-wide native calibration at each K is not claimed.
Table~\ref{tab:k70-sd-k-sensitivity} also retains K8 and K12 controls; K12
verification costs 142.825\,ms on K70. Fixed-budget SD continues proposing to its
budget even if the drafter emits EOS; target-committed EOS terminates output.
ASP instead may stop at draft EOS. The fixed-SD mean and ASP's maximum of
24 draft tokens on K70 are distinct operating points, not a matched-length
algorithm ablation.

\begin{table}[!t]
\centering
\caption{Standard SD throughput varies with fixed verification width on the same
688 windows. Physical $K$ includes one known root. For each fixed
configuration, TPS pools tokens and decode time. Mean denotes the equal-weight
arithmetic mean of TPS at K=5,6,7,8. Range shows variation across these
configurations, not run-to-run uncertainty. Both phones use the same averaging rule.}
\label{tab:sd-fixed-mean}
\small
\setlength{\tabcolsep}{3pt}
\begin{tabular*}{\linewidth}{@{\extracolsep{\fill}}lrrrrrr@{}}
\toprule
Setting & LibriSpeech & AMI & FLEURS & KeSpeech & M4Singer & All five \\
\midrule
\multicolumn{7}{l}{\textbf{Redmi K70 Pro}} \\
$K=5$ & 19.77 & 15.56 & 18.67 & 17.15 & 17.62 & 18.09 \\
$K=6$ & 21.41 & 15.82 & 19.87 & 17.90 & 18.43 & 19.10 \\
$K=7$ & 21.97 & 15.28 & 20.19 & 17.65 & 18.35 & 19.13 \\
$K=8$ & 22.93 & 15.19 & 20.90 & 17.81 & 18.62 & 19.55 \\
Mean & \textbf{21.52} & \textbf{15.46} & \textbf{19.91} & \textbf{17.63} & \textbf{18.26} & \textbf{18.97} \\
Range & 19.77--22.93 & 15.19--15.82 & 18.67--20.90 & 17.15--17.90 & 17.62--18.62 & 18.09--19.55 \\
\midrule
\multicolumn{7}{l}{\textbf{Xiaomi 11 Pro}} \\
$K=5$ & 10.64 & 8.31 & 10.03 & 9.19 & 9.45 & 9.71 \\
$K=6$ & 11.56 & 8.48 & 10.71 & 9.62 & 9.91 & 10.28 \\
$K=7$ & 12.29 & 8.50 & 11.27 & 9.83 & 10.23 & 10.67 \\
$K=8$ & 12.59 & 8.28 & 11.45 & 9.73 & 10.18 & 10.70 \\
Mean & \textbf{11.77} & \textbf{8.39} & \textbf{10.87} & \textbf{9.59} & \textbf{9.94} & \textbf{10.34} \\
Range & 10.64--12.59 & 8.28--8.50 & 10.03--11.45 & 9.19--9.83 & 9.45--10.23 & 9.71--10.70 \\
\bottomrule
\end{tabular*}
\end{table}

\begin{table}[!t]
\centering
\caption{Fixed-width SD at K8 outperforms K12 on all five datasets in the K70
replay cohort. $K$ includes one known root. These controls are separate
from the main-table K5--8 mean. Neither configuration is claimed to be
a validation-selected optimum.}
\label{tab:k70-sd-k-sensitivity}
\small
\begin{tabular*}{\linewidth}{@{\extracolsep{\fill}}lrrrrrr@{}}
\toprule
Fixed SD & LibriSpeech & AMI & FLEURS & KeSpeech & M4Singer & All five \\
\midrule
Standard SD (K8) & 22.93 & 15.19 & 20.90 & 17.81 & 18.62 & 19.55 \\
Standard SD (K12) & 22.79 & 12.45 & 19.53 & 15.55 & 16.52 & 17.76 \\
\bottomrule
\end{tabular*}
\end{table}

On K70, increasing K8 to K12 reduces verification rounds from 16,148 to
12,681, but increases draft forwards from 121,724 to 144,475 (18.7\%).
The replayed verification totals barely change (1,819 versus 1,811 seconds),
while drafting increases from 2,842 to 3,373 seconds. K12 therefore has
lower complete SD throughput despite using fewer target calls. The K5--8
mean is 3.0\% below K8 and 6.8\% above K12. Verifier throughput alone does
not determine the best serial SD width.

\subsubsection{Selection headroom and fixed controls}
\label{sec:joint-oracle-gap-protocol}

Table~\ref{tab:joint-oracle-gap} covers
688 complete windows and 94,511 target tokens per phone, including the
capped output. K70 searches Qwen3-ASR 0.6B, Parakeet CTC, and SenseVoice Small
at every integer proposal budget from 1 to 24. Xiaomi 11 Pro searches Qwen
and Parakeet at budgets 1--14. Xiaomi 14 and Redmi Note 9 Pro search all
three sources at budgets 1--14, supported by their native calibration profiles.
Language-ineligible sources remain masked. A budget $b$ permits at most $b$
proposals plus one known target root: physical verifier width is at most
$b+1$. Thus budget 7 belongs to physical K8, not physical K7.

\paragraph{Matched replay and objective.}
Every feasible source/budget pair uses preserved source candidates, causal
publication, target references, phone costs, confidence threshold 0.4 where
available, alignment, and correction. External-source confidence remains
missing rather than being imputed. Exposed drafting, verification, queue
work, concurrent slowdown, and producer drain use the same post-prefill
resource clock as the corresponding deployed result. The search covers
94,540 distinct phone/window/source/budget combinations. All rounds
are audited for token commits, readiness, physical width, and cost accounting.
3,862 budget-7 checks reproduce the prior source/window times at
$10^{-12}$ relative tolerance; a known-root boundary check reaches physical K25
and rejects missing cost entries.

For each window, the oracle minimizes resource-completion time over the
feasible source and fixed-cap pairs. It has hindsight and pays zero selection
cost. Selected and oracle TPS each pool tokens over total time, and
\[
\mathrm{Gap}=100\left(1-
\frac{\mathrm{TPS}_{\rm selected}}{\mathrm{TPS}_{\rm oracle}}\right).
\]
Deployment configurations are selected without test outcomes. Three phones
retain Qwen/budget 7; the Note 9 reanalysis selects fixed Qwen/budget 3 using
the original 100 validation windows, with no runtime configuration selection
(Appendix~\ref{sec:deployment-selection}).
Test outcomes provide only the oracle and fixed-configuration diagnostics.
This oracle chooses one source and cap per complete window, not a globally
optimal within-window source, tree, or budget schedule.

\paragraph{Results and fixed controls.}
Selected/oracle TPS is 27.92/28.54, 14.68/15.08, 28.40/29.03, 5.85/5.96
on K70, Xiaomi 11 Pro, Xiaomi 14, and Redmi Note 9 Pro, respectively.
A strong pooled fixed configuration can coexist with different per-dataset
and per-window choices. For fixed-source controls, an unsupported language
falls back to Qwen at the same budget.
Small TPS gaps need not imply identical actions on most windows: many
alternative choices can have similar costs.

\paragraph{Large-width cost sensitivity.}
Idle verifier costs cover all searched physical widths. Concurrent costs
above physical K8 are modeled rather than directly measured for every
source/width pair. Qwen retains the existing calibration-only ablation ratio
rule. External sources with width-specific profiles extend the median
available concurrent/idle ratio; observed ratio extremes provide sensitivity
bounds. Constant-factor profiles retain that factor. Source computation,
publication, and physical K1--8 costs remain fixed in this sensitivity check.

The full search is repeated under nominal, low-ratio, and high-ratio
scenarios. The pooled gaps span 1.94--2.19\%, 2.00--2.88\%, 2.17--2.24\%, 1.78--1.81\%
on the four phones in the same order. These are cost-model sensitivity
ranges, not statistical confidence intervals or proof of native-corpus
optimality.

\subsubsection{Translation and spoken QA results}
\label{sec:appendix-eval-omni-results}
The cohorts, model pairings, budgets, and replay cost coverage are defined
in Appendix~\ref{sec:appendix-eval-setup}. The following results retain
those frozen inputs and the QA length-selection rule.

Conditional \system acceptance on English-to-Chinese, Chinese-to-English,
and QA is 213/390 (54.6\%), 438/687 (63.8\%), and 537/634 (84.7\%),
respectively. Accepted proposals cover 26.6\%, 36.7\%, and 28.3\% of useful
target output. The corresponding $K>1$ verification rounds are 179/669
(26.8\%), 281/834 (33.7\%), and 251/1,438 (17.5\%). Thus high acceptance
among tested proposals does not guarantee broad candidate coverage or a
large throughput gain, even with a common model pair.

The Chinese-to-English audit matches all 80 earlier Target-only records
and checks 2,601 \system rounds across budget arms. The English-to-Chinese
80-ID audit matches all three earlier serial-method records and checks
2,049 \system rounds. The QA audit matches all 80 frozen audio and
reference hashes and reproduces every Target-only timing and token count.
All 80 independent QA sources reach natural EOS without truncation.
The first frozen QA sample also passes incremental-versus-rebuilt KV checks.
Independent batched-prefix reconstruction finds near-tied FP16 top-1
differences in one English-to-Chinese ID (SD and SpecASR) and one
Chinese-to-English ID (SpecASR). The executed serial $K=1$ traces, frozen
target outputs, and reported costs remain unchanged. These checks do not
establish corpus-wide native numerical equivalence.

\paragraph{Paired QA result.}

On these same requests, Target-only, SD, adapted SpecASR, and \system
reach 7.773497, 8.902087, 8.567286, and 9.617707 TPS, respectively.
\system improves TPS by 8.0\% over SD, the strongest baseline, and by
23.7\% over Target-only. All speculative methods use the Omni-3B drafter;
\system keeps budget 3, while the serial baselines keep their reported
budget 7. Budget arms outside this fixed comparison do not determine the
reported result. The QA source replay preserves publication times,
causal alignment, and in-flight source drain.

The Omni phone-cost replay does not separately calibrate confidence
softmax or drafter KV branch switching.

\subsubsection{Workload and native scheduling cases}
\label{sec:appendix-eval-cases}
The case-selection protocol, full-window slowdown counts, and native
scheduling traces below support the workload analysis in
Section~\ref{sec:main-case-studies}.

\paragraph{Case selection and throughput-degradation accounting.}
\phantomsection
\label{sec:representative-case-protocol}

Tables~\ref{tab:representative-cases} and~\ref{tab:throughput-coverage}
use the unchanged 688-window cohort and phone profiles from
Figure~\ref{fig:fourphone-tps-cdf}. All population statistics retain every
window, including short and capped outputs. No selection here changes the
deployment configuration, timing profile, or evaluation cohort.

\noindent\emph{Central-window selection.}\quad
For each of the five Xiaomi 14 datasets, we compute the median of the
per-window ratio $\mathrm{TPS}_{\system}/\max(\mathrm{TPS}_{\rm SD},
\mathrm{TPS}_{\rm SpecASR})$ over all test windows. The displayed case
minimizes the absolute log deviation from this median; ties use the sample
ID. Only displayed cases must reach EOS and contain at least 64 target
tokens. We impose no minimum speedup and do not rank by maximum gain.
The additional K70 read-speech case uses the same audio as its Xiaomi 14 counterpart,
rather than selecting a second favorable input. Xiaomi 14 supplies a
consistent device comparison across all five workload families; K70
illustrates the effect of a different device on one shared input. This is
a retrospective choice of explanatory regimes, not random sampling across
phones or a new estimator of overall performance.

The all-window median TPS improvements over the faster serial method on Xiaomi 14
are 91.0\% (LibriSpeech), 102.7\% (FLEURS),
114.5\% (KeSpeech), 100.4\% (M4Singer), and
101.3\% (AMI). Across all four phones, 431 of 2,752 phone/window
pairs improve TPS by at least 100\% over both serial methods; 365 of these pairs are
on Xiaomi 14. This comparison excludes target-only from the denominator.
Table~\ref{tab:representative-cases} therefore shows target-only TPS explicitly
and defines Gain as
$100[\mathrm{TPS}_{\system}/\max(\mathrm{TPS}_{\rm SD},
\mathrm{TPS}_{\rm SpecASR})-1]$, computed before rounding.

\begin{table}[!t]
\centering
\caption{Medians and interquartile ranges summarize each illustrative regime.
All windows are retained. Values are per-window TPS changes (\%) relative
to target-only, rather than changes in pooled TPS. The intervals describe
variation across windows, not confidence intervals. Negative values indicate slowdowns.}
\label{tab:case-group-distributions}
\small
\begin{tabular*}{\linewidth}{@{\extracolsep{\fill}}lrrr@{}}
\toprule
Case group & Standard SD & SpecASR & \system \\
\midrule
A & +28.8 [+24.3, +32.3] & +31.2 [+15.6, +42.4] & +124.4 [+104.1, +134.6] \\
B & -16.7 [-19.4, -14.5] & -14.7 [-20.1, -10.8] & +64.4 [+61.9, +66.6] \\
C & -22.3 [-27.6, -18.2] & -24.3 [-32.1, -18.4] & +61.2 [+54.5, +64.6] \\
D & -29.2 [-36.2, -24.2] & -31.9 [-40.1, -25.2] & +54.6 [+45.7, +60.1] \\
E & -27.9 [-33.7, -21.8] & -30.7 [-39.9, -22.5] & +50.8 [+38.3, +58.6] \\
F & -39.6 [-44.0, -35.1] & -48.0 [-53.1, -43.0] & +21.4 [+3.3, +33.0] \\
\bottomrule
\end{tabular*}
\end{table}

\noindent\emph{Representative windows.}\quad
Every method within a case uses the same audio and target-token count.
Groups B--F in Table~\ref{tab:case-group-distributions} correspond to the
five Xiaomi 14 workloads; group A is the K70 read-speech comparison:
\begin{itemize}
  \item Read speech (B): LibriSpeech, 57.005\,s and 206 tokens. The K70 comparison (A) uses the same input.
  \item Multilingual speech (C): FLEURS German, 60\,s and 254 tokens.
  \item Dialect speech (D): KeSpeech Southwestern, 60\,s and 96 tokens.
  \item Singing (E): M4Singer, 60\,s and 135 tokens.
  \item Meeting speech (F): AMI SDM, 60\,s and 221 tokens.
\end{itemize}

\noindent\emph{Same-input device comparison.}\quad
For the shared read-speech input, K70 reaches 16.92, 22.65, 23.44, and
40.42 TPS with target-only, SD, SpecASR, and \system, respectively.
The corresponding Xiaomi 14 values are 19.26, 16.65, 16.88, and 32.24 TPS.
SpecASR's drafting cost rises from 5.33\,s on K70 to 9.29\,s on Xiaomi 14,
while target-only time falls from 12.18\,s to 10.70\,s. Serial speculation
therefore helps on K70 but slows Xiaomi 14 for the same 206-token input.

\noindent\emph{Cost accounting.}\quad
SpecASR drafting includes scalar/paired predictions, confidence work, and
draft KV refresh. Its total adds target verification and the same per-round
loop cost as the full-cohort replay. On Xiaomi 14, drafting accounts for
73.1--76.1\% of SpecASR time in the selected cases. Concurrent source
operations cover 97.5--97.7\% of executed source-active time through drain
and cleanup, including publication. This fraction is not the fraction of
total decoding time saved: verification, queue work, and processor
contention remain. Instrumenting the existing replay timeline reproduces
each published \system time and SpecASR decomposition to $10^{-12}$
relative tolerance. The replay uses the same candidates, queue implementation,
phone costs, and causal-publication checks as the full-cohort comparison.
These are phone-profile replay examples, not newly measured native phone executions.

\noindent\emph{All-window degradation rates.}\quad
For each phone and method, we count a
window as faster or slower by comparing its TPS with target-only on the
same phone and audio. A ratio within $10^{-9}$ of one is a numerical tie;
there are none. Table~\ref{tab:throughput-coverage} reports the slower fraction.
This tolerance handles floating-point equality, not timing
uncertainty. Standard SD uses the per-window arithmetic mean of fixed
physical K=5,6,7,8 TPS, matching the CDF. Each of the 688 windows receives
equal weight, regardless of duration or token count. Thus the percentages
are window frequencies, not pooled speedups or dataset-weighted means.
The capped output remains in the denominator. These are throughput
regressions, not recognition-quality regressions; the replay does not
establish native-corpus timing or output parity.

\begin{table}[!t]
\centering
\small
\caption{The counts compare complete-window throughput with target-only.
All 688 windows per phone are retained. Losses measure TPS reductions
among slower windows.}
\label{tab:throughput-coverage-detail}
\begin{tabular*}{\linewidth}{@{\extracolsep{\fill}}llrrrr@{}}
\toprule
Phone & Method & Faster & Slower & Median loss & Max. loss \\
\midrule
Redmi K70 Pro & Standard SD & 541 & 147 & 9.8\% & 62.1\% \\
 & SpecASR & 375 & 313 & 19.2\% & 74.3\% \\
 & \system & 641 & 47 & 14.2\% & 21.9\% \\
\midrule
Xiaomi 11 Pro & Standard SD & 611 & 77 & 8.3\% & 58.5\% \\
 & SpecASR & 553 & 135 & 10.7\% & 67.8\% \\
 & \system & 657 & 31 & 7.4\% & 20.3\% \\
\midrule
Xiaomi 14 & Standard SD & 0 & 688 & 24.6\% & 74.7\% \\
 & SpecASR & 0 & 688 & 26.4\% & 79.6\% \\
 & \system & 647 & 41 & 6.3\% & 10.1\% \\
\midrule
Redmi Note 9 Pro & Standard SD & 2 & 686 & 16.6\% & 73.0\% \\
 & SpecASR & 90 & 598 & 20.7\% & 79.3\% \\
 & \system & 651 & 37 & 7.2\% & 11.3\% \\
\bottomrule
\end{tabular*}
\end{table}

Table~\ref{tab:throughput-coverage-detail} reports exact counts and
conditional loss magnitudes. The frozen \system deployment reduces the
frequency of slowdowns but does not guarantee a speedup on every window.

\paragraph{Native scheduling cases.}
\phantomsection
\label{sec:native-scheduling-cases}

Table~\ref{tab:native-scheduling-cases} collects existing native comparisons
that expose different bottlenecks. Each row uses a matched input and its
original runner; results from different rows are not pooled. The target-decoupled
comparisons change both prefix dependence and overlap, whereas the older
lookahead comparison retains prefix dependence and matched acceptance.
These calibration examples complement the 688-window test results.

\begin{table}[!t]
\centering
\caption{\textbf{Native cases compare source and scheduling choices.}
Arrows follow the comparison column, and $n$ gives measured repetitions
per arm. Acceptance $\alpha$ divides pooled accepted draft tokens by pooled
proposed draft tokens. Rows retain their own runtime boundaries.
Gains are not pooled across rows.}
\label{tab:native-scheduling-cases}
\small
\setlength{\tabcolsep}{3pt}
\begin{tabular*}{\linewidth}{@{\extracolsep{\fill}}llllrr@{}}
\toprule
Case & Device/data & Comparison & TPS & $\alpha$ (\%) & $n$ \\
\midrule
Supply & Mi 14/Libri & Qwen $\to$ Parakeet & 28.36 $\to$ 46.33 & 74.0 $\to$ 69.1 & 5 \\
Verification & Mi 14/AMI & Qwen $\to$ Parakeet & 22.11 $\to$ 22.74 & 41.1 $\to$ 23.9 & 5 \\
TD, high acc. & K70/Libri & Serial $\to$ TD & 15.41 $\to$ 29.81 & 77.4 $\to$ 81.7 & 5 \\
TD, low acc. & K70/KeSpeech & Serial $\to$ TD & 9.62 $\to$ 15.80 & 44.9 $\to$ 24.3 & 1 \\
TD, zero acc. & K70/M4Singer & Serial $\to$ TD & 11.85 $\to$ 12.52 & 80.4 $\to$ 0.0 & 1 \\
Prefix lookahead & K70/AMI & Serial $\to$ lookahead & 9.25 $\to$ 9.41 & 37.0 $\to$ 37.0 & 2 \\
\bottomrule
\end{tabular*}
\end{table}

\noindent\emph{Candidate supply versus physical overlap.}\quad
On the prespecified Xiaomi 14 LibriSpeech calibration window,
Qwen and Parakeet retain the same GPU target and budget 7. Each has one
warmup and five measured runs with identical 239-token outputs.
Qwen's drafting is physically overlapped for 97.4\% of its active time,
yet 78.7\% of verification rounds have physical K1 and no ready candidates.
Parakeet reduces that fraction to 44.0\% and raises pooled TPS from 28.36
to 46.33, although its accepted/proposed ratio is lower (69.1\% versus 74.0\%).
The benefit comes with faster candidate supply, not a higher acceptance
ratio or a larger fraction of drafting hidden. Producer service spans include
publication; GPU spans are target-call intervals rather than hardware-kernel
traces. K1 spans represent useful target computation, not GPU idle waiting.
Figure~\ref{fig:native-scheduling-timelines} selects the first
multi-candidate rejection and the first fully accepted K8 from the
median-duration Qwen run. The equal-width windows are not consecutive.
Across all five Qwen repeats, observed K6 and K8 calls have median native
durations of 97.4 and 103.9\,ms (28 and 31 calls), respectively. The
widths are chosen by candidate availability under a common budget-7 cap,
so these timings are not a fixed-width causal comparison.
TPS in the table pools all five runs, including the slow Qwen repeat.

\noindent\emph{Low acceptance can leave verification dominant.}\quad
On the prespecified Xiaomi 14 AMI calibration window, Parakeet's acceptance
is 23.9\% versus Qwen's 41.1\%. K1 rounds still account for 84.5\% with
Parakeet, and pooled throughput rises only from 22.11 to 22.74 TPS (2.9\%).
The same faster source therefore provides much less benefit when candidates
rarely advance the target. Neither source requires waiting for the target
prefix, but verification and correction still determine useful progress.

\noindent\emph{Removing exposed drafting at low acceptance.}\quad
K70's fixed-budget-7 KeSpeech calibration pair reaches 15.80 TPS with TD
versus 9.62 with serial prefix-conditioned drafting, despite acceptance
falling from 44.9\% to 24.3\%. Source independence and overlap remove enough
exposed work to offset the lost agreement. The M4Singer pair is a contrasting
case: TD accepts none of 30 proposals and improves only from 11.85 to
12.52 TPS (5.7\%). Both examples contain one measured pair and matching
outputs; they are individual observations rather than dataset estimates.
The five-repeat LibriSpeech calibration gives 15.41 versus 29.81 TPS.

\noindent\emph{Prefix-dependent lookahead can waste hidden work.}\quad
The earlier K70 AMI experiment retains the serial acceptance trajectory
while overlapping lookahead. At its original K8 setting, physical overlap
hides 34.9\% of drafting, but only 6.2\% is reusable after rejection and
realignment. Two native pairs give 9.25 versus 9.41 TPS (1.7\%). This is a
different runner from TD; it illustrates why physical overlap alone does
not establish useful speedup, rather than supplying a current baseline row.

\paragraph{Drafter cost and verification budget on Redmi K70 Pro.}
\phantomsection
\label{sec:k70-drafter-case}

This native case study checks whether inexpensive audio proposals provide
net decoding gains once generation, verification, and recovery are executed
together. It supplements acceptance-only screening with a measured
cost--acceptance trade-off under fixed source and budget settings.

\noindent\emph{Matched-audio protocol.}\quad
We use one 60-second AMI SDM development window from recording ES2004c,
whose Qwen3-ASR 1.7B target output contains 212 tokens including EOS.
The window is the first of at least 50 seconds in the fixed AMI development
pilot order. All audio is available before decoding, so this experiment
isolates decoding and does not measure sustained streaming or arrival lag.
The target uses OpenCL with low precision on Redmi K70 Pro. Parakeet CTC
and SenseVoice remain resident; only the selected source generates
candidates, using four CPU threads. Each processes the same 24 audio chunks
of at most 2.56 seconds, without reading the target text prefix.

We test $K\in\{3,4,8\}$ for each drafter and a target-only control.
Here $K$ is the maximum number of real token nodes in an admitted chain,
including its first token already predicted by the target. A rolling
buffer can span chunk boundaries. It retains and realigns candidate
suffixes after rejection; unmatched or unavailable proposals use target
K1 progress, after which later candidates remain eligible. Deferred KV
submission carries the target's correction or bonus token to the next
batch, avoiding an unconditional additional K1 forward after verification.
Each run keeps its selected drafter and budget unchanged.

Each of the seven configurations has one warmup and five measured runs.
Configuration order is shuffled within each repetition, and each treatment
is compared with the target-only control from that repetition. Timing runs
from target-prefill completion to committed EOS. Source production starts
on this clock, with no precomputed candidate text; exposed source work,
retokenization, queue alignment, target verification/readback, and recovery
are included. Initial model loading, common waveform loading, and target
prefill are excluded. Worker cleanup after EOS is recorded separately and
has a median below 5\,ms for every configuration. All 35 measured runs
reproduce the complete reference token sequence through EOS.

\begin{table}[!t]
\centering
\small
\setlength{\tabcolsep}{6pt}
\begin{tabular}{@{}lrrrrr@{}}
\toprule
Drafter & Max. $K$ & TPS & Gain (\%) & Extra tokens & Target calls \\
\midrule
Target-only & scalar & 14.58 & 0.0 & 0 & 211 \\
Parakeet CTC & 3 & 15.44 & 5.9 & 33 & 178 \\
Parakeet CTC & 4 & 16.25 & 11.5 & 39 & 172 \\
Parakeet CTC & 8 & 15.66 & 7.4 & 43 & 168 \\
SenseVoice & 3 & 16.03 & 10.0 & 24 & 187 \\
SenseVoice & 4 & 16.03 & 9.9 & 26 & 185 \\
SenseVoice & 8 & 15.94 & 9.3 & 30 & 181 \\
\bottomrule
\end{tabular}
\caption{Native decoding compares drafters on one 60-second AMI development
window on Redmi K70 Pro. TPS pools five measured runs per configuration.
Target-only uses scalar decoding.
Gain is relative to the interleaved target-only control. Extra tokens and
target calls are per-run medians. Extra tokens exclude the already-known
first token of each admitted chain. All 35 runs agree through EOS (212 tokens).}
\label{tab:k70-drafter-case}
\end{table}

\noindent\emph{Realized benefit and budget sensitivity.}\quad
Table~\ref{tab:k70-drafter-case} pools output tokens and decoding time
across the five measured repetitions. Parakeet/$K=4$ reaches 16.25 TPS,
11.5\% above target-only; its five paired gains range from 5.6\% to 16.2\%.
SenseVoice/$K=4$ reaches 16.03 TPS, a 9.9\% gain with a paired range of
5.8\%--14.9\%. These repeated gains concern this window, not uncertainty
across recordings. At $K=4$, all source chunks are ready after a median
2.13\,s for Parakeet and 7.74\,s for SenseVoice. These intervals overlap
target execution and must not be added to the total decoding time.

Parakeet/$K=4$ accepts a median of 39 additional tokens beyond the
already-known first token of each admitted chain, reducing target forward
calls from 211 to 172. At $K=8$, it accepts 43 additional tokens and uses
168 calls, yet throughput falls to 15.66 TPS. Fewer calls do not guarantee
lower time when each wider verification costs more. The very small
SenseVoice difference between $K=3$ and $K=4$ does not establish a stable
ordering. These results motivate selecting a drafter and budget by their
net decoding benefit rather than source size, acceptance, or invocation
count alone.

All tested source/budget configurations enter the comparison. These fixed
controls support the cost--acceptance interpretation for this audio input
and timing boundary.

\FloatBarrier

\subsection{RQ2: Component ablations}
\label{sec:appendix-eval-rq2}
The removals below share the corpus, target references, and phone costs
with the complete configurations. Appendix~\ref{sec:appendix-eval-deployment}
gives the calibration and development-selection protocols; each removal
states its additional cost assumptions and sensitivity checks.

\subsubsection{Cumulative Redmi K70 Pro ablations}
\label{sec:k70-cumulative-ablation}

Table~\ref{tab:evaluation-performance} evaluates two modules added to
SpecASR (ASP): budget configuration (BC), then target-decoupled overlapped
drafting (TD) with the fixed Qwen drafter. The main table presents
these configurations as \system, \system w/o TD, and the existing SpecASR
baseline, and adds the independent \system w/o BC control that retains TD.
Removing both modules is identical to ASP and is not repeated as a separate row.
All rows use the same
target model and full-audio 688-window cohort as
Table~\ref{tab:evaluation-performance}. This study evaluates post-prefill
decoding on complete windows with chain candidates. Deployment configurations
are fixed before testing. The ASP arm uses the calibrated operation costs in
Appendix~\ref{sec:k70-baseline-replay}.
All costs and configuration choices use the K70 development measurements.

\paragraph{First addition: budget configuration.}
We screen proposal limits 1--24 on 50 balanced development windows and retain
budget 7 after validation. This raises serial TPS from 16.81 to 19.77.
The fixed setting and development split are documented in
Appendix~\ref{sec:deployment-selection}.

\paragraph{Second addition: target-decoupled overlapped drafting.}
Source independence removes the target-prefix dependency that otherwise
limits concurrent drafting. We add target-decoupled proposals and overlap as
one module with the fixed Qwen 0.6B drafter. The configuration retains
CPU four-thread placement, the resident OpenCL target,
confidence threshold 0.4, and frozen budget 7. Both arms
consume the same complete waveform and finish their initial full-audio
prefills before timing. No output text is prefetched. The independent
producer conditions on audio and its own generated prefix; it reads no
target text or target KV. It publishes immutable batches of eight tokens
to an ordered buffer. These are token-publication records, not separate
audio chunks. The consumer applies the budget cap to already-published
candidates, using the existing alignment and target-KV handling. A miss
allows target-only progress followed by realignment without restarting
the producer.

This changes both the proposal interface and schedule, so the combined
gain is not pure overlap at unchanged acceptance. The event replay uses
actual server-generated source trajectories and the same fixed target
references as the serial rows. It preserves causal publication, immutable
verification snapshots, confidence stopping, correction, and completion
at EOS. The target verifier and common loop costs retain the existing
K70 anchors. Concurrent-service ratios, source publication, and alignment
costs are fitted only on calibration traces and frozen before test evaluation.
Small negative contention estimates receive no speed credit. Missing
solo-width calibration uses the median measured contention ratio. The replay charges
producer cancellation at an operation boundary and post-EOS drain once.

The evaluation retains all 688 windows and 94,511 reference tokens,
including the capped target reference. Source generation stops at EOS or
2,048 tokens; capped sources remain in the comparison. The queue audit
checks every accepted token against the fixed reference and every
candidate's publication time against its decision boundary. These checks
establish the reference-based replay contract, not exact corpus-wide
server/mobile numerical equivalence.

Earlier chunk-reset diagnostics required repeated audio prefills and delayed
publication until each chunk completed, yielding zero or low acceptance.
The evaluated configuration retains the full audio context and publishes
token batches during generation.

\paragraph{Drafter configuration.}
The development study retains Qwen at budget 7
(Appendix~\ref{sec:deployment-selection}). This same pair is used for
every test window in the reported 27.92 TPS target-decoupled configuration.

\paragraph{Independent budget-configuration removal with TD retained.}
The w/o BC control keeps the target-decoupled producer, overlap, immutable
8-token publication records, alignment, and target-KV handling. It
restores the original ASP configuration: fixed Qwen 0.6B, confidence
threshold 0.4, and a maximum of 24 proposals. This control is fixed
before its test replay; no test-time budget sweep or checkpoint selection
is performed. It removes configuration selection while holding the TD
runtime intact. Replaying all 688 windows retains all 94,511 reference
tokens and reproduces the budget-7 reference exactly for every window.

Without BC, pooled throughput is 25.60 TPS; the selected \system
configuration is 9.0\% faster. Dataset-specific gains range from 5.2\%
on FLEURS to 17.3\% on AMI. Physical verifier widths above eight use the
existing K70 solo profiles through width 25. Their concurrent slowdown
uses the median calibration ratio, clipped at one; these larger
widths lack direct concurrent measurements. They account for 4,732 of
36,514 w/o BC verification rounds. A predeclared sensitivity check
replaces their ratio with the largest measured calibration ratio, 1.049,
yielding 25.42 pooled TPS. Both values are profile-replay estimates.
An independent audit checks candidate visibility, reference-token progress,
confidence stopping, and charged time across all 111,741 rounds in the
reference, removal, and sensitivity runs, retaining the capped sequences.

\paragraph{Accounting boundary.}
TPS pools committed target tokens over charged decoding time. Startup,
exposed producer tail, and cleanup that delays the next window are charged
once under the same boundary for all variants. Failed and capped windows remain recorded.

\subsubsection{Xiaomi 11 Pro component removals}
\label{sec:xiaomi11-component-ablations}

The two component removals retain the same 688 test windows, Qwen source,
device profiles, and post-prefill clock. Without TD, we reprice the actual
serial ASP budget-7 operation traces, including confidence computation,
recycling, and KV refresh; we do not truncate cap-14 trajectories. These
counters reproduce all 688 previous K70 budget-7 timings under the K70
costs. Under current Xiaomi costs, this arm reaches 11.00 pooled TPS,
versus 14.68 for \system, a 33.4\% increase with TD. Without BC, we retain
target-decoupled overlap, eight-token publication and the same queue, while
restoring the original ASP cap of 14 proposals and confidence threshold
0.4. This arm reaches 13.79 TPS; \system is 6.4\% faster. Both caps are specified before
test evaluation; neither arm selects configurations from test outcomes.
Removing both modules yields the already listed SpecASR baseline.

For the w/o-BC arm, physical widths K9--15 use measured solo verifier
profiles, but their concurrent costs are extrapolated. We multiply solo
cost by 0.9651, the median of the directly measured calibration
concurrent/idle ratios at K1, K5 and K8 (each with at least five
observations), following the existing fit rule without clipping. Of
34,261 verification rounds, 6,108 exceed K8. A predeclared sensitivity
check substitutes the smallest and largest eligible calibration ratios,
0.9387 and 1.2396, yielding 13.86 and 13.16 pooled TPS, respectively.
These scenarios test this cost assumption; they are not confidence intervals
or new concurrent phone measurements at large K.

The original \system time is reproduced exactly on every window. Across
the baseline, removal and sensitivity arms, 139,236 overlapped rounds
pass token, queue-readiness and timing-accounting checks; 17,234 serial
budget-7 rounds also pass trace checks, validating replay consistency.

\subsubsection{Xiaomi 14 and Redmi Note 9 Pro component removals}
\label{sec:additional-phone-ablations}

The Xiaomi 14 and Redmi Note 9 Pro ablations use the same definitions as
Section~\ref{sec:main-ablation}. Without TD, we hold the frozen source and
budget fixed and price actual serial SpecASR traces at budget 7 on
Xiaomi 14 and budget 3 on Note 9. Both controls cover all 688 windows.
The Note 9 control combines 622 retained budget-3 traces with 66 newly
collected trajectories, all priced with the same frozen phone profile
and no selector fee. GPU collection time is excluded from phone TPS.
Without BC,
we keep target-decoupled overlap and restore confidence-controlled ASP with a
14-proposal cap. Each arm retains Qwen, the same windows, and the same
post-prefill resource clock. Removing TD preserves the fixed drafter and
budget; test outcomes do not enter configuration selection.

On Xiaomi 14, the retained deployment reaches 28.40 TPS, versus
27.84 without BC and 14.55 without TD. For the without-BC arm,
2,329 of 52,804 verification rounds exceed physical K8. Their concurrent
costs use 0.9961 times the directly measured idle-width cost, the
median eligible calibration ratio. Replacing it by the observed ratio extremes
(0.9032, 1.0179) yields 27.81--27.98 TPS. These are
cost-sensitivity scenarios, not confidence intervals or measured concurrent
K9--15 executions.

On Redmi Note 9 Pro, the fixed deployment reaches 5.85 TPS, versus
5.07 without BC and 3.84 without TD. The matched without-TD comparison
shows a 52.5\% gain from target-decoupled overlap. For the without-BC arm,
5,836 of 36,279 verification rounds exceed physical K8. Their concurrent
costs use 1.0593 times the directly measured idle-width cost, the
median eligible calibration ratio. Replacing it by the observed ratio extremes
(1.0565, 1.0694) yields 5.06--5.07 TPS. These are
cost-sensitivity scenarios, not confidence intervals or measured concurrent
K9--15 executions.

The validation-only fixed selection is documented in
Appendix~\ref{sec:deployment-selection}.

\FloatBarrier

\subsection{RQ3: Native on-demand execution}
\label{sec:appendix-eval-rq3}
\subsubsection{Native on-demand request protocol and baselines}
\label{sec:native-request-protocol}
The Redmi K70 Pro request experiment uses an Omni-7B target on six CPU
threads and a Qwen3-ASR-Audio-0.6B drafter on OpenCL, with both models
resident. MNN precision is \texttt{normal} for CPU inference and
\texttt{low} for the OpenCL decoder; these settings do not specify weight
quantization. Both models keep independent audio KV states. All speculative
methods use proposal budget $b=3$. SD generates up to three
prefix-conditioned candidates before
target verification. SpecASR uses the same budget, confidence threshold
0.4, and aligned recycling with serial old/new drafter branches. This is
an ASP chain adaptation, distinct from the masked-pair ASR replay above.
Request timing includes instruction-suffix prefill and generation through
EOS, excluding earlier audio preparation and model loading
(Section~\ref{sec:main-native-requests}).

The sweep uses 19 positions at 60,90,\ldots,600s of the same \emph{Pearl}
recording. Each request transcribes only the latest 60s window, whose
language-model audio KV is prepared before the instruction arrives.
No new audio enters a running request. Positions identify audio checkpoints,
not a real-time request schedule or progressively longer response inputs.
Each plotted point is one measured request after an initial execution check, and the
four curves combine separate runs under the same request protocol.
TPS divides committed body tokens by request-to-EOS wall time, including
instruction-suffix prefill. We report all 19 formal observations per method
without estimating run-to-run uncertainty.

\paragraph{Native request output checks.}
These checks measure cross-method agreement with Target-only, rather than
transcription accuracy against ground truth. At the 120s checkpoint,
all four methods produce the same 173 body tokens, so the peak 78.0\% TPS
gain also reduces response time from 25.05s to 14.07s. SD and SpecASR produce
the same token sequence as each other at all 19 positions. The word-agreement bound
requires both at most three word edits and an edit rate of at most 1\%
relative to Target-only. SD and SpecASR satisfy it at 18 of 19 positions,
and \system at 17 of 19. At 210s, \system differs by 45 word edits.
At 270s, SD, SpecASR, and \system differ by 9, 9, and 2 edits, respectively,
and all exceed the 1\% rate bound. All observations remain in
Figure~\ref{fig:native-ondemand-fourway}. Thus the throughput comparison
does not establish exact-output preservation across the sweep or
corpus-wide gains. Appendix~\ref{sec:numerical-sensitivity} discusses
numerical execution boundaries.

\paragraph{Request energy on Xiaomi 14.}
We measure the first 60s window of the Xiaomi 14 trace in
Figure~\ref{fig:native-ondemand-fourway}, using the same resident model pair,
processor placement, and proposal budget as the native request experiment.
Each method has one warmup and one formal execution, with the formal
\system request followed by Target-only. Both executions start from
prepared audio KV and finish naturally at EOS. USB is physically disconnected
throughout the measurement, while MNN remains in the foreground and external
cooling remains enabled.

The phone's battery current and voltage are sampled every 250\,ms through
Perfetto. Observed gauge values update approximately once per second.
We estimate request energy as $E=\int_{t_0}^{t_{\mathrm{end}}}I(t)V(t)\,dt$
using trapezoidal integration, where $t_0$ is the native request origin and
$t_{\mathrm{end}}$ is the time when both model workers have finished.
The endpoints and power samples use the phone's boot-time clock.
This boundary includes instruction-suffix prefill, decoding, and any producer
drain. Response time ends at target EOS.
Energy excludes model loading and earlier audio preparation.
The estimate covers the whole phone, including its display and background
activity, with no idle-power subtraction. Energy consumed by the external
cooler is excluded.

\begin{table}[t]
  \centering
  \small
  \caption{\system uses less estimated phone energy on one 60s Xiaomi 14
  window. Each row reports one formal execution after warmup.}
  \label{tab:native-request-energy}
  \begin{tabular}{lrrrrr}
    \toprule
    Method & Energy (J) & Mean power (W) & Response (s) & Body tokens & J/token \\
    \midrule
    Target-only & 341.6 & 10.47 & 32.64 & 256 & 1.334 \\
    \system & 239.4 & 11.33 & 21.13 & 248 & 0.965 \\
    \bottomrule
  \end{tabular}
\end{table}

The shorter response more than offsets the higher mean power, reducing
request energy by 29.9\% and energy per body token by 27.7\%
(Table~\ref{tab:native-request-energy}). The two formal responses contain
the same 194 normalized words in the same order, with differences in
punctuation, casing, and tokenization. Their raw token sequences are not
identical, so the table also reports token counts and energy per token.

These are battery-gauge estimates from one window and one formal run per
method. Over the full captured interval, current integration gives
145.9\,mAh, compared with a 159.0\,mAh decrease in the charge counter, an
8.2\% discrepancy. Both readings come from the phone's gauge and do not
constitute independent external calibration. The case measures request
energy under this deployment and does not estimate run-to-run variability
or battery-life gains across workloads.

\subsubsection{Sustained native streaming}
\label{sec:native-sustained-streaming}
\newcommand{\StreamCoolStages}{0.70 / 2.24 / 0.68}
\newcommand{\StreamCoolService}{3.63}
\newcommand{\StreamCoolRTF}{0.907 / 1.007}
\newcommand{\StreamCoolTPS}{13.87}
\newcommand{\StreamCoolPninetyfive}{4.68}
\newcommand{\StreamCoolFinalLag}{4.09}
\newcommand{\StreamCoolWait}{2.07 / 0.21}
\newcommand{\StreamCoolTemperature}{23.0 / 35.4}
\newcommand{\StreamCoolCaps}{0}
\newcommand{\StreamCoolShort}{28}
\newcommand{\StreamCoolLatePrefill}{2.34}
\newcommand{\StreamBareStages}{0.94 / 2.95 / 0.78}
\newcommand{\StreamBareService}{4.68}
\newcommand{\StreamBareRTF}{1.171 / 1.201}
\newcommand{\StreamBareTPS}{12.23}
\newcommand{\StreamBarePninetyfive}{103.39}
\newcommand{\StreamBareFinalLag}{120.50}
\newcommand{\StreamBareWait}{114.03 / 114.03}
\newcommand{\StreamBareTemperature}{28.0 / 42.1}
\newcommand{\StreamBareCaps}{0}
\newcommand{\StreamBareShort}{28}
\newcommand{\StreamBareLatePrefill}{3.46}
\newcommand{\StreamExactOutputs}{150}

\noindent\textbf{Protocol.}\quad
A resident MNN native runner on Redmi K70 Pro receives a continuous 600s
English audiobook excerpt in 150 non-overlapping 4s blocks, without added
pauses or repeated segments. The excerpt comes from the \emph{Pearl} chapter
of \emph{The Scarlet Letter}; the recording was chosen during exploratory
experiments and is a deployment case, not a held-out corpus result.
The Qwen2.5-Omni-7B target uses six CPU threads, its encoder uses four,
and Qwen3-ASR-Audio-0.6B drafts on OpenCL. The proposal budget is three
(up to four verifier nodes including the known root). Target encoding,
prefill, and verification execute serially within each block, concurrently
with the resident source worker. The worker reads only arrived audio;
verification uses only candidates ready at its decision time. Audio and
committed text context persist for 12s, then reset while model weights stay
resident. Every block and final resource drain is retained.

\noindent\textbf{Deployment cases.}\quad
Both runs retain 75 of the 100 target audio tokens per block; the source
receives uncompressed audio. The cooled run uses a magnetic thermoelectric
cooler in its maximum mode and an arrival-based source deadline.
The uncooled run uses a consumption-based source lifetime: an epoch stops
drafting when the target has consumed it, avoiding premature expiry during
backlog. In a separate 60s paired check, switching to consumption-based
expiry changes mean complete processing
from 3.62s to 3.63s with identical outputs in all 15 blocks; this does not
establish long-run equivalence. Starting temperatures also differ
(Table~\ref{tab:native-sustained}). The curves are descriptive deployment
cases rather than an isolated cooling ablation.

\begin{figure}[t]
  \centering
  \includegraphics[width=\linewidth]{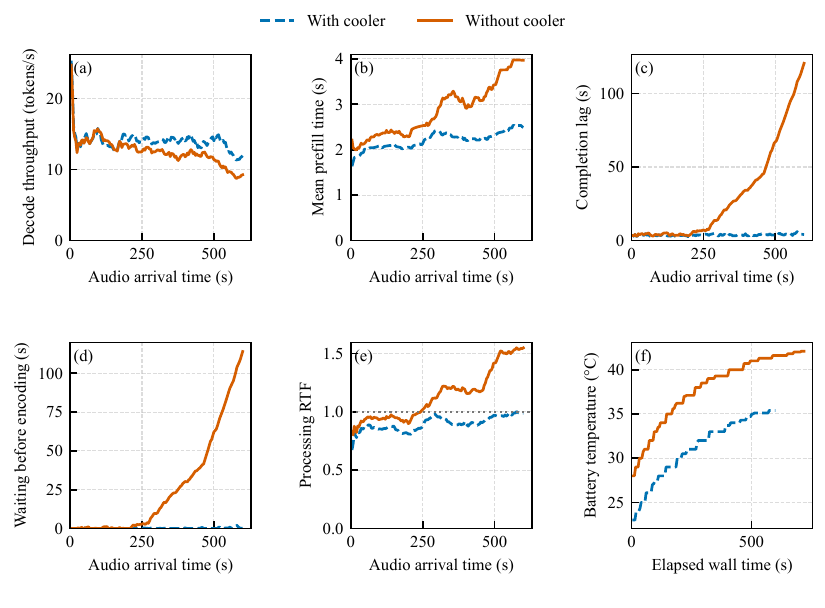}
  \caption{The cooled configuration keeps pace over 600s, while the uncooled
  case develops backlog. Decode throughput, prefill time, and processing RTF
  use trailing windows of at most 15 blocks. Completion lag and queue wait
  show every block. Panels (a)--(e) use input-arrival time. Temperature uses
  elapsed wall time and includes final drain. The dotted line marks RTF 1.
  Both cases use 75\% token retention, but source deadlines and initial
  thermal states differ as described above.}
  \label{fig:native-sustained-streaming}
\end{figure}

\noindent\textbf{Metrics and observations.}\quad
Completion lag starts when a block's final audio sample arrives and ends
when its committed output completes; it excludes the block's own 4s capture
interval. Queue wait ends when encoding begins. Active RTF divides serial target-path wall time (target encoder plus
prefill, decoding, and host service) by audio duration; concurrent drafter
processor work is not added a second time. Resource RTF instead uses total
elapsed time, including paced waits and all producer drain. Decode throughput
counts committed non-stop tokens over decoding time. Quantiles use linear
interpolation over all 150 blocks. Loading and warm-up precede the stream.

\begin{table}[t]
  \centering\small
  \caption{The table summarizes both complete native deployment runs.
  Temperature comes from the battery sensor. Source-lifetime and initial-state
  differences preclude cooling-only attribution.}
  \label{tab:native-sustained}
  \begin{tabular}{lrr}
    \toprule
    Metric & With cooler & Without cooler \\
    \midrule
    Mean encoder / prefill / decode (s) & \StreamCoolStages & \StreamBareStages \\
    Mean complete block processing (s) & \StreamCoolService & \StreamBareService \\
    Active / resource RTF & \StreamCoolRTF & \StreamBareRTF \\
    Decode throughput (tokens/s) & \StreamCoolTPS & \StreamBareTPS \\
    p95 completion lag (s) & \StreamCoolPninetyfive & \StreamBarePninetyfive \\
    Final completion lag (s) & \StreamCoolFinalLag & \StreamBareFinalLag \\
    Maximum / final queue wait (s) & \StreamCoolWait & \StreamBareWait \\
    First / maximum battery temperature ($^\circ$C) & \StreamCoolTemperature & \StreamBareTemperature \\
    Capped outputs & \StreamCoolCaps & \StreamBareCaps \\
    \bottomrule
  \end{tabular}
\end{table}

The cooled run averages 3.63s complete service per 4s block, with p95
completion lag 4.68s and maximum temporary queue wait 2.07s. Prefill accounts
for 61.6\% of complete service. The uncooled run has mean service
\StreamBareService\,s and final completion lag \StreamBareFinalLag\,s;
its last-300s mean prefill is \StreamBareLatePrefill\,s; encoding and
prefill together average 4.58s per 4s block in that interval.
Thus decoding acceleration alone does not guarantee sustained timing:
audio-prefill efficiency and heat removal remain important deployment issues.

\noindent\textbf{Auxiliary audio-token reduction.}\quad
Audio-token pruning and merging can shorten the input to an audio language
model while retaining useful acoustic information
\citep{gibier2025segmentwise,luo2026locality}. This deployment uses an
auxiliary filter to reduce prefill cost while preserving temporal coverage.
An offline-fitted linear ridge predictor scores projected encoder tokens by
predicted target-attention importance. At runtime, one highest-scoring token
is retained from each of 75 consecutive equal-time bins, preserving temporal
coverage and order; retained tokens use consecutive target positions.
Selection occurs before target prefill and its cost is included in encoder
service. The target is not fine-tuned, while the scoring prior is fitted
offline. This filter is a configuration of the sustained-run experiment.

\noindent\textbf{Limits.}\quad
This paced-recording study measures native runtime behavior rather than
microphone/UI latency or an \system versus Target-only/SD comparison.
It does not evaluate recognition accuracy under audio-token reduction. \StreamCoolShort/150 cooled outputs contain at most three
body tokens; all are retained. Output equality across the two deployment cases is
\StreamExactOutputs/150. Better prefill kernels, token reduction
with explicit quality evaluation, and thermal management are complementary
options outside the decoding contribution. One run per condition and a fixed
retention rate do not establish robustness across workloads or durations.

\FloatBarrier
\section{Extended related work}
\label{sec:related-work}

\subsection{On-demand audio understanding}
Prior systems let users retrieve information from previously captured audio.
Memoro combines a wearable audio interface with semantic memory retrieval
and concise LLM-generated assistance \citep{zulfikar2024memoro}. LA-RAG
converts long recordings into timestamped event records and answers queries
through structured retrieval \citep{hegde2026larag}. These systems motivate
responsive access to past audio, while choosing different representations
for the retained information.

Audio language models can instead retain model context for later generation.
MiniCPM-o 2.6 accepts continuous audio and video independently of user queries,
with separate streaming-prefill and generation interfaces
\citep{openbmb2025minicpmo26}. OmniMem allocates memory across audio and visual
inputs and selects informative KV states to support long-context understanding
within a memory budget \citep{sun2026omnimem}. This preparation reduces work
at request time, but generating a complete reply still requires decoding.
\system accelerates that stage under a fixed audio-preparation protocol.
Its objective is request-to-complete-response latency on a phone, making
input preparation, retrieval, and cache management complementary concerns.

\subsection{Speech proposals and target verification}
Standard speculative decoding separates cheaper proposal generation from
target verification \citep{leviathan2023speculative}. Proposals may come from
a separate model or reuse target layers, heads, and features
\citep{zhang2024draftverify,cai2024medusa,li2024eagle}. Speech provides an
additional source of candidate information. SpecASR exploits acoustic
alignment through adaptive draft lengths, local draft recycling, and sparse
trees \citep{wei2025specasr}. Its recycling procedure extends a branch from
verified text and merges it with earlier drafts when their tokens realign.
\system adapts this local reuse principle while letting the producer advance
without synchronizing to target-prefix updates. An autoregressive producer
still conditions on the request and its own generation history.

Speech-specific accelerators also differ in their search and acceptance rules.
SMUD uses decoder-assisted masking and masked/unmasked beam search to reduce
ASR search cost \citep{okabe2025smud}. Whisper-Medusa accepts multi-head
predictions using probability thresholds
\citep{segalfeldman2025whispermedusa}. CTC encoder drafts already avoid
dependence on target-text prefixes, but their confidence gating can bypass
LLM verification, and their relaxed verification uses token-likelihood
thresholds \citep{saon2026selfspec}. These acceptance rules need not reproduce
the target's greedy output. \system combines independently advancing proposals
with asynchronous execution, retaining target verification and correction.

Speculative speech recognition studies a different form of anticipation:
predicting future transcript tokens from audio heard so far
\citep{yusuf2024ssr}. Our requests concern already received audio.

\subsection{Asynchronous and mobile execution}
Asynchronous speculation changes when work is performed relative to target
decisions. PEARL uses pre- and post-verification to overlap draft generation
with verification \citep{liu2025pearl}. DSI schedules concurrent target and
drafter instances to reduce latency \citep{timor2025dsi}. SSD prepares and
caches continuations for predicted verification outcomes, reusing a
continuation when the prediction matches \citep{kumar2026ssd}. These methods
advance work from verified or hypothesized target-text prefixes. \system
instead draws candidates from the requested audio and instruction, so its
producer can continue without enumerating possible target corrections.
The target verifies ready candidates or decodes alone. Rejection triggers
candidate realignment while the producer continues.

Mobile systems additionally account for memory capacity, data movement, and
contention. LLMCad and EdgeLLM combine confidence-guided token-tree generation
with verification and model-loading pipelines
\citep{xu2023llmcad,xu2025edgellm}. EdgeLLM accounts for interference by
scheduling speculative work during target-weight loading rather than target
computation. Lever combines tree construction and pruning with CPU--NPU
execution for flash-backed targets \citep{wang2026lever}. PELM jointly
controls speculative execution and processor frequency, allowing
partial-depth verification to trade output fidelity for efficiency
\citep{yang2026pelm}. AHASD uses a custom NPU--PIM architecture and target
feedback to confirm or roll back asynchronous drafts \citep{ma2026ahasd}.
\system keeps the target resident and overlaps its computation with a
target-decoupled audio producer on available phone processors. The useful
overlap depends on candidate readiness and measured concurrent costs.

Resource-aware inference also spans other deployment settings.
\citet{liang2026dynamic} partition and preload an LLM across nearby mobile
devices, then schedule layer-wise execution while accounting for computation
and inter-device communication. On shared NVIDIA GPUs, SGDRC dynamically
allocates compute units and memory bandwidth among concurrent DNN inference
services \citep{zhang2025sgdrc}. These methods address distributed execution
and shared-GPU contention, respectively. \system instead coordinates an
audio producer and target verifier within one phone.

\subsection{Candidate structure and deployment configuration}
SpecInfer, Sequoia, and EAGLE-2 organize or adapt candidate trees to increase
coverage per verification pass
\citep{miao2024specinfer,chen2024sequoia,li2024eagle2}.
Our evaluated interface uses ordered candidate chains, which fit existing
causal-attention runtimes and support different audio drafters.

Model selection can adapt inference to changing input conditions.
Anole and its journal extension select among scene-specialized compressed
models for cross-scene inference on mobile devices
\citep{li2024anole,li2025sceneaware}. This selection changes the model used
for prediction. In \system, the target model remains fixed, while a cheaper
producer supplies candidates that the target verifies and corrects.

Configuration methods select drafters or candidate budgets.
UniSpec calibrates draft size against acceptance and device-specific execution
costs \citep{do2026unispec}. LTD trains state-dependent draft-tree depth and
verification-size policies before deployment \citep{zhang2026ltd}, while
MetaSD updates an alignment-feedback bandit online to select drafters
\citep{kim2026metasd}.
Online Speculative Decoding adapts the draft model itself
\citep{liu2024online}. \system uses a fixed task-specific drafter and
configures its budget using device profiles and development workloads
(Appendix~\ref{sec:deployment-selection}). This deployment choice supports
the asynchronous mechanism without requiring test-time learning.

\FloatBarrier
\section{Discussion and future directions}
\label{sec:discussion}

\subsection{Task-dependent speedup and complementary speculation}
\label{sec:discussion-candidate-scope}

Target-decoupled drafting is useful when the audio and instruction provide
candidates that remain aligned with the target's response. Acceleration
also depends on how quickly useful candidates arrive. Our translation and
spoken-QA results illustrate why a high acceptance ratio alone is
insufficient (Section~\ref{sec:main-omni-tasks}). Task-specific drafter
selection should account for both candidate supply and mobile execution
cost. The drafter conditions on its own generation history, but does not
observe the target's intermediate decisions. Open-ended reasoning or dialogue may depend more
strongly on those decisions, reducing usable candidate coverage. Our
recognition, translation, and spoken-QA results do not establish comparable
gains for such tasks. The target can continue alone when candidates are
unusable, but concurrent drafting still consumes resources.

Prefix-conditioned speculation offers a complementary source of candidates.
A future hybrid could retain the independent audio producer while selectively
using continuations conditioned on target-text prefixes. SSD's continuations
for predicted verification outcomes \citep{kumar2026ssd} and DSI's scheduling
of parallel target and drafter instances \citep{timor2025dsi} offer possible
components for this design. Whether additional coverage offsets prefix
synchronization, speculative state, and processor contention on a phone
requires evaluation.

\subsection{Supporting full-duplex dialogue}
\label{sec:discussion-interaction}

Each response uses an audio view fixed at request arrival.
New audio does not alter that response's input, and the evaluated requests
do not include user interruptions or multi-turn dialogue. Extending
\system to these settings requires defining when new audio, revised instructions, and dialogue
history enter the target context and which existing candidates remain
eligible. Target and drafter state must follow those context boundaries
while preserving causal input access. One direction is to update the
response context at well-defined points while reusing previously processed
audio and valid candidate continuations, extending \system toward
continuous interaction.

The prepared-cache request sweep evaluates decoding from completed audio KV
states. It does not establish an incremental cache-maintenance protocol
across successive requests. Evaluating that extension requires matching
audio preparation, request times, context, and commitment rules across
methods, with any unfinished preparation or state repair charged after
request arrival.

\subsection{Resource assumptions and fixed configurations}
\label{sec:future-work}

Useful overlap requires enough memory for resident target and drafter
models, their KV states, and candidate buffers, together with processors
that can execute concurrently. Shared-memory bandwidth and contention can
reduce the benefit even when both models fit. We leave adaptive processor
placement and parallel scheduling to future systems work. Such scheduling
could account for thermal throttling, competing applications, and the
relative execution costs of the target and drafter when assigning
processors and coordinating concurrent work.

Configurations are selected before evaluation and remain fixed for each
deployment. Changes in workload distribution, thermal state, or competing
applications may change which configuration is preferable. The observed
1.80--2.67\% gaps to the evaluated per-window oracle concern a finite
catalog under the measured costs. They establish neither global optimality
nor adaptation to changing operating conditions. Re-profiling the device
and revisiting its configuration are possible extensions for these changing
conditions.

\subsection{Response latency and preparation costs}

The objective is faster response generation after audio prefill.
Preparing audio before a request moves that work off the response path but
still consumes computation, memory, and energy. The Xiaomi 14 energy
measurement covers one prepared-cache request, excluding earlier audio
preparation and external cooling. The evaluation does not characterize
energy over long listening sessions, the cost of retaining
audio states for many possible requests, or battery-life tradeoffs.
Decoding gains therefore do not by themselves establish lower total energy
or end-to-end latency when audio preparation is unfinished. These costs
are relevant to a complete always-available audio interface.

\subsection{Evaluation scope}

The corpus-wide comparisons use phone-profile replay. ASR uses 688 full-audio
test windows and 94,511 target tokens per phone on four separately calibrated
profiles. Replay accounts for producer work, verification, correction,
queueing, and drain, but excludes model loading and audio/prompt prefill.
Translation and spoken QA use separate Omni-7B replay with a shared Omni-3B
drafter and zero token-conversion cost. Their cohort-selection rules differ,
and the profiles do not separately calibrate queue/cache-control costs or contention.
Appendix~\ref{sec:omni-qa-subgroup} specifies the post-hoc long-answer QA
selection rule and its limits.

Native cases support the mechanism at their stated inputs and repetition
counts. The native on-demand sweeps cover two phones, Redmi K70 Pro and Xiaomi 14,
using a different 10-minute audio trace on each phone, with one formal
request per checkpoint for each method.
The 10-minute continuous-input
study measures lag and temperature for a separate Omni-7B configuration
with audio-token reduction. Its cooled run keeps pace with arriving audio,
while the uncooled case develops backlog. The cases differ in initial thermal
state and source lifetime as well as cooling. Their input is a paced
recording with 12s context resets. They do not establish corpus-wide native
gains, recognition quality under pruning, or stability beyond the measured
duration. Appendix~\ref{sec:appendix-eval-rq3} gives both native protocols.

\subsection{Numerical sensitivity across execution backends}
\label{sec:numerical-sensitivity}

Finite-precision target execution can vary with the backend and verification
shape even under greedy decoding. \citet{yuan2025numerical} show that changing
batch size, hardware configuration, and arithmetic precision can alter LLM
outputs through rounding and floating-point non-associativity.
\citet{he2025nondeterminism} further distinguish repeatability at a fixed
configuration from batch invariance: matrix multiplication and attention may
use different reduction schedules as the number of processed tokens changes.
Small logit perturbations can therefore reverse the ranking of nearly tied
tokens without stochastic sampling.

Native checks retain all observed output discrepancies. The K70 ASP
calibration emits an additional comma relative to scalar decoding.
Five of 40 Note 9 Pro resident-source calibration episodes fail exact output
matching. The native on-demand sweep also contains output differences from
Target-only. Appendix~\ref{sec:appendix-evaluation} retains these checks.
These observations do not
identify the cause of each discrepancy. They preclude an exact-output
claim for the affected runs, even when their operation timings remain
useful for cost calibration.

Replay cannot validate partial-KV repair. A real target may retain a
numerically perturbed state even when observed greedy tokens agree.
A state-equivalence claim requires validated accepted-prefix state
restoration, with any repair cost included in the measured execution.
Proposition~\ref{prop:policy-exactness} is conditional on reference-equivalent
execution and state restoration. It does not assert bitwise reproducibility
across distinct numerical implementations. Target checking and repair still
govern which speculative tokens may be committed.

\end{document}